\documentclass[pdflatex,sn-mathphys-num]{sn-jnl}

\usepackage{graphicx}%
\usepackage{multirow}%
\usepackage{amsmath,amssymb,amsfonts}%
\usepackage{amsthm}%
\usepackage{mathrsfs}%
\usepackage[title]{appendix}%
\usepackage{xcolor}%
\usepackage{textcomp}%
\usepackage{manyfoot}%
\usepackage{booktabs}%
\usepackage{algorithm}%
\usepackage{algorithmicx}%
\usepackage{algpseudocode}%
\usepackage{listings}%

\usepackage{multirow}
\usepackage{subfigure}
\usepackage{adjustbox} 
\usepackage{array}
\usepackage{xcolor}
\definecolor{lightsteelblue}{rgb}{0.6901960784313725, 0.7686274509803922, 0.8705882352941177}
\definecolor{cornflowerblue}{rgb}{0.39215686274509803, 0.5843137254901961, 0.9294117647058824}
\definecolor{steelblue}{rgb}{0.27450980392156865, 0.5098039215686274, 0.7058823529411765}
\definecolor{firebrick}{rgb}{10.6980392156862745, 0.13333333333333333, 0.13333333333333333}
\definecolor{lightcoral}{rgb}{0.9411764705882353, 0.5019607843137255, 0.5019607843137255}

\graphicspath{ {./figures/}}

\theoremstyle{thmstyleone}%
\newtheorem{theorem}{Theorem}
\theoremstyle{thmstyletwo}%

\theoremstyle{thmstylethree}%
\newtheorem{definition}{Definition}%

\begin{document}

\title[A systematic comparison of measures for publishing
k-anonymous network data]{A systematic comparison of measures for publishing
$k$-anonymous social network data}


\author*[1,2]{\fnm{Rachel G.} \spfx{de} \sur{Jong}}\email{r.g.de.jong@liacs.leidenuniv.nl}

\author[2, 1]{\fnm{Mark P. J.} \spfx{van der} \sur{Loo}}

\author[1]{\fnm{Frank W.} \sur{Takes}}

\affil*[1]{\orgdiv{LIACS}, \orgname{Universiteit Leiden}, \orgaddress{\street{Einsteinweg 55}, \city{Leiden}, \postcode{2333 CC}, \country{the Netherlands}}}

\affil[2]{\orgname{CBS (Statistics Netherlands)}, \orgaddress{\street{Henri Faasdreef 312}, \city{Den Haag}, \postcode{2492 JP}, \country{the Netherlands}}}

\abstract{
    Sharing or publishing social network data while accounting for privacy of individuals is a difficult task due to the interconnectedness of nodes in networks. 
    A key question in $k$-anonymity, a widely studied notion of privacy, is how to \emph{measure} the anonymity of an individual, as this determines the attacker scenarios one protects against.
    In this paper, we systematically compare the most prominent anonymity measures from the literature in terms of the `completeness' and `reach' of the structural information they take into account. 
    We present a theoretical characterization and a distance-parametrized strictness ordering of the existing measures for $k$-anonymity in networks.
    In addition, we conduct empirical experiments on a wide range of real-world network datasets with up to millions of edges. 
    Our findings reveal that the choice of the measure significantly impacts the measured level of anonymity and hence the effectiveness of the corresponding attacker scenario, the privacy vs. utility trade-off, and computational cost.  
    Surprisingly, we find that the anonymity measure representing the most effective attacker scenario considers a greater node vicinity yet utilizes only limited structural information and therewith minimal computational resources. 
    Overall, the insights provided in this work offer researchers and practitioners practical guidance  for selecting appropriate anonymity measures when sharing or publishing social network data under privacy constraints. 
}


\maketitle

\section{Introduction}\label{sec:intro}
Network science is a research field centered around modeling real-world phenomena using nodes, representing entities, and edges representing some type of meaningful connection between these entities.
Various methods have been introduced to study networks, for example, to find important nodes~\cite{saxena2020centrality}, communities~\cite{leskovec2010empirical} or anomalies~\cite{bhuyan2013network}.
In social networks, these methods can be applied to model disease spread~\cite{azizi2020epidemics}, measure segregation~\cite{kazmina2024socio, bojanowski2014measuring}, or identify suspicious entities in (social) networks~\cite{savage2014anomaly}. 
For such research purposes it is beneficial to have access to networks that contain realistic information on individual entities, for example, about people.
Often data from social media platforms such as Meta/Facebook or X/Twitter are used to construct networks, but also population-scale networks generated from official government register data  where the set of nodes represents the entire population of a country~\cite{bokanyi2023anatomy, van2023whole} are frequent sources for social network analysis research. 

What abovementioned network datasets have in common is that sharing or publicly publishing such large-scale social networks, even if it is for research purposes, comes with considerable privacy risks. 
Even after personally identifiable information is removed and replaced with artificial identifiers, i.e., pseudonymization, a network is not sufficiently protected against disclosure risk, as the structure of the network itself can be used to identify individuals~\cite{backstrom2007wherefore, romanini2021privacy, dejong2023effect}.
As a consequence, real-world social network data is often not shared or can only be accessed under strict conditions.
At the same time, promoted by, e.g., the FAIR data initiative~\cite{wilkinson2016fair}, researchers are often expected to digitally publish social network datasets used together with their paper for reproducibility and reusability~\cite{neal2024recommendations}. 
It is important to note that in this concrete scenario, it is not possible to adopt techniques based on differential privacy. 
Such methods require generating synthetic data at the cost of analytical fidelity, or only work in ``online'' systems that hold the original data and dynamically add noise as the graph is queried~\cite{jiang2021applications, wang2013preserving}. 
Concretely, in this work we ultimately focus on obtaining a publishable file of a network dataset that remains representative of the real system, while protecting individual privacy. 

The problem of sharing social network data in some form while accounting for the privacy of people represented has four important aspects that should be taken into account. 
This includes 1) the type of desired output, 2) the `attacker scenarios' against which should be protected, 
3) the utility of the network data in terms of a) which structural properties should be preserved and b) performance on various downstream tasks for which the network will be used, such as detecting communities, and 4) computational cost.
The most commonly used methods for the protection of network data have initially been introduced in the context of Statistical Disclosure Control (SDC)~\cite{hundepool2012statistical, willenborg2012elements}, aiming to protect tabular data, i.e., data where each entity is described by a number of attributes.
In this field one distinguishes between \emph{key attributes}, \emph{quasi-identifiers} and \emph{sensitive attributes}.
Key attributes, such as name, personal identifier or phone number,  can uniquely identify a person and must be removed from anonymized data.
Quasi-identifiers are attributes which may be known to the attacker, who aims to de-anonymize the network, and could possibly be used for identification.
Sensitive attributes, such as income or illness, are unknown and of interest to the attacker.
Methods for SDC aim to protect against identity disclosure, i.e., finding the person in the tabular dataset, or attribute disclosure, i.e., deriving sensitive attributes of individuals.
The utility of the resulting datasets is then measured in terms of general utility, the distribution of attribute values in the data, hence describing the
data itself, and specific utility, concerning results of analyzing the data~\cite{snoke2018general}. 

The most commonly used approaches for anonymizing tabular data are differential privacy~\cite{dwork_differential_2008, dwork2006calibrating, dwork2006differential} and $k$-anonymity~\cite{sweeney2002k}.
In differential privacy, a mechanism gives noisy answers to user queries such that the privacy of entities is preserved.
For $k$-anonymity, the aim is to create a version of the data such that for each entity there are at least $k-1$ equivalent candidates so that any individual can be identified with a probability of at most $1/k$, thereby protecting against identity disclosure.
Typically, a value of $k>1$ is used.
$k$-Anonymity is then achieved by a combination of generalization by making information on, for example, age or zip code less specific, and suppressing values in the dataset.
However, in a $k$-anonymous dataset it could still be possible for an attacker to obtain information if all $k$ entities have the same value for a sensitive attribute or if the distribution in a certain class deviates from the distribution over the entire dataset.
Hence, $k$-anonymity in itself does not sufficiently protect against attribute disclosure.
This is counteracted by extending $k$-anonymity with $\ell$-diversity~\cite{machanavajjhala_l-diversity_2007}, $t$-closeness~\cite{li_t-closeness_2007} or ($\alpha$, $k$)-anonymity~\cite{wong2006alpha}.
Other approaches incorporate synthetic data~\cite{drechsler2011synthetic} by either generating a new dataset from scratch or combining synthetic data with the actual data.
Yet another approach is presented in~\cite{xiao2006anatomy}, in which the quasi-identifiers and sensitive attributes are shared in separate tables such that it can not be inferred which set of sensitive attributes corresponds to which set of quasi-identifiers.

Now focusing specifically on network data, based on the methods in SDC several generically applicable approaches have been introduced based on differential privacy~\cite{sala2011sharing, proserpio2012calibrating, wang2013preserving} as well as $k$-anonymity~\cite{liu2008towards, hay2008resisting, romanini2021privacy, zou2009k}.
Additionally, other network-specific anonymization techniques have been introduced in the literature such as randomization~\cite{ying2008randomizing, liu2016smartwalk, mittal2012preserving} and clustering~\cite{campan2008data, hay2008resisting, bhagat2009class, liu2016linkmirage, yazdanjue2020evolutionary, ford2009p}.
These methods ensure that a possible attacker is unable to either identify nodes by ensuring \emph{node privacy}, or the fact that there is a connection between two entities, i.e., \emph{edge privacy}.

The approach we focus on in this paper is $k$-anonymity, as this ultimately enables one to publish or share an anonymized version of the network data that can be used for any type of analysis.
For $k$-anonymity, one needs to choose both an \emph{anonymity measure}, and \emph{anonymization algorithm}.
The anonymity measure, which determines when two nodes are equivalent, is used to assess whether the graph satisfies $k$-anonymity.
A network is $k$-anonymous when each node is equivalent to at least $k-1$ other nodes.
The measure essentially determines the attacker scenario against which a $k$-anonymous network is protected.
At first glance, it seems desirable to choose a strict measure, as this protects against more attacker scenarios.
However, the stricter the measure, the more properties should be satisfied for two nodes to be equivalent, which means more data utility has to be sacrificed to make nodes anonymous.
However, if a too lenient measure is used, hence protecting against a weaker attacker scenario, one takes the risk of underestimating the attacker knowledge and the attacker may still be able to identify entities in the shared network and obtain sensitive information.
Usually, real-world networks do not satisfy $k$-anonymity and an anonymization algorithm should be used to perturb the graph and achieve a $k$-anonymous network which can safely be published.
This makes the $k$-anonymity approach very flexible.
Since each anonymity measure essentially models a specific level of attacker knowledge, it allows one to choose the attacker scenario to protect against accordingly.
The anonymization algorithm can be designed to preserve certain network properties considered relevant by the user, such as clustering or centrality of nodes~\cite{wang2014high}.
Utility of the resulting anonymized network can be measured based on how well structural network properties are preserved, such as the clustering coefficient, and performance on various downstream tasks such as the community detection, or finding the most central nodes. 
While both the anonymity measure and algorithm play an important role, in this paper we focus on the effect of using different anonymization measures.

It is clear that the choice of anonymity measure is of high importance.
Although many existing works aim to achieve a $k$-anonymous network for various measures~\cite{liu2008towards, hay2008resisting, romanini2021privacy, zou2009k} and various surveys have been published on network anonymization in general~\cite{ji2016graph,zhou2008brief, casas2017survey,
jiang2021applications, beigi2020survey, yazdanjue2025comprehensive}, this paper offers the first systematic comparison of measures for $k$-anonymity. 
As a key part of this paper, we give an overview of different measures, and compare them both theoretically, and empirically on a wide range of empirical networks.
We characterize the anonymity measures based on how far they `reach', i.e., up to where they take structural information into account, and how complete the structural information is and order the measures using a formalized notion of \emph{strictness}.
We empirically compare the measures based on the effectiveness of the attacker scenario considered, the anonymity vs. utility trade-off and runtime.
We find that measures with a larger reach, even with minimal structural information, perceive a higher risk of disclosure compared to measures with perfect structural knowledge of the direct neighborhood of a node.

To summarize, the main contributions of this paper are the following.
\begin{itemize}
    \item We systematically summarize and categorize existing measures for $k$-anonymity in networks (Section~\ref{sec:kanonmeas}).
    \item We theoretically compare the measures based on reach and the completeness of the structural information they take into account, and derive a formal strictness-based ordering of measures (Section~\ref{sec:theory}).
    \item We empirically compare these measures in terms of 1) measured anonymity, 2) the achieved trade-off between anonymity and utility during anonymization, and 3) runtime on a wide range of real-world network datasets (Section~\ref{sec:emp_comparing}).
\end{itemize}

The remainder of the paper is structured as follows.
In Section~\ref{sec:related}, we give an overview of methods introduced in literature to share network data in a privacy aware manner.
Then, in Section~\ref{sec:kanon}, we provide preliminary concepts and notation, followed by an overview of the considered measures for $k$-anonymity in Section~\ref{sec:kanonmeas}.
We compare a diverse subset of representative measures theoretically in Section~\ref{sec:theory}, after which we compare these measures empirically on a wide range of network datasets in terms of measured anonymity, utility and runtime in Section~\ref{sec:emp_comparing}.
Lastly, we conclude the paper and suggest possible future work in Section~\ref{sec:conc}.

\section{Related work}\label{sec:related}
In this section, we summarize approaches introduced in literature to share network data while accounting for privacy of individuals represented.
We categorize the approaches based on the type of \emph{output} they produce, being: 1) an interactive mechanism, 2) synthetic data, 3) an intermediate representation, and 4) a perturbed version of the data, which includes $k$-anonymity.
We discuss methods in each of these four categories, focusing on the type of output, which attacker scenario one protects against and how utility comes into play.
A summary of the methods, corresponding category and type of privacy preserved can be found in Table~\ref{tab:categories}. 

    \begin{table}[ht]
        \caption{Categorization of methods for privacy sensitive sharing of networks. $\checkmark$ indicates the method is in this category, $\lozenge$ that it could be in this category, but is less commonly studied as such.}
    \label{tab:categories}
    \centering
    \scriptsize
    \setlength{\tabcolsep}{2pt}
        \begin{tabular}{r|cccc|cc}
        \hline
        \multicolumn{1}{c|}{Method}           & \multicolumn{4}{c|}{Output type}                          & \multicolumn{2}{c}{Privacy for} \\
        \multicolumn{1}{l|}{}           & Interactive  & Synthetic    & Intermediate & Perturbed    & Nodes          & Edges          \\ \hline
        Differential privacy: queries~\cite{proserpio2012calibrating, hay2009boosting, macwan2018node}   & $\checkmark$ &              &              &              & $\lozenge$    & $\checkmark$  \\
        Differential privacy: synthetic~\cite{sala2011sharing, wang2013preserving, yang2020secure} &              & $\checkmark$ &              &              & $\lozenge$    & $\checkmark$  \\
        Clustering~\cite{campan2008data, bhagat2009class, liu2016linkmirage}                       &              &              & $\checkmark$ &              &   $\checkmark$            &   \\
        Injecting uncertainty~\cite{boldi2012injecting}                     &              &              &              $\checkmark$& &               & $\checkmark$  \\
        Randomization~\cite{ying2008randomizing, liu2016smartwalk, mittal2012preserving}           &              &              & &              $\checkmark$&               & $\checkmark$  \\
        $k$-Anonymity~\cite{liu2008towards, hay2008resisting, romanini2021privacy, zou2009k}                     &              &              &              & $\checkmark$ & $\checkmark$  &               \\
         \hline
        \end{tabular}
    \end{table}

\subsection{Interactive mechanism}\label{sub:int}
We first discuss interactive approaches which consists of differential privacy: queries. Approaches based on differential privacy involving synthetic data are discussed in Section~\ref{sub:synt}.
Differential privacy allows users to ask queries about the network dataset~\cite{dwork2006differential}.
The user receives an answer to which noise is added, such that it preserves either node or edge privacy.
The scenario assumed is that the attacker has knowledge about the entire network except for one node or edge.
Based on the answer, it should not be possible to infer whether a specific edge exists in the dataset, when protecting edge privacy, or whether a specific node exists, when protecting node privacy.
The amount of noise added is based on the sensitivity of the posed query and ensures privacy is guaranteed.
Most works on differential privacy focus on edge privacy~\cite{wang2013preserving}, and some on node privacy~\cite{macwan2018node, jian2021publishing}.
As the (non)existence of a node can potentially have more effect on the query than the (non)existence of an edge, it is overall more difficult to ensure node privacy.
Other works aim to answer a specific user query with as little noise as possible while still achieving the privacy guarantee~\cite{proserpio2012calibrating, hay2009boosting, macwan2018node, jiang2021applications}.

A disadvantage of this approach is that a user can only obtain answers to queries that have been implemented, which limits the possible analysis that can be done. 
Furthermore, when more queries are posed, more noise should be added to ensure differential privacy~\cite{jiang2021applications}.
This strongly reduces the quality of the answers and with that the utility of the resulting dataset.

\subsection{Synthetic data}\label{sub:synt}
The second category of approaches generates a synthetic network based on structural properties of the original network.
Various models exist to generate networks that capture certain (but not all) real-world network properties.
Examples are the Barabási–Albert model~\cite{barabasi1999emergence}, which captures the powerlaw degree distribution, or the Watts-Strogatz model~\cite{watts1998collective} which captures the small world property.
An example of such a model specifically for achieving privacy, is the dK-graph model, for which the degree of information taken into account can be set with parameter $d$~\cite{mahadevan2006systematic}.
For $d=1$, the model uses the degree distribution, and for $d=2$ the joint degree distribution, which incorporates for each edge the degree of the two nodes it connects and how frequently this combination of degrees occurs. 
However, it was found that using this model alone may not guarantee privacy as many nodes can still be identified in the resulting dK-graph~\cite{horawalavithana2019privacy}.

Work using differential privacy, generates differentially private output which is used to generate synthetic networks.
Examples are the dK-graph which uses the joint degree distribution~\cite{sala2011sharing, wang2013preserving}, or the exponential random graph model~\cite{lu2014exponential}. 
For this approach, the quality of the resulting network depends both on the quality of the answer provided by the differential privacy mechanism and the graph model used~\cite{jiang2021applications}.
Additionally, with the increasing use of differential privacy in deep learning~\cite{demelius2025recent}, recent work proposed a generative deep learning model to generate a network while ensuring edge privacy~\cite{yang2020secure}.

\subsection{Intermediate representation}
The third category we discuss outputs an intermediate representation of the original network.
This intermediate representation can be directly analyzed, or a network can be sampled from the representation.
Two types of intermediate representations have been introduced for the purpose of sharing privacy sensitive network data, being clustering~\cite{campan2008data, hay2008resisting, bhagat2009class, liu2016linkmirage,  yazdanjue2020evolutionary, ford2009p} and injecting uncertainty~\cite{boldi2012injecting}.

For the clustering approaches the nodes of the network are merged into super nodes~\cite{campan2008data, hay2008resisting, bhagat2009class, liu2016linkmirage,  yazdanjue2020evolutionary, ford2009p} according to some predefined mechanism.
These supernodes additionally contain labels denoting the number of nodes and edges within each cluster.
If there is at least one edge between nodes in two clusters, this is modeled using a superedge denoting the number of connections between nodes in the different clusters.
Various works focus on improving the clusters made~\cite{ yazdanjue2020evolutionary}, and some approaches add additional constraints to account for node labels~\cite{ford2009p}.

By ensuring that each cluster has a size of at least a preset number of nodes, anonymity of the entities can be ensured.
The clustering process can to some extent preserve global properties, yet local properties may be destroyed as the exact neighborhood structure surrounding a node can not be inferred from the supernodes and superedges.

The second approach in this category is to add uncertainty by assigning a probability to each edge in the network to hide certain node properties~\cite{boldi2012injecting}.
The resulting graph with edge probabilities can be published and used to sample new networks.
This approach would hence ensure edge anonymity.
Which network properties would be preserved and destroyed depends on how the probability is assigned to the edges.

\subsection{Perturbed network}\label{sub:perturbed}
For the last category of approaches, the goal is to share an adapted version of the original network such that privacy of entities represented is ensured. 
Overall there are two types of approaches: randomization~\cite{ying2008randomizing, liu2016smartwalk, mittal2012preserving} to ensure edge privacy, and $k$-anonymity~\cite{liu2008towards, hay2008resisting, romanini2021privacy, zou2009k} to ensure node privacy.

To protect edge privacy, operations on edges including edge addition, deletion or rewiring, can be used to perturb the network.
After the graph is perturbed, an attacker can not be certain about the existence of a connection between two entities, hence ensuring edge anonymity.
The work in~\cite{ying2008randomizing} introduces such an approach, additionally preserving utility by retaining the spectrum, i.e., the set of eigenvalues of the adjacency matrix representing the network.
In~\cite{liu2016smartwalk, mittal2012preserving}, edges are added based on random walks with the aim of preserving more data utility.

One of the most common methods for ensuring node privacy, which is also the focus of this work, is $k$-anonymity~\cite{liu2008towards, hay2008resisting, romanini2021privacy, zou2009k}.
An anonymity measure is used to assess whether a graph satisfies $k$-anonymity, which is true when for each node in the graph there are at least $k-1$ equivalent candidates given an amount of structural information indicated by the anonymity measure.
The chosen anonymity measure corresponds to an attacker scenario where we assume that a possible attacker has this amount of knowledge.
If the network is not $k$-anonymous an anonymization algorithm can be applied.
These algorithms can be designed to preserve utility, i.e., certain structural properties or performance on downstream tasks~\cite{wang2014high}.
Inevitably, some properties need to be destroyed in order to make all nodes $k$-anonymous.

In the remainder of this paper, we focus on $k$-anonymity as this approach allows one to safely publish or share an altered version of the network.

\section{Preliminaries}\label{sec:kanon}
In this section, we define the required concepts on networks and $k$-anonymity.

\subsection{Networks}\label{sub:prelnetworks}
We define a network or graph $G = (V, E)$ as a set of nodes $V$, and a set of edges $\{v, w\} \in E$, with $v, w \in V$.
The degree of a node equals the number of connections it has: $degree(v) = |\{w : \{v, w\} \in E \}|$.
We can create a distribution of node degrees in the network.
In real-world networks, this often resembles a powerlaw distribution~\cite{barabasi2004network} where degrees are distributed as $c *degree^{-\alpha}$.
Here $c$ is a constant and $\alpha$ the powerlaw exponent.
The value of $\alpha$ indicates the steepness of the slope of the degree distribution.
Hence, a lower value of $\alpha$ indicates a heavier tail and greater imbalance between low- and high-degree nodes, whereas a higher $\alpha$ results in a steeper slope, implying a more balanced degree distribution with fewer high-degree nodes.

In real-world networks it often occurs that two neighboring nodes form a triangle with a third node.
We measure the tendency to form triangles using the node \emph{clustering coefficient}.
For a node, this equals the number of triangles it is part of, divided by the maximum number of triangles it could be part of, which equals $\frac{1}{2}degree(v)(degree(v) -1)$.
For a graph, we summarize the clustering coefficient as the average clustering coefficient over all nodes.
Nodes may connect to nodes with a similar degree or a different degree.
We capture this tendency by \emph{assortativity}.
A negative value means nodes tend to connect to dissimilar nodes while a positive value indicates they tend to connect to similar nodes.

We define the distance between two nodes, $distance(v, w)$, as the minimum number of edges that needs to be traversed to reach one node from the other.
Since the edges are undirected, it follows that $distance(v, w) = distance(w, v)$ and $distance(v, v) = 0$.
When there is no path between two given nodes, i.e., a sequence of edges connecting $v$ and $w$, then $distance(v, w) = \infty$.
This occurs when the nodes are in different \emph{components}, i.e., maximal subsets of nodes in which there is a path between all pairs of nodes.
The diameter of the graph, $D(G)$, is the largest shortest path length between two nodes that is not equal to $\infty$.
The average distance in a graph equals the average of the length of all shortest paths between all node pairs in the network that are in the same component.

We define the $d$-neighborhood of a node $N_d(v) = (V_{N_d(v)}, E_{N_d(v)})$ as the set of nodes that are at most distance $d$ from the considered node $v$, and the set of all edges between these nodes.
When two neighborhoods in a graph are structurally indistinguishable they are isomorphic. More precisely, two graphs $G = (V, E)$ and $G' = (V', E')$ are isomorphic if there exists a bijective function $\phi: V \rightarrow V'$ such that for each $v, w \in V$ it holds that $\{\phi(v), \phi(w) \} \in E'$ precisely when $\{v, w\} \in E$.
We can determine if two graphs are isomorphic by comparing their \emph{canonical labeling}~\cite{nauty}.
This is a label assigned by a function $\mathcal{C}$ such that two graphs have the same label value only if they are isomorphic.
A special case of isomorphism is automorphism: an isomorphism from the graph onto itself.
If two nodes can be mapped onto each other by an automorphism, they are in the same \emph{orbit}.
This implies that the nodes are structurally indistinguishable from each other.

To assess the utility of the anonymized networks in Section~\ref{sub:utility}, we use three metrics that are commonly used in downstream network analysis tasks.
Most of the nodes in real-world networks are in the largest connected component, also referred to as the giant component.
The first metric, robustness, measures the fraction of nodes remaining in the largest connected component as the grpah is being perturbed during anonymization.
The second utility metric considers the task of community detection and therewith looks at the meso-level structure of the graph.
Nodes in networks often tend to form communities which can be found by using community detection algorithms~\cite{traag2019louvain}.
Since randomization is involved in these algorithms, more stable communities can be found by adding a consensus clustering step~\cite{lancichinetti2012consensus}.
To determine how well the community structure is preserved, we compute the normalized mutual information (NMI)~\cite{lancichinetti2012consensus} between the communities found before and after anonymization.
Third, nodes can have different roles in the network and some nodes have a more central position than others.
This can be measured by centrality measures such as betweenness centrality, which measures the fraction of shortest paths going through a given node~\cite{brandes2001faster}.
To determine how well the most central nodes are preserved, we compute the overlap in the top 100 most central nodes before and after anonymization.

\subsection{$k$-Anonymity and equivalence}
In order for a node to be $k$-anonymous, it should be equivalent to at least $k-1$ other nodes according to a particular anonymity measure $M$ (discussed in Section~\ref{sec:kanonmeas}).
Any equivalence measure $M$ partitions the set of nodes into a set of equivalence classes $P_M$.
If two nodes $v, w \in V$ are equivalent using a measure $M$, we denote this as $v \cong_M w$.
$P_M(v)$ denotes the equivalence class of a given node $v$.
A node $v$ is $k$-anonymous when $|P_M(v)|\geq k$. 
We summarize the anonymity of a graph given a certain measure by its \emph{uniqueness}, a commonly used measure~\cite{romanini2021privacy} which is determined using the formula in Equation~\ref{eq:uniqueness}.
It is essentially the inverse of anonymity and equals the fraction of nodes for which there is no other equivalent node, i.e., $k=1$.

\begin{equation}
        U_M(G) = \frac{|\{v : v\in V, |P_M(v)| = 1\}|}{|V|} \label{eq:uniqueness}
\end{equation}    

\section{$k$-Anonymity measures}\label{sec:kanonmeas}
In this section, we focus on $k$-anonymity, as previously discussed in Section~\ref{sub:perturbed}, and give an overview of measures for $k$-anonymity introduced in literature. 
We categorize these based on what type of information they take into account: 1) degree based, 2) neighborhood based, 3) automorphism based and 4) hybrid measures.
For each measure we describe both the measure itself and, even though the anonymization process itself is beyond the scope of this work, give an estimate of how complex it is to, given this measure, achieve anonymity with as few alterations as possible.
A parameterized version for most measures can be created to account for structural information beyond the direct neighborhood of the considered node, denoted by parameter $d$.
Additionally, we discuss how the measures can be extended by a cascading step. 
For completeness, we end the overview by summarizing measures that do not fit in the four categories.
Note that while various works have introduced measures for network extensions such as node labels~\cite{zhou2011k, tripathy2012algorithm, ren2022personalized, yuan2010personalized}, edge labels~\cite{hao2014k} or edge weights~\cite{liu2015k}, in the remainder of this work we focus on undirected unlabeled networks.

Table~\ref{tab:meas_overview} summarizes the measures we focus on in the remainder of this paper.
These measures are chosen based on two criteria.
First, each of the measures assigns one or more values to each node, which are used to divide the set of all nodes into equivalence classes.
Second, the chosen measures model feasible attacker scenarios.

\begin{table}[ht]
\caption{The six anonymity measures compared in this paper. For $d=1$ and $d=2$, we illustrate the reach of the measures (second and fourth column) and the outcome for the given measure (third and fifth column).
}
\label{tab:meas_overview}
\begin{tabular}{@{}rcccc@{}}
\cmidrule(l){2-5}
                        & \multicolumn{2}{c}{$d=1$}                                                                                                                                                                                                                                               & \multicolumn{2}{c}{$d=2$}                                                                                                                                                                                                                                   \\
                        & Reach                                                                                                                                      & Value                                                                                                                      & Reach                                                                                                                     & Value                                                                                                                           \\  \cmidrule(l){1-5} 
Degree ~\cite{liu2008towards, macwan2017k, casas2013algorithm, lu2012fast}                  & \begin{minipage}{.15\textwidth}      \includegraphics[width=\linewidth]{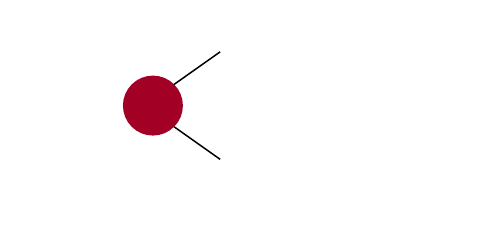}    \end{minipage}                  & 2                                                                                                                          & -                                                                                                                         & -                                                                                                                               \\ \cmidrule(l){1-5}
Count~\cite{dejong2023algorithms}                   & \begin{minipage}{.15\textwidth}      \includegraphics[width=\linewidth]{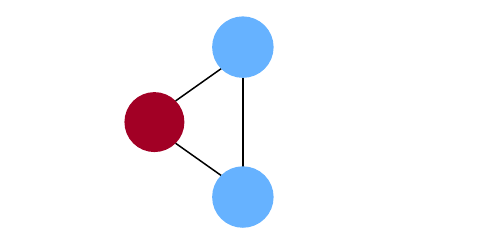}    \end{minipage}                  & (3, 3)                                                                                                                     & \begin{minipage}{.15\textwidth}      \includegraphics[width=\linewidth]{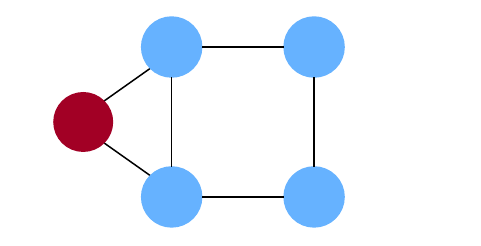}    \end{minipage} & (5, 6)                                                                                                                          \\ \cmidrule(l){1-5}
Degree distribution~\cite{dejong2023algorithms}     & \begin{minipage}{.15\textwidth}      \includegraphics[width=\linewidth]{figures/measures/table/d12.pdf}    \end{minipage}                  & $\{2, 2, 2\}$                                                                                                              & \begin{minipage}{.15\textwidth}      \includegraphics[width=\linewidth]{figures/measures/table/d21.pdf}    \end{minipage} & $\{2, 2, 2, 3, 3\}$                                                                                                             \\ \cmidrule(l){1-5}
$d$-$k$-Anonymity~\cite{romanini2021privacy, zhou2008preserving, dejong2023effect, alavi2019attacker}       & \begin{minipage}{.15\textwidth}      \includegraphics[width=\linewidth]{figures/measures/table/d12.pdf}    \end{minipage}                  & \begin{minipage}{.15\textwidth}      \includegraphics[width=\linewidth]{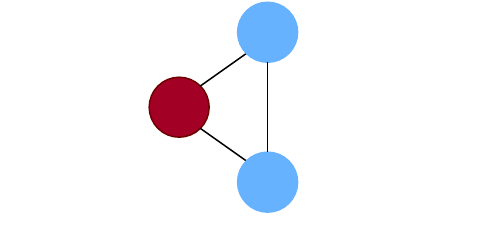}    \end{minipage} & \begin{minipage}{.15\textwidth}      \includegraphics[width=\linewidth]{figures/measures/table/d21.pdf}    \end{minipage} & \begin{minipage}{.15\textwidth}      \includegraphics[width=\linewidth]{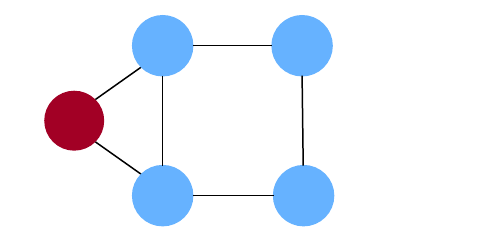}    \end{minipage}      \\ \cmidrule(l){1-5}
VRQ~\cite{hay2007anonymizing, hay2008resisting}                     & \begin{minipage}{.15\textwidth}      \includegraphics[width=\linewidth]{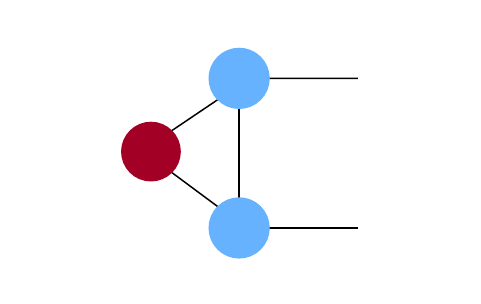}    \end{minipage}                  & $\{2, 3, 3\}$                                                                                                              & \begin{minipage}{.15\textwidth}      \includegraphics[width=\linewidth]{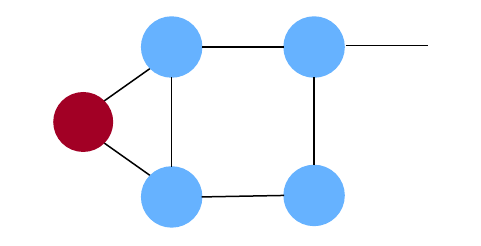}    \end{minipage} & $\{2, 2, 3, 3, 3\}$                                                                                                             \\ \cmidrule(l){1-5}
\multirow{2}{*}{Hybrid~\cite{wang2013outsourcing}} 
    & \multirow{2}{*}{\begin{minipage}{.15\textwidth}      \includegraphics[width=\linewidth]{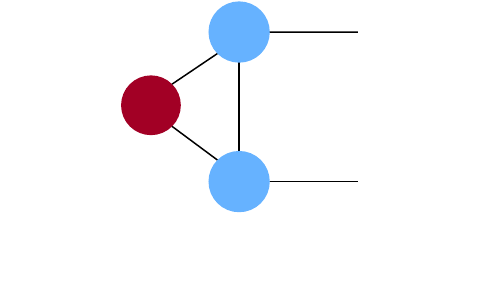}    \end{minipage}} 
    & \begin{minipage}{.15\textwidth}      \includegraphics[width=\linewidth]{figures/measures/table/d1-4.pdf}    \end{minipage} 
    & \multirow{2}{*}{\begin{minipage}{.15\textwidth}      \includegraphics[width=\linewidth]{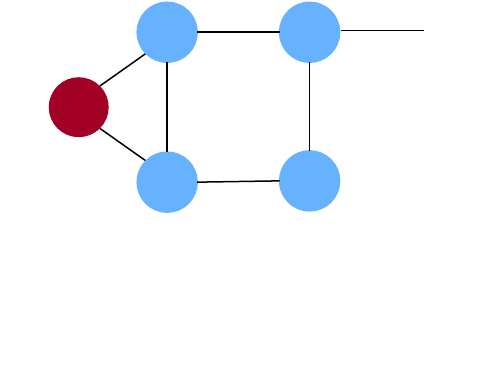}    \end{minipage}} 
    & \begin{minipage}{.15\textwidth}      \includegraphics[width=\linewidth]{figures/measures/table/d2-3.pdf}         \end{minipage} \\
                        &                                                                                                                                            & \multicolumn{1}{c}{$\{2, 3, 3\}$}                                                                                     & \multicolumn{1}{c}{}                                                                                                      & \multicolumn{1}{c}{$\{2, 2, 3, 3, 3\}$}                                                                                    \\ \cmidrule(l){1-5} 
\end{tabular}
\end{table}

\subsection{Degree based}
\label{sec:degreebased}
We distinguish between three measures within the category of degree based measures, being degree, vertex refinement queries and a measure based on joint degree.

Most of the works using degree based measures focus on the degree of the node~\cite{liu2008towards, hao2024mlda, rajabzadeh2020graph, zhang2019large, mohapatra2019graph, macwan2017k, casas2013algorithm, lu2012fast}.
A graph is then $k$-anonymous if each degree occurs at least $k$ times in the network.
Compared to the other measures we discuss, this is the simplest measure, as it accounts for the least structural information and is trivial to compute for a given node.
As a result, anonymization algorithms have been introduced that, given this measure, use the minimum number of edge deletions and additions to obtain an anonymous graph. These approaches can be used on networks with over a million nodes and edges~\cite{liu2008towards, casas2013algorithm, lu2012fast}.
Several works specifically aim to preserve certain network properties to retain utility~\cite{rajabzadeh2020graph, mohapatra2019graph, macwan2017k}.
The work in~\cite{gao2018resisting} builds upon the notion of $k$-degree anonymity and assumes there might be uncertainty in the knowledge of the attacker.
This is accounted for by a so-called binning approach in which nodes with a similar degree, which are in the same bin, are also equivalent.
Since this approach accounts for uncertainty, it can merge equivalence classes that are distinct according to the degree measure, hence it is as most as strict as this measure.

In ~\cite{hay2007anonymizing, hay2008resisting}, a parameterized measure is introduced that takes into account the degree of nearby nodes is introduced.
This work introduces the notion of \emph{vertex refinement queries} (VRQ), denoted $\mathcal{H}_i$, for which the distance can be set with parameter $i$.
When $\mathcal{H}_1$ is used, only the degree of the node is taken into account.
For $\mathcal{H}_i$, with $i>1$, the degree distribution of the nodes at distance at most $i-1$ is taken into account.
The measure is applied recursively, i.e., $\mathcal{H}_i$ is used to compute $\mathcal{H}_{i+1}$ for each $i\geq 1$.
This measure is more difficult to anonymize for and requires an approximation algorithm to minimize the number of edge deletions to arrive at an anonymous graph.

A measure that is more strict than degree, but not as strict as VRQ, has been introduced in~\cite{tai2011privacy} and accounts for degrees of all node pairs connected by an edge.
In this case, a node can be identified if at least one of its edges has a unique combination of endpoint degrees.
While this measure accounts for only slightly more information than node degree itself, anonymization is already much more difficult.
To anonymize a network given this measure, the work proposes an integer programming problem and an approximation algorithm to anonymize larger networks.
This measure differs from the remaining measures as it does not compute a value for a node based on which the equivalence classes are determined.
Hence we do not include the measure in our comparison.
However, it does show conceptual similarities to the cascading algorithm we discuss in Section~\ref{sub:anon-cascade} and~\ref{sec:casc}.

\subsection{Neighborhood based}
\label{sec:neighbbased}
For neighborhood based measures, we distinguish between measures that account for complete structural information and measures that are based on meaningful properties of the neighborhood.
The most commonly used neighborhood based measure accounts for the complete structure of the 1-neighborhood~\cite{romanini2021privacy, zhou2008preserving, zhou2011k, tripathy2012algorithm, ren2022personalized, alavi2019attacker}.
Here nodes are said to be equivalent if their neighborhoods are isomorphic, which implies the neighborhoods are not distinguishable based on network structure.
Other works propose a parameterized version of this measure assuming structural knowledge of the $d$-neighborhoods~\cite{alavi2019attacker, ren2022personalized, dejong2023algorithms, dejong2023effect}, resulting in a measure referred to as $d$-$k$-anonymity.
Similarly to \textsc{VRQ}, this measure is computed recursively.
Hence, in order for a node to be $d$-$k$-anonymous, with $d \geq 2$, it needs to be $(d-1)$-$k$-anonymous. 

The work in~\cite{dejong2023algorithms} uses heuristics to speed up the computation of $d$-$k$-anonymity.
These heuristics are graph invariants: properties that are equal if two graphs, or in this case $d$-neighborhoods, are isomorphic.
If the values are not equal, the $d$-neighborhoods can not be isomorphic and as a result the nodes can not be equivalent.
However, the converse does not hold: if the graph invariants are equal it does not imply that the neighborhoods are isomorphic.
The heuristics used are \textsc{count}, which measures the number of nodes and edges in the $d$-neighborhood and \textsc{degrees}, which we will refer to as \textsc{degdist}, which measures the degree distribution of the $d$-neighborhood.
As these heuristics in themselves can be used as measures for equivalence and are in terms of what information they measure elegantly positioned between degree and $d$-$k$-anonymity, we choose to include these in our comparison.

\subsection{Hybrid}
The previously mentioned measures can also be combined, which is done in~\cite{wang2013outsourcing}, where the aim is to protect against the so-called $1^*$-neighborhood attack.
For the corresponding measure, nodes are equivalent if they have isomorphic 1-neighborhoods and if the degree distribution over the neighboring nodes is equal.
Hence, this combines the measures of 1-neighborhood isomorphism~\cite{romanini2021privacy, zhou2008preserving, dejong2023algorithms, dejong2023effect, ren2022personalized}, and \textsc{VRQ}~\cite{hay2007anonymizing, hay2008resisting}.
A parameterized version would combine the notion of $d$-$k$-anonymity with \textsc{VRQ}$(d)$.

\subsection{Automorphism based}
Several works focus on a very strict scenario in which an attacker should not be able to distinguish between two nodes, even with complete structural information~\cite{zou2009k, wu2010k, cheng2010k}.
There are two variants. 
In~\cite{zou2009k, wu2010k}, the aim is that for each node there should be at least $k-1$ nodes in the same orbit.
This implies that the nodes have the exact same structural properties and hence can not be distinguished even when an attacker has perfect structural information of the entire network.
In~\cite{cheng2010k}, the network is $k$-anonymous if it can be partitioned into $k$ subgraphs that are isomorphic to each other.
To achieve $k$-anonymity for these measures, symmetry should be introduced to the network.
In particular, for a network to be $k$-anonymous, all its components should be symmetric in at least $k-1$ points, which is not realistic for most real-world network datasets.
Anonymizing for this measure would therefore have major impact on the utility of an anonymized version of the network.
As this approach is computationally expensive and approximated by the distance-parameterized structural neighborhood measure $d$-$k$-anonymity mentioned in Section~\ref{sec:neighbbased}, we do not include automorphism based measures in our comparison.

\subsection{Anonymity-cascade}\label{sub:anon-cascade}

Recent work introduces anonymity-cascade~\cite{dejong2023effect}.
This algorithm models the scenario where an attacker reuses the nodes that are uniquely identified to identify more nodes in the network.
The cascading algorithm starts with all nodes that have a unique $1$-neighborhood structure.
Then, it continues in a cascading fashion by identifying neighboring nodes with a unique structure.
This way, the reach of the measure is extended beyond the direct neighborhood.
This is similar to the attacker scenarios that include a propagation step~\cite{narayanan2009anonymizing, fu2015effective}. 

\subsection{Other measures}
For completeness sake, we end with approaches introduced in literature that take into account different structural information than the aforementioned measures.
We summarize four such methods. 

First, besides \textsc{vrq} as discussed in Section~\ref{sec:degreebased}, the work of Hay et al.~\cite{hay2007anonymizing, hay2008resisting} introduces two different attacker scenarios for which the first scenario is based on \emph{edge facts}.
This approach measures anonymity by counting the number of candidates given that an attacker knows a number of edges surrounding the node.
However, computing all possible subgraphs given a number of edges is very computationally expensive and would not scale to large graphs.
At the same time, if the edges are limited to the $d$-neighborhood of the graph, it can be at most as strict as $d$-$k$-anonymity mentioned at the end of Section~\ref{sec:neighbbased}.

Second, the \emph{hub fingerprints} approach assumes that the attacker knows the distance for each node to so-called \emph{hub} nodes.
These are nodes that have a central position in the network and often have a high degree.
However, experiments showed that this knowledge helps to identify only few nodes.
As this is not an effective attacker scenario, we chose to not include this measure in our comparison.

Third,~\cite{mohapatra2017level} introduces an attacker scenario where the attacker has additional structural information besides the neighborhood of a node.
This work assumes the attacker can obtain information about the centrality of nodes and explicitly aims to protect against this scenario.

Fourth and last,~\cite{dejong2023effect} introduces the notion of twin-uniqueness.
Under the assumption that the graph is shared, the attacker can identify so-called \emph{twin nodes}, which are sets of nodes that are connected to the same nodes. 
When an attacker finds that all candidates for an entity of interest are twins, they can learn all structural information there is to know about the node representing the entity, including the connections. 
The notion of \emph{twin-uniqueness} accounts for this phenomenon when computing anonymity and adds the constraint that for a node to be anonymous at least one equivalent node should not be a twin.

\section{Theoretical comparison of $k$-anonymity network measures}\label{sec:theory}
In this section, we aim to understand and compare the most prominent measures for anonymity introduced in the literature.
We do so by summarizing the measures based on the the structural information they take into account, which allows us to order them based on \emph{strictness}.
We formally define this notion in Definition~\ref{def:strictness}.

\begin{definition}[Strictness]\label{def:strictness}
    Given two measures $M_1$ and $M_2$, we say that $M_1$ is more strict than $M_2$, denoted $M_1 \geq M_2$, if for each pair of nodes $v, w$ it holds that $v \cong_{M_1} w$ implies that $\cong_{M_2} w$.
\end{definition}

This relation defines a partial order on the set of equivalence measures for nodes on  a network and has three important implications.
First, the more strict measure protects against all attacker scenarios represented by the more lenient measures.
Second, given two measures such that $M_1 \geq M_2$, for any graph $G=(V, E)$ it holds that $U_{M_1}(G) \geq U_{M_2}(G)$: the more strict measure always results in a higher or equal uniqueness, resulting in a lower anonymity.
Third, when anonymizing the network for a more strict measure, the minimum number of edge alterations, i.e., node or edge deletion and addition, required for anonymization is at least equal to the number of alterations required when anonymizing for the more lenient measure.
We formally prove this in Theorem~\ref{thm:anon} in Appendix~\ref{app:order}.

\begin{table}[t!]
\scriptsize
\caption{The $k$-anonymity measures compared in this paper (left column), their abbreviation (center column) and formula for computing their value $M(v, d)$ given an input node $v$ and distance $d$ (right column).}
\label{tab:measures}
\begin{tabular}{@{}rcl@{}}
\toprule
Measure                 & Abbreviation & $M(v, d)$                                                         \\ \midrule
Degree~\cite{liu2008towards, macwan2017k, casas2013algorithm, lu2012fast}
                  & \textsc{degree}          & $degree(v)$                                                                \\
Count~\cite{dejong2023algorithms}                   & \textsc{count}        & $(|V_{N_d(v)}|, |E_{N_d(v)} | )$                                           \\
Degree distribution~\cite{dejong2023algorithms}     & \textsc{degdist}      & $ \{degree(w) : w \in V_{N_d(v)} \}$                                       \\
$d$-$k$-Anonymity~\cite{alavi2019attacker, romanini2021privacy, zhou2008preserving, dejong2023effect}       & \textsc{$d$-$k$-anonymity}           & $\mathcal{C}(N_d(v))$\\
Vertex Refinement Query~\cite{hay2007anonymizing, hay2008resisting} & \textsc{vrq}          & $ \{degree(w) : w \in V, \ distance(v, w) \leq d\}$                                   \\
Hybrid~\cite{wang2013outsourcing}                  & \textsc{hybrid}      & $(\mathcal{C}(N_d(v)), \ \{degree(w) :  w \in V, distance(v, w) \leq d \})$             \\ \bottomrule
\end{tabular}
\end{table}

In our comparison, we account for the measures included in Table~\ref{tab:measures}, which shows for each measure its abbreviation and the equation to compute the value for a given node $v$.
All measures included, except for \textsc{degree}, are parameterized by parameter $d$. 
This value determines up to what distance from the considered node the structural information is taken into account.
For these parameterized variants, the equivalence classes can be computed recursively~\cite{dejong2023algorithms, hay2007anonymizing} by splitting each into new classes as the distance increases.
Below, we briefly summarize each of the measures used.

\begin{itemize}
    \item \textsc{degree}: the number of connections a node has.
    \item \textsc{count($d$)}: the number of nodes and edges in the $d$-neighborhood of the considered node.
    \item \textsc{degdist($d$)}: the degree distribution of the $d$-neighborhood of the considered node. Note that this is a multiset.
    \item \textsc{$d$-$k$-anonymity}: the exact structure of the $d$-neighborhood of the considered node. The structure is summarized by its canonical labeling.
    \item \textsc{vrq($d$)}: the degree of all nodes at distance at most $d$ from the considered node. Note that this is a multiset.
    \item \textsc{hybrid($d$)}: set containing the results of \textsc{$d$-$k$-anonymity} and \textsc{vrq($d$)}, i.e., the exact structure of the $d$-neighborhood of the considered node summarized by its canonical labeling and the degree of all nodes at distance at most $d$ from the considered node.  
\end{itemize}

We can use strictness to order the measures for $k$-anonymity by comparing them based on their reach, and what structural information they take into account.
First, we look at the neighborhood based measures being \textsc{degree}, \textsc{count}, \textsc{degdist} and \textsc{$d$-$k$-anonymity}.
As illustrated in the second and fourth column of Table~\ref{tab:meas_overview}, these measures all have the same reach, namely the $d$-neighborhood of the node.
However, the measures do differ in the extent to which they capture particular structural properties.
For \textsc{vrq}, the reach is larger, but the structural information is less precise.
As a result, \textsc{vrq} can not be included in the strictness ordering of neighborhood based measures.
As the \textsc{hybrid} measure combines both \textsc{$d$-$k$-anonymity} and \textsc{vrq}, this measure is more strict than both.
We formally prove that this ordering holds by showing that the stricter measures capture properties of the less strict measures.
The ordering resulting from Theorems~\ref{thm:meas_order} and~\ref{thm:meashay}, which can be found in Appendix~\ref{app:order}, is visualized in Figure~\ref{fig:ordering}.

    \begin{figure}[!t]
        \centering
        \includegraphics[width=\textwidth]{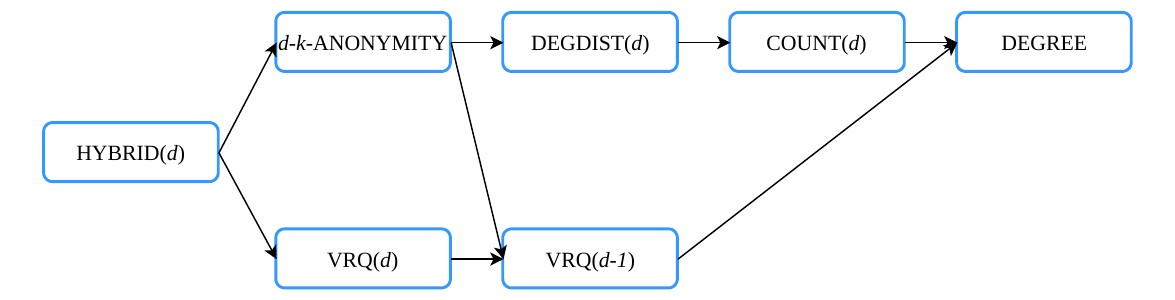}
        \caption{Ordering of anonymity measures based on strictness. An arrow between two measures ($A \rightarrow B$) implies measure $A$ is more strict than measure $B$. For readability, arrows implying strictness through transitivity are not included.
        }
        \label{fig:ordering}
    \end{figure}

\section{Empirical comparison of k-anonymity measures}\label{sec:emp_comparing}
In this section, we empirically compare the anonymity measures in Table~\ref{tab:measures} on a wide range of network datasets.
First, we summarize the experimental setup.
Second, we show how different measures affect the anonymity by comparing uniqueness.
Third, we study the effect of a cascading step on anonymity.
Fourth, we investigate how the choice in measure affects the trade-off between anonymity and data utility when anonymizing the network.
Lastly, we assess the runtime required to compute anonymity with each measure.

\begin{table}[b!]
\caption{Real-world network datasets used in the experiments. For each network, we list the number of nodes, edges, average degree, median degree, largest degree, clustering coefficient, assortativity, diameter, average distance and power law exponent $\alpha$.}
\label{tab:data}
\tiny
\setlength{\tabcolsep}{4pt}
\begin{tabular}{@{}rrrrrrrrrrr@{}}

\toprule
                Network & \multicolumn{1}{r}{$|V|$} &
                \multicolumn{1}{r}{$|E|$} 
                & \multicolumn{1}{r}{\begin{tabular}[c]{@{}r@{}}Avg. \\ degree\end{tabular}} 
                & \multicolumn{1}{r}{\begin{tabular}[c]{@{}r@{}}Median \\ degree\end{tabular}} 
                & \multicolumn{1}{r}{\begin{tabular}[c]{@{}r@{}}Max \\ degree\end{tabular}} 
                & \multicolumn{1}{r}{\begin{tabular}[c]{@{}r@{}}Clustering \\ coefficient\end{tabular}}  
                &\multicolumn{1}{r}{\begin{tabular}[c]{@{}r@{}}Assorta- \\ tivity \end{tabular}} 
                & \multicolumn{1}{r}{$D(G)$} 
                & \multicolumn{1}{r}{\begin{tabular}[c]{@{}r@{}}Avg. \\ distance \end{tabular}} 
                & \multicolumn{1}{r}{$\alpha$} \\ \midrule
 Radoslaw emails~\cite{kunegis2013konect}      & 167    & 3,250     & 38.92          & 40            & 139        & 0.69         & -0.30         & 5        & 1.97                & 4.61  \\
 Primary school~\cite{sociopatterns}           & 236    & 5,899     & 49.99        & 49            &    98        & 0.50         & 0.17          & 3        & 1.86                & 9.08  \\
 Moreno innov.~\cite{kunegis2013konect}        & 241    & 923       & 7.66           & 7             & 28         & 0.31         & -0.06         & 5        & 2.47                & 4.62  \\
 Gene fusion~\cite{kunegis2013konect}          & 291    & 279       & 1.92           & 1             & 34         & 0.00         & -0.35         & 9        & 3.90                & 2.50  \\
 Copnet calls~\cite{sapiezynski2019copenhagen} & 536    & 621       & 2.32           & 2             & 18         & 0.25         & 0.17          & 22       & 7.37                & 3.82  \\
 Copnet sms~\cite{sapiezynski2019copenhagen}   & 568    & 697       & 2.45           & 2             & 11         & 0.22         & 0.19          & 20       & 7.32                & 3.90  \\
 Copnet FB~\cite{sapiezynski2019copenhagen}    & 800    & 6,418     & 16.05          & 13            & 101        & 0.32         & 0.18          & 7        & 2.98                & 3.21  \\
 FB Reed98~\cite{networksrepository}           & 962    & 18,812    & 39.11          & 29            & 313        & 0.33         & 0.02          & 6        & 2.46                & 4.38  \\
 Arenas email~\cite{kunegis2013konect}         & 1,133  & 5,451     & 9.62           & 7             & 71         & 0.25         & 0.08          & 8        & 3.61                & 6.78  \\
 Network science~\cite{kunegis2013konect}      & 1,461  & 2,742     & 3.75           & 3             & 34         & 0.88         & 0.46          & 17       & 5.82                & 3.61  \\
 FB Simmons81~\cite{networksrepository}        & 1,518  & 32,988    & 43.46          & 37            & 300        & 0.33         & -0.06         & 7        & 2.57                & 4.74  \\
 DNC emails~\cite{kunegis2013konect}           & 1,893  & 4,385     & 4.63           & 1             & 402        & 0.59         & -0.31         & 8        & 3.37                & 2.01  \\
 Moreno health~\cite{kunegis2013konect}        & 2,539  & 10,455    & 8.24           & 8             & 27         & 0.15         & 0.25          & 10       & 4.56                & 8.24  \\
 FB Wellesley22~\cite{networksrepository}      & 2,970  & 94,899    & 63.91          & 52            & 746        & 0.27         & 0.06          & 8        & 2.59                & 4.60  \\
 Bitcoin alpha~\cite{networksrepository}       & 3,783  & 14,124    & 7.47           & 2             & 511        & 0.28         & -0.17         & 10       & 3.57                & 2.09  \\
 GRQC collab.~\cite{snapnets}                  & 5,242  & 14,484    & 5.53           & 3             & 81         & 0.69         & 0.66          & 17       & 6.05                & 2.11  \\
 FB Carnegie49~\cite{networksrepository}       & 6,637  & 249,967   & 75.33          & 54            & 840        & 0.29         & 0.12          & 8        & 2.74                & 4.98  \\
 Pajek Erdős~\cite{kunegis2013konect}          & 6,927  & 11,850    & 3.42           & 1             & 507        & 0.40         & -0.12         & 4        & 3.78                & 2.16  \\
 DT interaction~\cite{biosnapnets}             & 7,341  & 15,138    & 4.12           & 1             & 584        & 0.00         & -0.12         & 18       & 6.15                & 1.88  \\
 DG assoc.~\cite{biosnapnets}                  & 7,813  & 21,357    & 5.47           & 2             & 485        & 0.00         & -0.29         & 8        & 4.23                & 1.94  \\
 FB GWU54~\cite{networksrepository}            & 12,193 & 469,528   & 77.02          & 60            & 2,002      & 0.22         & 0.03          & 9        & 2.83                & 4.78  \\
 Anybeat~\cite{networksrepository}             & 12,645 & 49,132    & 7.77           & 2             & 4,800      & 0.40         & -0.12         & 10       & 3.17                & 1.75  \\
 CE-CX~\cite{networksrepository}               & 15,229 & 245,952   & 32.30          & 13            & 375        & 0.23         & 0.34          & 13       & 3.85                & 4.02  \\
 Astro Physics~\cite{networksrepository}       & 18,771 & 198,050   & 21.10          & 9             & 504        & 0.68         & 0.21          & 14       & 4.19                & 4.50  \\
 FB BU10~\cite{networksrepository}             & 19,700 & 637,528   & 64.72          & 51            & 1,819      & 0.20         & 0.05          & 9        & 3.03                & 5.44  \\
 FB Uillinois~\cite{networksrepository}        & 30,664 & 1,048,574 & 68.39          & 50            & 2,718      & 0.20         & 0.05          & 9        & 3.08                & 5.39  \\
 Enron email~\cite{snapnets}                   & 36,692 & 183,831   & 10.02          & 3             & 1,383      & 0.72         & -0.11         & 13       & 4.03                & 1.97  \\
 FB Penn94~\cite{networksrepository}           & 41,536 & 1,362,220 & 65.59          & 48            & 4,410      & 0.22         & 0.00          & 8        & 3.12                & 4.16  \\
 FB wall 2009~\cite{kunegis2013konect}         & 46,952 & 183,412   & 7.81           & 4             & 223        & 0.15         & 0.22          & 18       & 5.60                & 5.24  \\
 Brightkite~\cite{networksrepository}          & 58,228 & 214,078   & 7.35           & 2             & 1,134      & 0.27         & 0.01          & 18       & 4.92                & 2.48  \\
 The marker cafe~\cite{dataforgoodlab}         & 69,413 & 1,644,843 & 47.39          & 6             & 8,930      & 0.24         & -0.15         & 9        & 3.06                & 2.86  \\
 Slashdot zoo~\cite{kunegis2013konect}         & 79,116 & 467,731   & 11.82          & 2             & 2,534      & 0.09         & -0.07         & 12       & 4.04                & 3.46  \\ \bottomrule
\end{tabular}
\end{table}

\subsection{Experimental setup}\label{sec:expsetup} 
For the experiments, we use the networks in Table~\ref{tab:data}, which can each be found in their corresponding repositories.
All networks are interpreted as undirected, and additional properties such as weights or timestamps are ignored.
We assess anonymity in terms of uniqueness, i.e., the fraction of nodes that are $1$-anonymous, using the measures listed in Table~\ref{tab:measures} and a more generic version of anonymity-cascade as discussed in Section~\ref{sub:anon-cascade}.
We report on results for both $d=1$ and $d=2$, as previous work found a relatively large increase in uniqueness for $d$-$k$-anonymity with $d=2$ with only relatively small differences compared to larger values of $d$~\cite{dejong2023effect}.
To compute uniqueness, we use the same computational approach as described in~\cite{dejong2023algorithms}.
Code used can be found on github\footnote{\url{https://github.com/RacheldeJong/ANONET}}.

To assess the anonymity vs. utility trade-off, we use edge sampling~\cite{romanini2021privacy} as anonymization algorithm, which removes edges at random.
The resulting graphs are created by repeatedly deleting 1\% of the edges, resulting in 101 networks for each run including the initial network.
For each of these graphs, utility metrics as discussed in Section~\ref{sub:prelnetworks} are computed using igraph~\cite{igraph}. 
Data utility is assessed in terms of performance on three common downstream tasks in network analysis listed below:
\begin{itemize}
    \item \textbf{Robustness}: the fraction of nodes remaining in the largest connected component after anonymization.
    \item \textbf{Community structure}: the NMI~\cite{lancichinetti2012consensus} of communities found before and after anonymization. Communities are found by 10 runs of the Leiden algorithm~\cite{traag2019louvain} followed by a consensus clustering step~\cite{lancichinetti2012consensus}. 
    \item \textbf{Node centrality}: the overlap of nodes in the top 100 most central nodes according to betweenness centrality~\cite{brandes2001faster} before and after anonymization.
\end{itemize}
To account for non-determinism of the edge sampling algorithm, the anonymization process and utility computation for all 101 obtained networks are repeated five times.
Utility, uniqueness and runtime values are averaged over these five runs with a time limit of 3 hours. Results $\pm$ one standard deviation are reported.

All experiments are conducted on a machine with 1TB RAM, 64 AMD EPYC 7601 cores, and 128 threads. During the experiments, each run uses one thread that is not shared with other processes.

    \begin{figure}[t!]
        \centering
        \includegraphics[width=\textwidth]{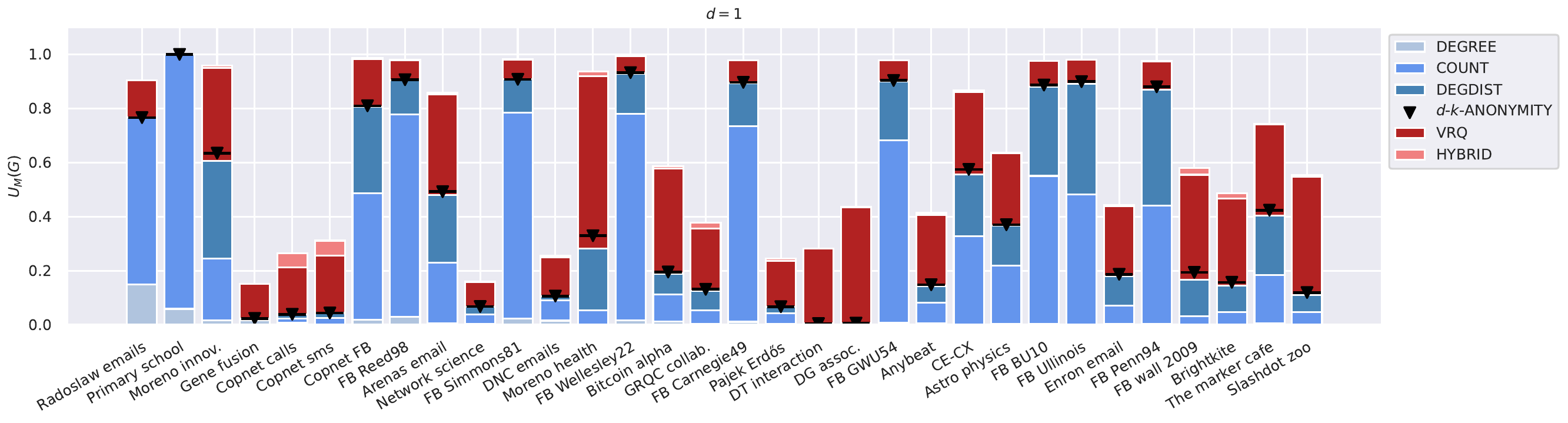}
        \includegraphics[width=\textwidth]{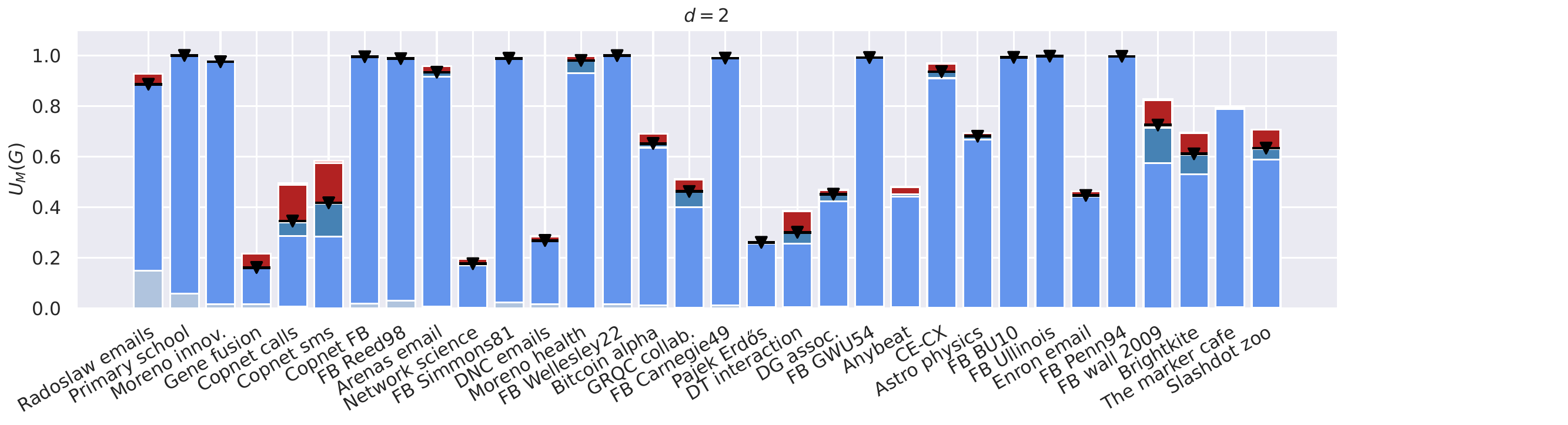}
        \caption{Uniqueness in empirical networks sorted by size from left to right. Fraction of unique nodes (vertical axis) on different datasets (horizontal axis) using different measures for $k$-anonymity: {\textsc{degree}} (\textcolor{lightsteelblue}{lightblue}), \textsc{count} (\textcolor{cornflowerblue}{blue}),  \textsc{degdist} (\textcolor{steelblue}{dark blue}), \textsc{$d$-$k$-anonymity} (black line with triangle), \textsc{vrq} (\textcolor{firebrick}{red}) and \textsc{hybrid} (\textcolor{lightcoral}{pink}).
        }
        \label{fig:measures}
    \end{figure}

\subsection{Comparing uniqueness}\label{sub:emp_comparing}
From Section~\ref{sec:theory}, we know that uniqueness, the inverse of anonymity, is expected to be higher for stricter measures and lower for less strict measures.
Figure~\ref{fig:measures} shows how the uniqueness varies in empirical networks by reporting this value for each measure at both $d=1$ and $d=2$.
In Appendix~\ref{app:diff}, we report for each measure the difference with the previous (less strict) measure in the legend of Figure~\ref{fig:measures}.
In an attempt to understand performance differences, we compare each difference with various topological network properties, of which the results can be found in Table~\ref{tab:corr_diff}. 
This table shows for each combination of difference and network property the Pearson correlation and corresponding $p$-value.
Plots of the values for all significantly correlated combinations can be found in Appendix~\ref{app:diffexplain}. 
We use these findings in the discussion of the results.

\begin{table}[b!]
\caption{Correlation between network properties (leftmost column) and the difference observed between the use of two measures (second up to last column). For each combination of network property and outcome both the correlations and $p$-value are reported. Values with $p<0.05$ and Pearson correlation larger than 0.4 or smaller than -0.4 are shown in bold.}
\label{tab:corr_diff}
\tiny
\setlength{\tabcolsep}{4pt}
\begin{tabular}{@{}rrlrlrlrlrl@{}}
\toprule
                            & \multicolumn{2}{c}{\textsc{count} vs. \textsc{degree}} & \multicolumn{2}{c}{\textsc{degdist} vs. \textsc{count}} & \multicolumn{2}{c}{\textsc{$d$-$k$-anonymity} vs. \textsc{degdist}} & \multicolumn{2}{c}{\textsc{vrq} vs. \textsc{$d$-$k$-anonymity}} & \multicolumn{2}{c}{\textsc{hybrid} vs. \textsc{vqr}} \\ \cmidrule(l){2-11} \cmidrule(l){2-11} 
Network property        & \begin{tabular}[c]{@{}r@{}}Pearson\\ correlation\end{tabular} & p-value  & \begin{tabular}[c]{@{}r@{}}Pearson\\ correlation\end{tabular} & p-value  & \begin{tabular}[c]{@{}r@{}}Pearson\\ correlation\end{tabular} & p-value  & \begin{tabular}[c]{@{}r@{}}Pearson\\ correlation\end{tabular} & p-value  & \begin{tabular}[c]{@{}r@{}}Pearson\\ correlation\end{tabular} & p-value  \\ \midrule
Nodes               & -0.216                                                        & 2.52E-01          & 0.184                                                         & 3.30E-01          & 0.223                                                         & 2.36E-01          & 0.275                                                         & 1.42E-01          & 0.002                                                         & 9.92E-01          \\
Edges               & 0.223                                                         & 2.37E-01          & \textbf{0.661}                                                & \textbf{7.11E-05} & 0.083                                                         & 6.63E-01          & -0.223                                                        & 2.37E-01          & -0.223                                                        & 2.36E-01          \\
Average degree      & 0.118                                                         & 5.33E-01          & 0.267                                                         & 1.54E-01          & -0.065                                                        & 7.33E-01          & -0.058                                                        & 7.62E-01          & -0.050                                                        & 7.94E-01          \\
Median degree       & \textbf{0.898}                                                & \textbf{1.69E-11} & \textbf{0.437}                                                & \textbf{1.58E-02} & -0.163                                                        & 3.91E-01          & \textbf{-0.651}                                               & \textbf{9.94E-05} & \textbf{-0.426}                                               & \textbf{1.91E-02} \\
Max degree          & \textbf{0.587}                                                & \textbf{6.56E-04} & 0.294                                                         & 1.14E-01          & -0.161                                                        & 3.95E-01          & \textbf{-0.416}                                               & \textbf{2.23E-02} & -0.324                                                        & 8.06E-02          \\
Alpha               & \textbf{0.536}                                                & \textbf{2.25E-03} & 0.354                                                         & 5.52E-02          & \textbf{0.419}                                                & \textbf{2.12E-02} & -0.074                                                        & 6.97E-01          & -0.046                                                        & 8.08E-01          \\
Density             & \textbf{0.528}                                                & \textbf{2.68E-03} & -0.238                                                        & 2.05E-01          & -0.186                                                        & 3.25E-01          & -0.324                                                        & 8.10E-02          & -0.196                                                        & 2.98E-01          \\
Transitivity        & 0.081                                                         & 6.70E-01          & -0.193                                                        & 3.07E-01          & -0.228                                                        & 2.25E-01          & -0.303                                                        & 1.03E-01          & -0.090                                                        & 6.35E-01          \\
Assortativity       & 0.018                                                         & 9.23E-01          & 0.194                                                         & 3.04E-01          & 0.254                                                         & 1.76E-01          & -0.005                                                        & 9.79E-01          & 0.343                                                         & 6.38E-02          \\
Diameter            & \textbf{-0.601}                                               & \textbf{4.50E-04} & -0.273                                                        & 1.45E-01          & 0.030                                                         & 8.75E-01          & 0.273                                                         & 1.44E-01          & \textbf{0.680}                                                & \textbf{3.53E-05} \\
Average distance & \textbf{-0.748}                                               & \textbf{2.01E-06} & -0.381                                                        & 3.78E-02 & 0.024                                                         & 8.98E-01          & 0.355                                                         & 5.44E-02          & \textbf{0.757}                                                & \textbf{1.27E-06} \\ \bottomrule
\end{tabular}

\end{table}

First, we focus on the results for $d=1$ shown in the top of Figure~\ref{fig:measures}.
We observe that \textsc{degree} is not very effective for identifying nodes as for all networks there is only a tiny fraction of unique nodes.
However, with \textsc{count}, which additionally accounts for the number of edges in the neighborhood, the uniqueness increases drastically for many networks, sometimes showing results similar to those of $1$-$k$-anonymity.
The difference appears to be especially larger in networks with higher degree, lower diameter or average distance, and a more balanced degree distribution, indicated by a higher alpha value (see Table~\ref{tab:data}).
Uniqueness obtained by \textsc{degdist} is in many cases equal to $1$-$k$-anonymity, and otherwise very similar.
For networks with a high alpha, the difference is larger.
This shows that \textsc{count} and \textsc{degdist}, measures that use imprecise structural information and hence represent a weaker attacker scenario, appear to be useful for approximating \textsc{$d$-$k$-anonymity}, which accounts for exact structural information.

The \textsc{vrq} measure, which could not fit in our theoretical strictness ordering of neighborhood based measures,  achieves higher uniqueness than \textsc{$1$-$k$-anonymity}, which indicates that having knowledge beyond the $1$-neighborhood of a considered node can have a strong de-anonymizing effect: it is in many cases even more effective than perfect knowledge of the $1$-neighborhood.
We find that the difference is often larger for networks with a high diameter or average distance, or lower median degree.
For the \textsc{hybrid} measure, adding exact knowledge of the $1$-neighborhood structure to \textsc{vrq}, we see for only a few networks a slight increase in uniqueness.
This implies that the sets of nodes that can be identified with the measures tend to overlap largely, and this additional knowledge of the neighborhood structure only has a small effect on uniqueness.
Hence, when choosing an attacker scenario to protect against, \textsc{vrq} may be more effective than \textsc{$d$-$k$-anonymity}.

Results for $d=2$ at the bottom of Figure~\ref{fig:measures} show a very large increase of uniqueness compared to the results for $d=1$.
Here the uniqueness values obtained by \textsc{degdist} are always the same as that of \textsc{$d$-$k$-anonymity} and that of \textsc{count} is similar to both. 
\textsc{vrq} always achieves a uniqueness larger than or equal to that of \textsc{$d$-$k$-anonymity}.
The generally small difference between \textsc{vrq} and \textsc{$d$-$k$-anonymity} with $d=2$ is consistent with the finding in~\cite{dejong2023effect} that accounting for information beyond the $2$-neighborhood does not have a large effect on uniqueness.
For the \textsc{hybrid} measure, we see that this never results in a higher uniqueness when considering the $2$-neighborhood.

Overall, these results demonstrate that having imprecise information, even as little as the number of connections of nodes at distances 1 up to $d$, can be sufficient to uniquely identify a large fraction of nodes.
Moreover, this knowledge is more revealing than complete knowledge using a smaller radius.
One may not even need complete structural information to identify many nodes, as it appears that having complete structural information often has a very limited additional effect on uniqueness.

\subsection{Anonymity-cascade}\label{sec:casc}
The results from the previous section showed that measures that reach further lead to a much higher uniqueness and hence represents more effective attacker scenarios.
In this section, we investigate the effect on uniqueness when using a different method to increase, namely by cascading the information of unique nodes to identify more nodes.
To do so, we measure uniqueness using a generic variant of anonymity-cascade~\cite{dejong2023effect}, explained in Section~\ref{sub:anon-cascade}.
The algorithm consists of two steps.
First, we identify nodes using the \emph{initial measure}.
Second, we use the \emph{cascade measure} to identify more nodes in a cascading fashion. 
This is also illustrated in Figure~\ref{fig:cascade}.
The cascading process can continue for multiple levels.
If only one level of cascading is used, we refer to this as $C_1$.
The process can then repeat until no more nodes can be identified, referred to as \emph{cascading final} $C_f$.

    \begin{figure}[t]
       \centering
       \includegraphics[width=0.35\textwidth]{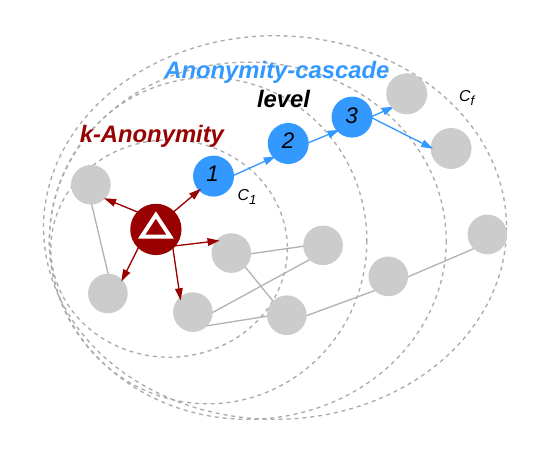}
       \caption{Illustration of anonymity-cascade. The red node marked with a triangle denotes a unique node identified with the initial measure. From this node the cascade process is started using the cascade measure. The numbers denote the level of the cascading process in which the node is identified.
       Blue nodes are identified as part of the cascading process.}
       \label{fig:cascade}
    \end{figure}

    \begin{figure}[t!]
        \centering
        \includegraphics[height=0.85\textheight]{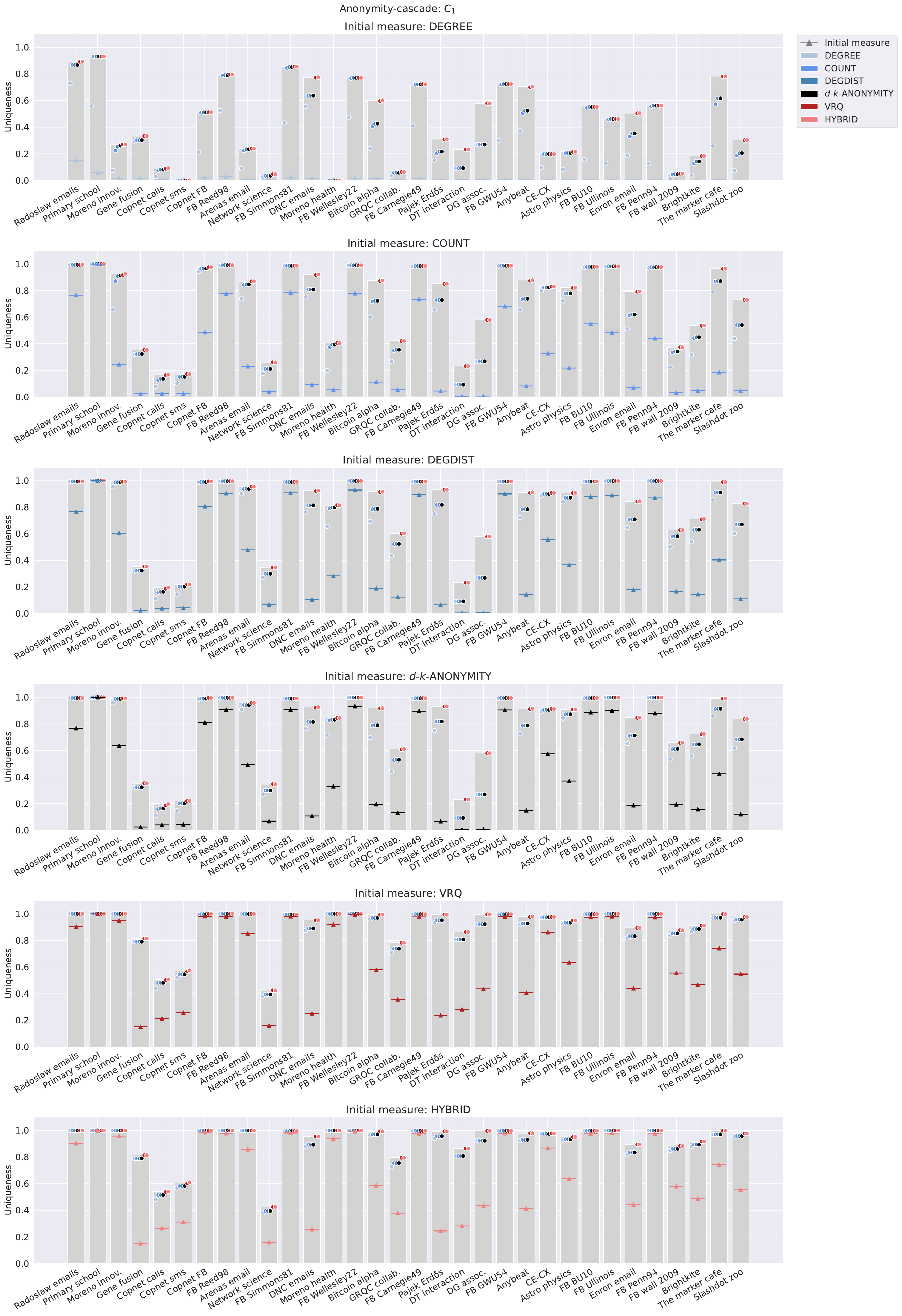}
        \caption{Uniqueness using anonymity-cascade one level ($C_1$). Each figure corresponds to a different initial measure (triangle and line) and is combined with all of the cascade measures (dots).
        The grey bar indicates the highest obtained uniqueness for the network given the initial measure. Measures used are: {\textsc{degree}} (\textcolor{lightsteelblue}{lightblue}), \textsc{count} (\textcolor{cornflowerblue}{blue}),  \textsc{degdist} (\textcolor{steelblue}{dark blue}), \textsc{$d$-$k$-anonymity} (black triangle), \textsc{vrq} (\textcolor{firebrick}{red}) and \textsc{hybrid} (\textcolor{lightcoral}{pink}).
        }
        \label{fig:rescascade}
    \end{figure}

The combinations of initial and cascade measure consist of all 36 combinations of the six measures in Table~\ref{tab:measures}.
This allows us to investigate all combinations of attacker knowledge.
Besides the cases where the levels of knowledge are equal, this also includes the cases where an attacker has little starting knowledge and more local knowledge surrounding the identified nodes, hence more cascade knowledge.
While we fully acknowledge that not every combination of measures is equally realistic, we chose to include all to observe the cascading effect of further reach on the uniqueness.
We focus on one level of cascading, as this is a more realistic scenario. (The results for cascading final can be found in Appendix~\ref{app:cascade})

Figure~\ref{fig:rescascade} shows the uniqueness for anonymity-cascade with one level of cascading.
Each subfigure shows results for a different initial measure, denoted by the line with a triangle, and results of one level of cascading, denoted by the colored dots, using each of the six measures as cascading measure.
The gray bar is added for readability of the results and shows the highest uniqueness obtained.

First, the results for \textsc{degree} show a very low initial uniqueness.
However, one level of cascading results in a much higher uniqueness, indicating that even with minimal initial knowledge, i.e., the degree of a node, many more nodes can be identified with additional cascade knowledge.
When using \textsc{degree} as cascading measure, this impact is often smaller.
In most cases using \textsc{vrq} as cascading measure has the largest impact on the uniqueness.

The observed uniqueness when using \textsc{degree} as cascading measure, still differs substantially from results for the other measures.
Results show that \textsc{count}, \textsc{degdist} and \textsc{$d$-$k$-anonymity} achieve similar results, where using \textsc{vrq} as cascade measure does increase uniqueness.
With \textsc{vrq} as initial measure we do observe both a higher initial uniqueness and uniqueness after cascading.

Overall, these results show, similar to the results of Section~\ref{sub:emp_comparing}, that by reaching further, with one or more (see Appendix~\ref{app:cascade}) cascading steps, the uniqueness increases immensely.
This effect is especially clear when using \textsc{degree} as initial measure: if the nodes with unique degrees can be identified, this can be reused to identify many other nodes. 

    \begin{figure}[t!]
       \centering
       \includegraphics[width=\textwidth]{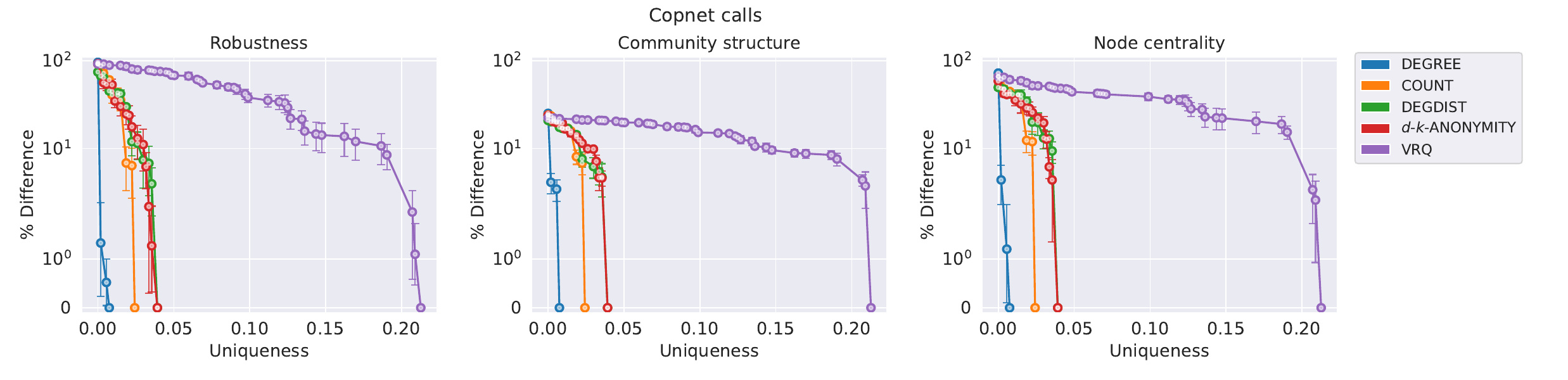}
       \includegraphics[width=\textwidth]{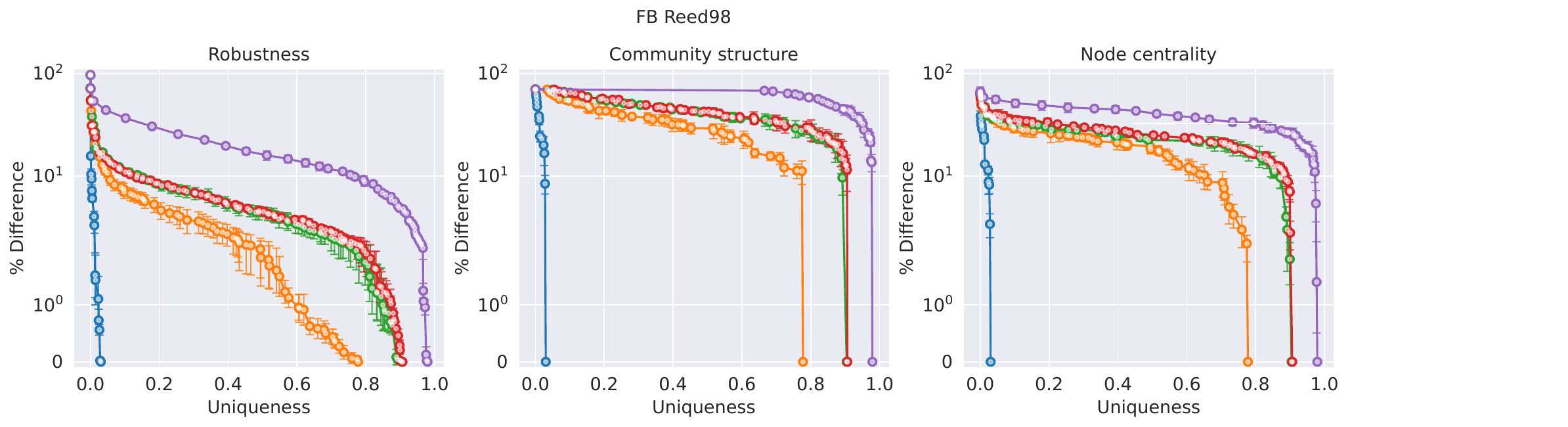}
       \caption{Pareto optimal solutions found in terms of uniqueness (horizontal axis) and performance on a downstream task (vertical axis). Each dot represents a solution found (partially anonymized network) when anonymizing using edge sampling for one of the five anonymity measures (colors). Error bars indicate results $\pm$ one standard deviation.
       }
       \label{fig:util}
    \end{figure}

\subsection{Data utility}\label{sub:utility}
In this section, we empirically investigate how the chosen measure affects the trade-off between utility and anonymity when anonymizing the networks by means of edge sampling~\cite{romanini2021privacy}.
Due to computational reasons, we perform anonymization on the 14 smallest networks.
For each partially anonymous network found during the anonymization process, we determine the uniqueness, difference in robustness, community structure, and 100 most central nodes.
To represent the trade-offs found when using different measures, Figure~\ref{fig:util} contains for two networks the Pareto fronts which represent the best trade-off values in terms of observed uniqueness and difference in performance on the three downstream tasks.
Hence, each result (dot) represents a network with a different percentage of deleted edges.
Figure~\ref{fig:util} shows the results on two networks. 
Results for the remainder of the networks, which appear to follow similar behavior, can be found
in Appendix~\ref{app:utility}.

Overall, we can see a large difference between the obtained Pareto fronts for the measures where the most lenient measures, \textsc{degree} and \textsc{count}, find better trade-offs and the measure with the largest starting uniqueness, \textsc{vrq}, finds worse trade-offs.
When using the \textsc{degdist} and \textsc{$d$-$k$-anonymity} measures, the solutions found are very similar in terms of uniqueness and utility.
Surprisingly, for the \textsc{count} measure, the trade-offs found are clearly better.
This difference is larger for networks where the starting uniqueness differs more between the \textsc{count} and \textsc{degdist} measure, such as ``FB Reed98".

Comparing the downstream tasks, we observe that for both community structure and node centrality a big decrease in utility needs to be sacrificed to improve uniqueness. 
For the largest component size, this difference is more gradual, yet steep for \textsc{vrq}.
Hence, these results show that when using a more strict measure, or one that reaches further, more utility has to be sacrificed in order to obtain the same uniqueness.

\subsection{Comparing runtimes}
In this section, we investigate the computational aspect by measuring the time required to compute uniqueness using each of the six anonymity measures.
It is worth nothing that computing anonymity consists of two parts: computing the value for a node using the chosen measure and bookkeeping operations to determine, based on the values computed, the right equivalence class for each node to arrive at the network's uniqueness value.
In Figure~\ref{fig:measures_runtime}, we show the runtime for each of the measures on the set of considered real-world networks.
We first focus on the results for $d=1$ in the top figure.

        \begin{figure}[b!]
       \centering
       \includegraphics[width=\textwidth]{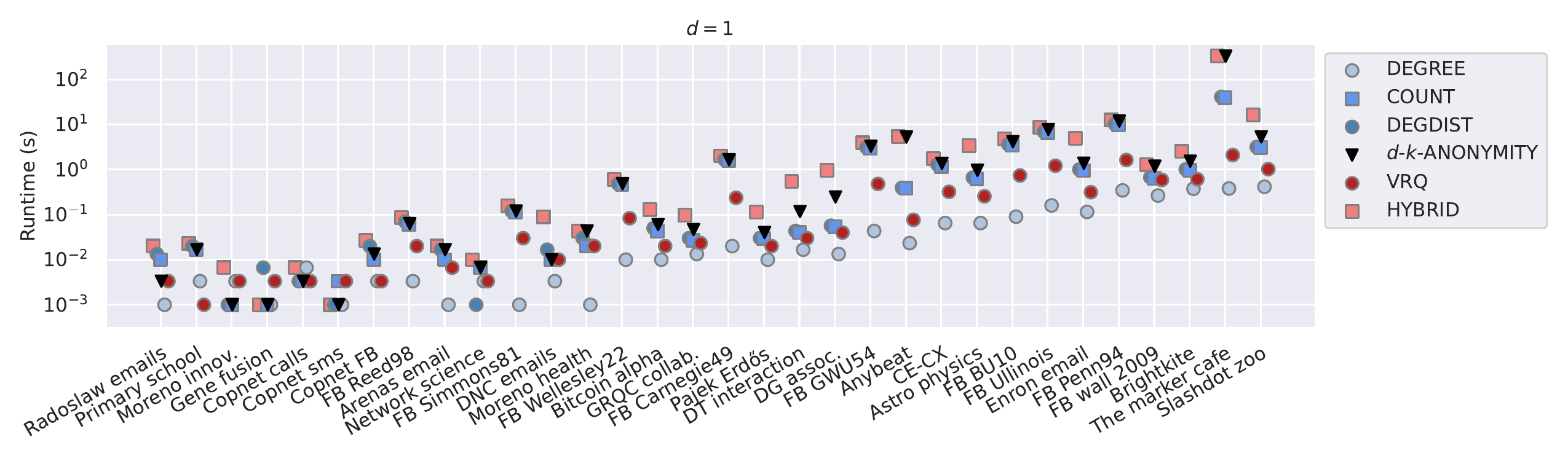}
       \includegraphics[width=\textwidth]{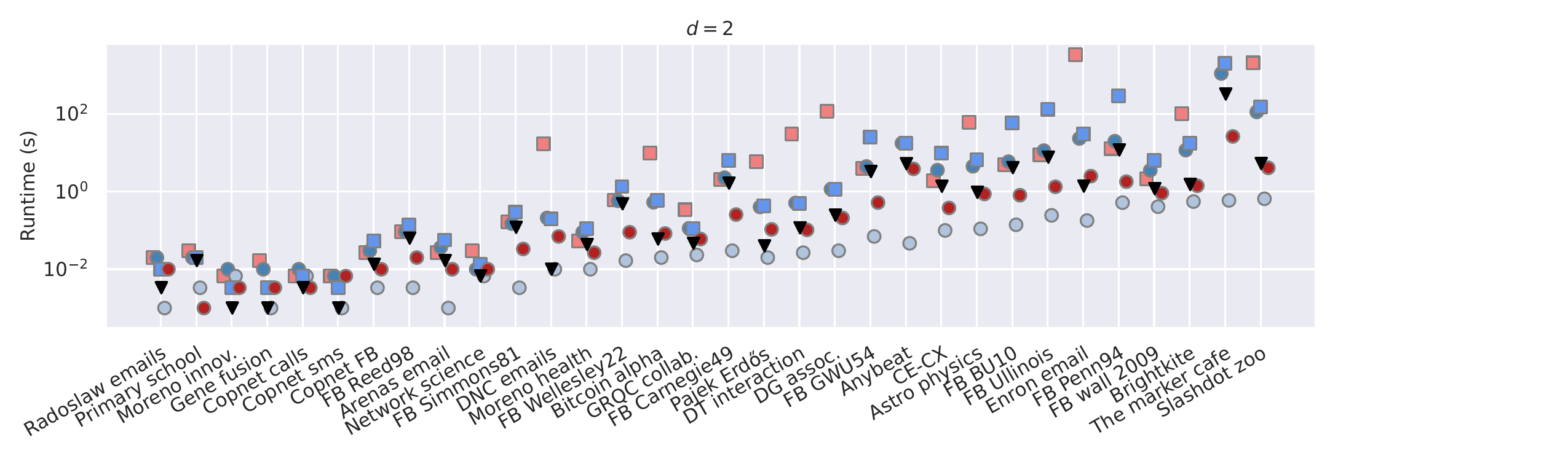}
       \caption{Runtime in seconds (vertical axis) on different network datasets (horizontal axis) for different $k$-anonymity measures: {\textsc{degree}} (\textcolor{lightsteelblue}{lightblue circle}), \textsc{count} (\textcolor{cornflowerblue}{blue square}),  \textsc{degdist} (\textcolor{steelblue}{dark blue circle}), \textsc{$d$-$k$-anonymity} (black triangle), \textsc{vrq} (\textcolor{firebrick}{red circle}) and \textsc{hybrid} (\textcolor{lightcoral}{pink square}).
       }
       \label{fig:measures_runtime}
    \end{figure}

Overall, \textsc{degree} and \textsc{vrq} are the least expensive to compute.
This can be explained by the fact that for the neighborhood based measures, \textsc{count}, \textsc{degdist} and \textsc{$d$-$k$-anonymity}, the $d$-neighborhood of the considered node needs to be extracted, which is relatively time consuming.
Instead, for \textsc{degree} merely the degree of the considered node needs to be determined, and for \textsc{vrq} one only needs to iterate over the neighboring nodes and determine their degree.
For the neighborhood based measures, computing the degree distribution as required for \textsc{degdist} or determining the canonical labeling as required for \textsc{$d$-$k$-anonymity} and \textsc{hybrid}, increases the runtime notably for some networks.
For most networks \textsc{count} is less time consuming than the other neighborhood based measures except for a few other networks, such as ``Radoslaw emails" and ``FB Wellesley22", for which more time appears to be consumed by bookkeeping operations.
As a result, \textsc{count} achieves runtimes similar to or higher than other measures, including \textsc{$d$-$k$-anonymity} and \textsc{hybrid}.

For $d=2$, shown in the bottom of Figure~\ref{fig:measures_runtime}, we see that the runtimes are more diverse, up to ~15 minutes for some of the denser networks.
For several large networks, ``The marker cafe" and ``Anybeat", the values could not be computed for \textsc{$d$-$k$-anonymity} and \textsc{hybrid} within the given time limit.
This is due to the canonical labeling computation which is in very few cases expensive to compute~\cite{dejong2023effect} and can immensely increase the runtime.
This unpredictability in runtime is not a desirable property.
For $d=2$ we again notice that for some networks the runtime is much larger for \textsc{count}.
This is likely due to the lower uniqueness achieved by this measure for $d=1$.
As a result, the value for more nodes should be determined and more bookkeeping operations are required to obtain the equivalence partition for $d=2$.
It appears that \textsc{hybrid} is often the most time consuming measure as it combines \textsc{vrq} and \textsc{$d$-$k$-anonymity}.

\section{Conclusions}\label{sec:conc}
In this paper, we have systematically compared six measures for $k$-anonymity, both theoretically and empirically.
The choice in anonymity measure strongly affects three important aspects, being, 1) the considered attacker scenario and its effectiveness, 2) the anonymity vs. utility trade-off, and 3) computation time. 

First, we provided an overview of the different measures introduced in literature resulting in a categorization into degree based, neighborhood based and automorphism based measures. 
Then, we compared the selected $k$-anonymity measures theoretically, showing how they can  be distinguished based on how far they reach, and the completeness of the structural information they take into account.
Based on this, we derived a formal distance-parameterized strictness-based ordering of the measures.
If an anonymity measure is more strict than another measure, this has three implications: 1) The more strict measure protects against all attacker scenarios the less strict measure protects against, 2) the anonymity of a network can be at most equal to that measured by the less strict measure, and 3) anonymization for the more strict measure requires at least as many graph alterations.

To better understand how the different measures impact the three aspects, we performed four  experiments on a wide range of empirical networks.
First, we measured uniqueness of the networks to assess the effectiveness of the considered attacker scenario mimicked by the anonymity measure.
Most importantly, the results showed that it is more de-anonymizing to have incomplete information that accounts for structural information that reaches further, beyond the direct neighborhood of the considered node, than complete information concerning its direct neighborhood.
Even more so, by means of the \textsc{hybrid} measure we show that often having (more) complete information only has a small additional effect on the uniqueness if the reach remains the same.
At the same time these less complete measures, even if they reach further, require less time to compute anonymity.
Hence, when choosing a measure for anonymity, it would be better to choose a measure that reaches beyond the direct neighborhood (assuming that this is a realistic attacker scenario), rather than a measure that accounts for as complete as possible structural information.

Second, we investigate the effect of extending the reach of the anonymity measure by reusing unique nodes in a cascading fashion to identify more nodes.
Just one level of cascading has an immense impact on the anonymity, especially when using \textsc{degree} as initial measure, which by itself identifies very few nodes.

Third, we investigated empirically how the chosen measure affects the trade-off between anonymity and utility by measuring the difference in performance on three downstream tasks during anonymization.
The results showed that using more strict measures, which protect against more strict attacker scenarios, results in higher observed uniqueness and substantially more utility has to be sacrificed to obtain more anonymity.
The \textsc{count} measure resulted in slightly better trade-offs than \textsc{degdist} and \textsc{$d$-$k$-anonymity}, which achieved similar results in terms of uniqueness and utility.

Lastly, we investigated the runtime required to measure anonymity with the different measures.
The most affecting factor was the reach of the measures.
As the neighborhood based measures, \textsc{count}, \textsc{degdist} and \textsc{$d$-$k$-anonymity}, need to extract the neighborhood of considered nodes, these measures overall required a longer runtime compared to \textsc{degree} and \textsc{vrq} which need solely to compute the degree of nodes.

There are still many possibilities for future work.
Most measures for anonymity so far focus on the direct neighborhood. 
At the same time our results on measuring anonymity show the necessity to look beyond the 1-neighborhood.
A crucial next step would be to design anonymization algorithms accounting for measures that look beyond the 1-neighborhood of the node.
Furthermore, the existing anonymity measures can be extended to account for additional network properties, such as node or edge labels, or even temporality.
Overall, this paper provides a foundation for future research on $k$-anonymity maximization algorithms of networks, for designing anonymity measures, and, in general, aiding in the selection of suitable existing measures in applied scenarios.

\bibliography{sn-bibliography}


\begin{thebibliography}{96}
\ifx \bisbn   \undefined \def \bisbn  #1{ISBN #1}\fi
\ifx \binits  \undefined \def \binits#1{#1}\fi
\ifx \bauthor  \undefined \def \bauthor#1{#1}\fi
\ifx \batitle  \undefined \def \batitle#1{#1}\fi
\ifx \bjtitle  \undefined \def \bjtitle#1{#1}\fi
\ifx \bvolume  \undefined \def \bvolume#1{\textbf{#1}}\fi
\ifx \byear  \undefined \def \byear#1{#1}\fi
\ifx \bissue  \undefined \def \bissue#1{#1}\fi
\ifx \bfpage  \undefined \def \bfpage#1{#1}\fi
\ifx \blpage  \undefined \def \blpage #1{#1}\fi
\ifx \burl  \undefined \def \burl#1{\textsf{#1}}\fi
\ifx \doiurl  \undefined \def \doiurl#1{\url{https://doi.org/#1}}\fi
\ifx \betal  \undefined \def \betal{\textit{et al.}}\fi
\ifx \binstitute  \undefined \def \binstitute#1{#1}\fi
\ifx \binstitutionaled  \undefined \def \binstitutionaled#1{#1}\fi
\ifx \bctitle  \undefined \def \bctitle#1{#1}\fi
\ifx \beditor  \undefined \def \beditor#1{#1}\fi
\ifx \bpublisher  \undefined \def \bpublisher#1{#1}\fi
\ifx \bbtitle  \undefined \def \bbtitle#1{#1}\fi
\ifx \bedition  \undefined \def \bedition#1{#1}\fi
\ifx \bseriesno  \undefined \def \bseriesno#1{#1}\fi
\ifx \blocation  \undefined \def \blocation#1{#1}\fi
\ifx \bsertitle  \undefined \def \bsertitle#1{#1}\fi
\ifx \bsnm \undefined \def \bsnm#1{#1}\fi
\ifx \bsuffix \undefined \def \bsuffix#1{#1}\fi
\ifx \bparticle \undefined \def \bparticle#1{#1}\fi
\ifx \barticle \undefined \def \barticle#1{#1}\fi
\bibcommenthead
\ifx \bconfdate \undefined \def \bconfdate #1{#1}\fi
\ifx \botherref \undefined \def \botherref #1{#1}\fi
\ifx \url \undefined \def \url#1{\textsf{#1}}\fi
\ifx \bchapter \undefined \def \bchapter#1{#1}\fi
\ifx \bbook \undefined \def \bbook#1{#1}\fi
\ifx \bcomment \undefined \def \bcomment#1{#1}\fi
\ifx \oauthor \undefined \def \oauthor#1{#1}\fi
\ifx \citeauthoryear \undefined \def \citeauthoryear#1{#1}\fi
\ifx \endbibitem  \undefined \def \endbibitem {}\fi
\ifx \bconflocation  \undefined \def \bconflocation#1{#1}\fi
\ifx \arxivurl  \undefined \def \arxivurl#1{\textsf{#1}}\fi
\csname PreBibitemsHook\endcsname

\bibitem[\protect\citeauthoryear{Saxena and Iyengar}{2020}]{saxena2020centrality}
\begin{botherref}
\oauthor{\bsnm{Saxena}, \binits{A.}},
\oauthor{\bsnm{Iyengar}, \binits{S.}}:
Centrality measures in complex networks: A survey.
arXiv preprint arXiv:2011.07190
(2020)
\end{botherref}
\endbibitem

\bibitem[\protect\citeauthoryear{Leskovec et~al.}{2010}]{leskovec2010empirical}
\begin{bchapter}
\bauthor{\bsnm{Leskovec}, \binits{J.}},
\bauthor{\bsnm{Lang}, \binits{K.J.}},
\bauthor{\bsnm{Mahoney}, \binits{M.}}:
\bctitle{Empirical comparison of algorithms for network community detection}.
In: \bbtitle{Proceedings of the 19th International Conference on World Wide Web},
pp. \bfpage{631}--\blpage{640}
(\byear{2010})
\end{bchapter}
\endbibitem

\bibitem[\protect\citeauthoryear{Bhuyan et~al.}{2013}]{bhuyan2013network}
\begin{barticle}
\bauthor{\bsnm{Bhuyan}, \binits{M.H.}},
\bauthor{\bsnm{Bhattacharyya}, \binits{D.K.}},
\bauthor{\bsnm{Kalita}, \binits{J.K.}}:
\batitle{Network anomaly detection: methods, systems and tools}.
\bjtitle{IEEE Communications Surveys \& Tutorials}
\bvolume{16}(\bissue{1}),
\bfpage{303}--\blpage{336}
(\byear{2013})
\end{barticle}
\endbibitem

\bibitem[\protect\citeauthoryear{Azizi et~al.}{2020}]{azizi2020epidemics}
\begin{barticle}
\bauthor{\bsnm{Azizi}, \binits{A.}},
\bauthor{\bsnm{Montalvo}, \binits{C.}},
\bauthor{\bsnm{Espinoza}, \binits{B.}},
\bauthor{\bsnm{Kang}, \binits{Y.}},
\bauthor{\bsnm{Castillo-Chavez}, \binits{C.}}:
\batitle{Epidemics on networks: Reducing disease transmission using health emergency declarations and peer communication}.
\bjtitle{Infectious Disease Modelling}
\bvolume{5},
\bfpage{12}--\blpage{22}
(\byear{2020})
\end{barticle}
\endbibitem

\bibitem[\protect\citeauthoryear{Kazmina et~al.}{2024}]{kazmina2024socio}
\begin{barticle}
\bauthor{\bsnm{Kazmina}, \binits{Y.}},
\bauthor{\bsnm{Heemskerk}, \binits{E.M.}},
\bauthor{\bsnm{Bok{\'a}nyi}, \binits{E.}},
\bauthor{\bsnm{Takes}, \binits{F.W.}}:
\batitle{Socio-economic segregation in a population-scale social network}.
\bjtitle{Social Networks}
\bvolume{78},
\bfpage{279}--\blpage{291}
(\byear{2024})
\end{barticle}
\endbibitem

\bibitem[\protect\citeauthoryear{Bojanowski and Corten}{2014}]{bojanowski2014measuring}
\begin{barticle}
\bauthor{\bsnm{Bojanowski}, \binits{M.}},
\bauthor{\bsnm{Corten}, \binits{R.}}:
\batitle{Measuring segregation in social networks}.
\bjtitle{Social Networks}
\bvolume{39},
\bfpage{14}--\blpage{32}
(\byear{2014})
\end{barticle}
\endbibitem

\bibitem[\protect\citeauthoryear{Savage et~al.}{2014}]{savage2014anomaly}
\begin{barticle}
\bauthor{\bsnm{Savage}, \binits{D.}},
\bauthor{\bsnm{Zhang}, \binits{X.}},
\bauthor{\bsnm{Yu}, \binits{X.}},
\bauthor{\bsnm{Chou}, \binits{P.}},
\bauthor{\bsnm{Wang}, \binits{Q.}}:
\batitle{Anomaly detection in online social networks}.
\bjtitle{Social networks}
\bvolume{39},
\bfpage{62}--\blpage{70}
(\byear{2014})
\end{barticle}
\endbibitem

\bibitem[\protect\citeauthoryear{Bok{\'a}nyi et~al.}{2023}]{bokanyi2023anatomy}
\begin{barticle}
\bauthor{\bsnm{Bok{\'a}nyi}, \binits{E.}},
\bauthor{\bsnm{Heemskerk}, \binits{E.M.}},
\bauthor{\bsnm{Takes}, \binits{F.W.}}:
\batitle{The anatomy of a population-scale social network}.
\bjtitle{Scientific Reports}
\bvolume{13}(\bissue{1}),
\bfpage{9209}
(\byear{2023})
\end{barticle}
\endbibitem

\bibitem[\protect\citeauthoryear{Van~der Laan et~al.}{2023}]{van2023whole}
\begin{barticle}
\bauthor{\bsnm{Laan}, \binits{J.}},
\bauthor{\bsnm{Jonge}, \binits{E.}},
\bauthor{\bsnm{Das}, \binits{M.}},
\bauthor{\bsnm{Te~Riele}, \binits{S.}},
\bauthor{\bsnm{Emery}, \binits{T.}}:
\batitle{A whole population network and its application for the social sciences}.
\bjtitle{European sociological review}
\bvolume{39}(\bissue{1}),
\bfpage{145}--\blpage{160}
(\byear{2023})
\end{barticle}
\endbibitem

\bibitem[\protect\citeauthoryear{Backstrom et~al.}{2007}]{backstrom2007wherefore}
\begin{bchapter}
\bauthor{\bsnm{Backstrom}, \binits{L.}},
\bauthor{\bsnm{Dwork}, \binits{C.}},
\bauthor{\bsnm{Kleinberg}, \binits{J.}}:
\bctitle{Wherefore art thou r3579x? anonymized social networks, hidden patterns, and structural steganography}.
In: \bbtitle{Proceedings of the 16th International Conference on World Wide Web},
pp. \bfpage{181}--\blpage{190}
(\byear{2007})
\end{bchapter}
\endbibitem

\bibitem[\protect\citeauthoryear{Romanini et~al.}{2021}]{romanini2021privacy}
\begin{barticle}
\bauthor{\bsnm{Romanini}, \binits{D.}},
\bauthor{\bsnm{Lehmann}, \binits{S.}},
\bauthor{\bsnm{Kivel{\"a}}, \binits{M.}}:
\batitle{Privacy and uniqueness of neighborhoods in social networks}.
\bjtitle{Scientific Reports}
\bvolume{11}(\bissue{1}),
\bfpage{20104}
(\byear{2021})
\end{barticle}
\endbibitem

\bibitem[\protect\citeauthoryear{de~Jong et~al.}{2024}]{dejong2023effect}
\begin{barticle}
\bauthor{\bsnm{Jong}, \binits{R.G.}},
\bauthor{\bsnm{Loo}, \binits{M.P.J.}},
\bauthor{\bsnm{Takes}, \binits{F.W.}}:
\batitle{The effect of distant connections on node anonymity in complex networks}.
\bjtitle{Scientific Reports}
\bvolume{14}(\bissue{1}),
\bfpage{1156}
(\byear{2024})
\end{barticle}
\endbibitem

\bibitem[\protect\citeauthoryear{Wilkinson et~al.}{2016}]{wilkinson2016fair}
\begin{barticle}
\bauthor{\bsnm{Wilkinson}, \binits{M.D.}},
\bauthor{\bsnm{Dumontier}, \binits{M.}},
\bauthor{\bsnm{Aalbersberg}, \binits{I.J.}},
\bauthor{\bsnm{Appleton}, \binits{G.}},
\bauthor{\bsnm{Axton}, \binits{M.}},
\bauthor{\bsnm{Baak}, \binits{A.}},
\bauthor{\bsnm{Blomberg}, \binits{N.}},
\bauthor{\bsnm{Boiten}, \binits{J.-W.}},
\bauthor{\bsnm{Silva~Santos}, \binits{L.B.}},
\bauthor{\bsnm{Bourne}, \binits{P.E.}}, \betal:
\batitle{The fair guiding principles for scientific data management and stewardship}.
\bjtitle{Scientific data}
\bvolume{3}(\bissue{1}),
\bfpage{1}--\blpage{9}
(\byear{2016})
\end{barticle}
\endbibitem

\bibitem[\protect\citeauthoryear{Neal et~al.}{2024}]{neal2024recommendations}
\begin{barticle}
\bauthor{\bsnm{Neal}, \binits{Z.P.}},
\bauthor{\bsnm{Almquist}, \binits{Z.W.}},
\bauthor{\bsnm{Bagrow}, \binits{J.}},
\bauthor{\bsnm{Clauset}, \binits{A.}},
\bauthor{\bsnm{Diesner}, \binits{J.}},
\bauthor{\bsnm{Lazega}, \binits{E.}},
\bauthor{\bsnm{Lovato}, \binits{J.}},
\bauthor{\bsnm{Moody}, \binits{J.}},
\bauthor{\bsnm{Peixoto}, \binits{T.P.}},
\bauthor{\bsnm{Steinert-Threlkeld}, \binits{Z.}}, \betal:
\batitle{Recommendations for sharing network data and materials}.
\bjtitle{Network Science}
\bvolume{12}(\bissue{4}),
\bfpage{404}--\blpage{417}
(\byear{2024})
\end{barticle}
\endbibitem

\bibitem[\protect\citeauthoryear{Jiang et~al.}{2021}]{jiang2021applications}
\begin{barticle}
\bauthor{\bsnm{Jiang}, \binits{H.}},
\bauthor{\bsnm{Pei}, \binits{J.}},
\bauthor{\bsnm{Yu}, \binits{D.}},
\bauthor{\bsnm{Yu}, \binits{J.}},
\bauthor{\bsnm{Gong}, \binits{B.}},
\bauthor{\bsnm{Cheng}, \binits{X.}}:
\batitle{Applications of differential privacy in social network analysis: A survey}.
\bjtitle{IEEE Transactions on Knowledge and Data Engineering}
\bvolume{35}(\bissue{1}),
\bfpage{108}--\blpage{127}
(\byear{2021})
\end{barticle}
\endbibitem

\bibitem[\protect\citeauthoryear{Wang and Wu}{2013}]{wang2013preserving}
\begin{barticle}
\bauthor{\bsnm{Wang}, \binits{Y.}},
\bauthor{\bsnm{Wu}, \binits{X.}}:
\batitle{Preserving differential privacy in degree-correlation based graph generation}.
\bjtitle{Transactions on Data Privacy}
\bvolume{6}(\bissue{2}),
\bfpage{127}
(\byear{2013})
\end{barticle}
\endbibitem

\bibitem[\protect\citeauthoryear{Hundepool et~al.}{2012}]{hundepool2012statistical}
\begin{bbook}
\bauthor{\bsnm{Hundepool}, \binits{A.}},
\bauthor{\bsnm{Domingo-Ferrer}, \binits{J.}},
\bauthor{\bsnm{Franconi}, \binits{L.}},
\bauthor{\bsnm{Giessing}, \binits{S.}},
\bauthor{\bsnm{Nordholt}, \binits{E.S.}},
\bauthor{\bsnm{Spicer}, \binits{K.}},
\bauthor{\bsnm{De~Wolf}, \binits{P.-P.}}:
\bbtitle{Statistical Disclosure Control}
vol. \bseriesno{2},
(\byear{2012})
\end{bbook}
\endbibitem

\bibitem[\protect\citeauthoryear{Willenborg and de~Waal}{2001}]{willenborg2012elements}
\begin{bbook}
\bauthor{\bsnm{Willenborg}, \binits{L.}},
\bauthor{\bsnm{Waal}, \binits{T.}}:
\bbtitle{Elements of Statistical Disclosure Control}
vol. \bseriesno{155},
(\byear{2001})
\end{bbook}
\endbibitem

\bibitem[\protect\citeauthoryear{Snoke et~al.}{2018}]{snoke2018general}
\begin{barticle}
\bauthor{\bsnm{Snoke}, \binits{J.}},
\bauthor{\bsnm{Raab}, \binits{G.M.}},
\bauthor{\bsnm{Nowok}, \binits{B.}},
\bauthor{\bsnm{Dibben}, \binits{C.}},
\bauthor{\bsnm{Slavkovic}, \binits{A.}}:
\batitle{General and specific utility measures for synthetic data}.
\bjtitle{Journal of the Royal Statistical Society Series A: Statistics in Society}
\bvolume{181}(\bissue{3}),
\bfpage{663}--\blpage{688}
(\byear{2018})
\end{barticle}
\endbibitem

\bibitem[\protect\citeauthoryear{Dwork}{2008}]{dwork_differential_2008}
\begin{bchapter}
\bauthor{\bsnm{Dwork}, \binits{C.}}:
\bctitle{Differential {Privacy}: {A} {Survey} of {Results}}.
In: \bbtitle{Theory and {Applications} of {Models} of {Computation}},
\bconflocation{Berlin, Heidelberg},
pp. \bfpage{1}--\blpage{19}
(\byear{2008})
\end{bchapter}
\endbibitem

\bibitem[\protect\citeauthoryear{Dwork et~al.}{2006}]{dwork2006calibrating}
\begin{bchapter}
\bauthor{\bsnm{Dwork}, \binits{C.}},
\bauthor{\bsnm{McSherry}, \binits{F.}},
\bauthor{\bsnm{Nissim}, \binits{K.}},
\bauthor{\bsnm{Smith}, \binits{A.}}:
\bctitle{Calibrating noise to sensitivity in private data analysis}.
In: \bbtitle{Theory of Cryptography},
\bconflocation{Berlin, Heidelberg},
pp. \bfpage{265}--\blpage{284}
(\byear{2006})
\end{bchapter}
\endbibitem

\bibitem[\protect\citeauthoryear{Dwork}{2006}]{dwork2006differential}
\begin{bchapter}
\bauthor{\bsnm{Dwork}, \binits{C.}}:
\bctitle{Differential privacy}.
In: \bbtitle{Automata, Languages and Programming},
\bconflocation{Berlin, Heidelberg},
pp. \bfpage{1}--\blpage{12}
(\byear{2006})
\end{bchapter}
\endbibitem

\bibitem[\protect\citeauthoryear{Sweeney}{2002}]{sweeney2002k}
\begin{barticle}
\bauthor{\bsnm{Sweeney}, \binits{L.}}:
\batitle{k-anonymity: A model for protecting privacy}.
\bjtitle{International journal of uncertainty, fuzziness and knowledge-based systems}
\bvolume{10}(\bissue{05}),
\bfpage{557}--\blpage{570}
(\byear{2002})
\end{barticle}
\endbibitem

\bibitem[\protect\citeauthoryear{Machanavajjhala et~al.}{2007}]{machanavajjhala_l-diversity_2007}
\begin{botherref}
\oauthor{\bsnm{Machanavajjhala}, \binits{A.}},
\oauthor{\bsnm{Kifer}, \binits{D.}},
\oauthor{\bsnm{Gehrke}, \binits{J.}},
\oauthor{\bsnm{Venkitasubramaniam}, \binits{M.}}:
L-diversity: {Privacy} beyond k-anonymity.
ACM Transactions on Knowledge Discovery from Data
\textbf{1}(1)
(2007)
\end{botherref}
\endbibitem

\bibitem[\protect\citeauthoryear{Li et~al.}{2007}]{li_t-closeness_2007}
\begin{bchapter}
\bauthor{\bsnm{Li}, \binits{N.}},
\bauthor{\bsnm{Li}, \binits{T.}},
\bauthor{\bsnm{Venkatasubramanian}, \binits{S.}}:
\bctitle{t-{Closeness}: {Privacy} {Beyond} k-{Anonymity} and l-{Diversity}}.
In: \bbtitle{2007 {IEEE} 23rd {International} {Conference} on {Data} {Engineering}},
pp. \bfpage{106}--\blpage{115}
(\byear{2007}).
\bcomment{ISSN: 2375-026X}
\end{bchapter}
\endbibitem

\bibitem[\protect\citeauthoryear{Wong et~al.}{2006}]{wong2006alpha}
\begin{bchapter}
\bauthor{\bsnm{Wong}, \binits{R.C.-W.}},
\bauthor{\bsnm{Li}, \binits{J.}},
\bauthor{\bsnm{Fu}, \binits{A.W.-C.}},
\bauthor{\bsnm{Wang}, \binits{K.}}:
\bctitle{($\alpha$, k)-anonymity: an enhanced k-anonymity model for privacy preserving data publishing}.
In: \bbtitle{Proceedings of the 12th ACM SIGKDD International Conference on Knowledge Discovery and Data Mining},
pp. \bfpage{754}--\blpage{759}
(\byear{2006})
\end{bchapter}
\endbibitem

\bibitem[\protect\citeauthoryear{Drechsler}{2011}]{drechsler2011synthetic}
\begin{bbook}
\bauthor{\bsnm{Drechsler}, \binits{J.}}:
\bbtitle{Synthetic Datasets for Statistical Disclosure Control: Theory and Implementation}
vol. \bseriesno{201}.
\bpublisher{Springer},
\blocation{New York, NY}
(\byear{2011})
\end{bbook}
\endbibitem

\bibitem[\protect\citeauthoryear{Xiao and Tao}{2006}]{xiao2006anatomy}
\begin{bchapter}
\bauthor{\bsnm{Xiao}, \binits{X.}},
\bauthor{\bsnm{Tao}, \binits{Y.}}:
\bctitle{Anatomy: simple and effective privacy preservation}.
In: \bbtitle{Proceedings of the 32nd International Conference on Very Large Data Bases},
pp. \bfpage{139}--\blpage{150}
(\byear{2006})
\end{bchapter}
\endbibitem

\bibitem[\protect\citeauthoryear{Sala et~al.}{2011}]{sala2011sharing}
\begin{bchapter}
\bauthor{\bsnm{Sala}, \binits{A.}},
\bauthor{\bsnm{Zhao}, \binits{X.}},
\bauthor{\bsnm{Wilson}, \binits{C.}},
\bauthor{\bsnm{Zheng}, \binits{H.}},
\bauthor{\bsnm{Zhao}, \binits{B.Y.}}:
\bctitle{Sharing graphs using differentially private graph models}.
In: \bbtitle{Proceedings of the 2011 ACM SIGCOMM Conference on Internet Measurement Conference},
pp. \bfpage{81}--\blpage{98}
(\byear{2011})
\end{bchapter}
\endbibitem

\bibitem[\protect\citeauthoryear{Proserpio et~al.}{2014}]{proserpio2012calibrating}
\begin{barticle}
\bauthor{\bsnm{Proserpio}, \binits{D.}},
\bauthor{\bsnm{Goldberg}, \binits{S.}},
\bauthor{\bsnm{McSherry}, \binits{F.}}:
\batitle{Calibrating data to sensitivity in private data analysis: A platform for differentially-private analysis of weighted datasets}.
\bjtitle{Proceedings of the VLDB Endowment}
\bvolume{7}(\bissue{8}),
\bfpage{637}--\blpage{648}
(\byear{2014})
\end{barticle}
\endbibitem

\bibitem[\protect\citeauthoryear{Liu and Terzi}{2008}]{liu2008towards}
\begin{bchapter}
\bauthor{\bsnm{Liu}, \binits{K.}},
\bauthor{\bsnm{Terzi}, \binits{E.}}:
\bctitle{Towards identity anonymization on graphs}.
In: \bbtitle{Proceedings of the ACM SIGMOD International Conference on Management of Data},
pp. \bfpage{93}--\blpage{106}
(\byear{2008})
\end{bchapter}
\endbibitem

\bibitem[\protect\citeauthoryear{Hay et~al.}{2008}]{hay2008resisting}
\begin{bchapter}
\bauthor{\bsnm{Hay}, \binits{M.}},
\bauthor{\bsnm{Miklau}, \binits{G.}},
\bauthor{\bsnm{Jensen}, \binits{D.}},
\bauthor{\bsnm{Towsley}, \binits{D.}},
\bauthor{\bsnm{Weis}, \binits{P.}}:
\bctitle{Resisting structural re-identification in anonymized social networks}.
In: \bbtitle{Proceedings of the VLDB Endowment},
vol. \bseriesno{1},
pp. \bfpage{102}--\blpage{114}
(\byear{2008})
\end{bchapter}
\endbibitem

\bibitem[\protect\citeauthoryear{Zou et~al.}{2009}]{zou2009k}
\begin{bchapter}
\bauthor{\bsnm{Zou}, \binits{L.}},
\bauthor{\bsnm{Chen}, \binits{L.}},
\bauthor{\bsnm{Özsu}, \binits{M.T.}}:
\bctitle{K-automorphism: a general framework for privacy preserving network publication}.
In: \bbtitle{Proceedings of the of the VLDB Endowment},
vol. \bseriesno{2},
pp. \bfpage{946}--\blpage{957}
(\byear{2009})
\end{bchapter}
\endbibitem

\bibitem[\protect\citeauthoryear{Ying and Wu}{2008}]{ying2008randomizing}
\begin{bchapter}
\bauthor{\bsnm{Ying}, \binits{X.}},
\bauthor{\bsnm{Wu}, \binits{X.}}:
\bctitle{Randomizing social networks: a spectrum preserving approach}.
In: \bbtitle{Proceedings of the 2008 SIAM International Conference on Data Mining},
pp. \bfpage{739}--\blpage{750}
(\byear{2008})
\end{bchapter}
\endbibitem

\bibitem[\protect\citeauthoryear{Liu et~al.}{2016}]{liu2016smartwalk}
\begin{bchapter}
\bauthor{\bsnm{Liu}, \binits{Y.}},
\bauthor{\bsnm{Ji}, \binits{S.}},
\bauthor{\bsnm{Mittal}, \binits{P.}}:
\bctitle{Smartwalk: Enhancing social network security via adaptive random walks}.
In: \bbtitle{Proceedings of the 2016 ACM SIGSAC Conference on Computer and Communications Security},
pp. \bfpage{492}--\blpage{503}
(\byear{2016})
\end{bchapter}
\endbibitem

\bibitem[\protect\citeauthoryear{Mittal et~al.}{2012}]{mittal2012preserving}
\begin{botherref}
\oauthor{\bsnm{Mittal}, \binits{P.}},
\oauthor{\bsnm{Papamanthou}, \binits{C.}},
\oauthor{\bsnm{Song}, \binits{D.}}:
Preserving link privacy in social network based systems.
arXiv preprint arXiv:1208.6189
(2012)
\end{botherref}
\endbibitem

\bibitem[\protect\citeauthoryear{Campan and Truta}{2009}]{campan2008data}
\begin{bchapter}
\bauthor{\bsnm{Campan}, \binits{A.}},
\bauthor{\bsnm{Truta}, \binits{T.M.}}:
\bctitle{Data and structural k-anonymity in social networks}.
In: \bbtitle{Privacy, Security, and Trust in KDD},
\bconflocation{Berlin, Heidelberg},
pp. \bfpage{33}--\blpage{54}
(\byear{2009})
\end{bchapter}
\endbibitem

\bibitem[\protect\citeauthoryear{Bhagat et~al.}{2009}]{bhagat2009class}
\begin{bchapter}
\bauthor{\bsnm{Bhagat}, \binits{S.}},
\bauthor{\bsnm{Cormode}, \binits{G.}},
\bauthor{\bsnm{Krishnamurthy}, \binits{B.}},
\bauthor{\bsnm{Srivastava}, \binits{D.}}:
\bctitle{Class-based graph anonymization for social network data}.
In: \bbtitle{Proceedings of the VLDB Endowment},
vol. \bseriesno{2},
pp. \bfpage{766}--\blpage{777}
(\byear{2009})
\end{bchapter}
\endbibitem

\bibitem[\protect\citeauthoryear{Liu and Mittal}{2016}]{liu2016linkmirage}
\begin{bchapter}
\bauthor{\bsnm{Liu}, \binits{C.}},
\bauthor{\bsnm{Mittal}, \binits{P.}}:
\bctitle{Linkmirage: Enabling privacy-preserving analytics on social relationships}.
In: \bbtitle{Proceedings of the 23rd Annual Network and Distributed System Security Symposium}
(\byear{2016})
\end{bchapter}
\endbibitem

\bibitem[\protect\citeauthoryear{Yazdanjue et~al.}{2020}]{yazdanjue2020evolutionary}
\begin{barticle}
\bauthor{\bsnm{Yazdanjue}, \binits{N.}},
\bauthor{\bsnm{Fathian}, \binits{M.}},
\bauthor{\bsnm{Amiri}, \binits{B.}}:
\batitle{Evolutionary algorithms for k-anonymity in social networks based on clustering approach}.
\bjtitle{The Computer Journal}
\bvolume{63}(\bissue{7}),
\bfpage{1039}--\blpage{1062}
(\byear{2020})
\end{barticle}
\endbibitem

\bibitem[\protect\citeauthoryear{Ford et~al.}{2009}]{ford2009p}
\begin{barticle}
\bauthor{\bsnm{Ford}, \binits{R.}},
\bauthor{\bsnm{Truta}, \binits{T.M.}},
\bauthor{\bsnm{Campan}, \binits{A.}}:
\batitle{P-sensitive k-anonymity for social networks}.
\bjtitle{Proceedings of the 5th International Conference on Data Mining}
\bvolume{9},
\bfpage{403}--\blpage{409}
(\byear{2009})
\end{barticle}
\endbibitem

\bibitem[\protect\citeauthoryear{Wang et~al.}{2014}]{wang2014high}
\begin{barticle}
\bauthor{\bsnm{Wang}, \binits{Y.}},
\bauthor{\bsnm{Xie}, \binits{L.}},
\bauthor{\bsnm{Zheng}, \binits{B.}},
\bauthor{\bsnm{Lee}, \binits{K.C.}}:
\batitle{High utility k-anonymization for social network publishing}.
\bjtitle{Knowledge and Information Systems}
\bvolume{41}(\bissue{3}),
\bfpage{697}--\blpage{725}
(\byear{2014})
\end{barticle}
\endbibitem

\bibitem[\protect\citeauthoryear{Ji et~al.}{2016}]{ji2016graph}
\begin{barticle}
\bauthor{\bsnm{Ji}, \binits{S.}},
\bauthor{\bsnm{Mittal}, \binits{P.}},
\bauthor{\bsnm{Beyah}, \binits{R.}}:
\batitle{Graph data anonymization, de-anonymization attacks, and de-anonymizability quantification: A survey}.
\bjtitle{IEEE Communications Surveys \& Tutorials}
\bvolume{19}(\bissue{2}),
\bfpage{1305}--\blpage{1326}
(\byear{2016})
\end{barticle}
\endbibitem

\bibitem[\protect\citeauthoryear{Zhou et~al.}{2008}]{zhou2008brief}
\begin{barticle}
\bauthor{\bsnm{Zhou}, \binits{B.}},
\bauthor{\bsnm{Pei}, \binits{J.}},
\bauthor{\bsnm{Luk}, \binits{W.}}:
\batitle{A brief survey on anonymization techniques for privacy preserving publishing of social network data}.
\bjtitle{{SIGKDD} Explor.}
\bvolume{10}(\bissue{2}),
\bfpage{12}--\blpage{22}
(\byear{2008})
\end{barticle}
\endbibitem

\bibitem[\protect\citeauthoryear{Casas-Roma et~al.}{2017}]{casas2017survey}
\begin{barticle}
\bauthor{\bsnm{Casas-Roma}, \binits{J.}},
\bauthor{\bsnm{Herrera-Joancomart{\'\i}}, \binits{J.}},
\bauthor{\bsnm{Torra}, \binits{V.}}:
\batitle{A survey of graph-modification techniques for privacy-preserving on networks}.
\bjtitle{Artificial Intelligence Review}
\bvolume{47},
\bfpage{341}--\blpage{366}
(\byear{2017})
\end{barticle}
\endbibitem

\bibitem[\protect\citeauthoryear{Beigi and Liu}{2020}]{beigi2020survey}
\begin{barticle}
\bauthor{\bsnm{Beigi}, \binits{G.}},
\bauthor{\bsnm{Liu}, \binits{H.}}:
\batitle{A survey on privacy in social media: Identification, mitigation, and applications}.
\bjtitle{ACM Transactions on Data Science}
\bvolume{1}(\bissue{1}),
\bfpage{1}--\blpage{38}
(\byear{2020})
\end{barticle}
\endbibitem

\bibitem[\protect\citeauthoryear{Yazdanjue et~al.}{2025}]{yazdanjue2025comprehensive}
\begin{botherref}
\oauthor{\bsnm{Yazdanjue}, \binits{N.}},
\oauthor{\bsnm{Yazdanjouei}, \binits{H.}},
\oauthor{\bsnm{Gharoun}, \binits{H.}},
\oauthor{\bsnm{Khorshidi}, \binits{M.S.}},
\oauthor{\bsnm{Rakhshaninejad}, \binits{M.}},
\oauthor{\bsnm{Amiri}, \binits{B.}},
\oauthor{\bsnm{Gandomi}, \binits{A.H.}}:
A comprehensive bibliometric analysis on social network anonymization: current approaches and future directions.
Knowledge and Information Systems,
1--80
(2025)
\end{botherref}
\endbibitem

\bibitem[\protect\citeauthoryear{Hay et~al.}{2010}]{hay2009boosting}
\begin{barticle}
\bauthor{\bsnm{Hay}, \binits{M.}},
\bauthor{\bsnm{Rastogi}, \binits{V.}},
\bauthor{\bsnm{Miklau}, \binits{G.}},
\bauthor{\bsnm{Suciu}, \binits{D.}}:
\batitle{Boosting the accuracy of differentially private histograms through consistency}.
\bjtitle{Proceedings of the VLDB Endowment}
\bvolume{3}(\bissue{1–2}),
\bfpage{1021}--\blpage{1032}
(\byear{2010})
\end{barticle}
\endbibitem

\bibitem[\protect\citeauthoryear{Macwan and Patel}{2018}]{macwan2018node}
\begin{barticle}
\bauthor{\bsnm{Macwan}, \binits{K.R.}},
\bauthor{\bsnm{Patel}, \binits{S.J.}}:
\batitle{Node differential privacy in social graph degree publishing}.
\bjtitle{Procedia Computer Science}
\bvolume{143},
\bfpage{786}--\blpage{793}
(\byear{2018})
\end{barticle}
\endbibitem

\bibitem[\protect\citeauthoryear{Yang et~al.}{2020}]{yang2020secure}
\begin{botherref}
\oauthor{\bsnm{Yang}, \binits{C.}},
\oauthor{\bsnm{Wang}, \binits{H.}},
\oauthor{\bsnm{Zhang}, \binits{K.}},
\oauthor{\bsnm{Chen}, \binits{L.}},
\oauthor{\bsnm{Sun}, \binits{L.}}:
Secure deep graph generation with link differential privacy.
Proceedings of the Thirtieth International Joint Conference on Artificial Intelligence
(2020)
\end{botherref}
\endbibitem

\bibitem[\protect\citeauthoryear{Boldi et~al.}{2012}]{boldi2012injecting}
\begin{bchapter}
\bauthor{\bsnm{Boldi}, \binits{P.}},
\bauthor{\bsnm{Bonchi}, \binits{F.}},
\bauthor{\bsnm{Gionis}, \binits{A.}},
\bauthor{\bsnm{Tassa}, \binits{T.}}:
\bctitle{Injecting uncertainty in graphs for identity obfuscation},
vol. \bseriesno{5},
pp. \bfpage{1376}--\blpage{1387}
(\byear{2012})
\end{bchapter}
\endbibitem

\bibitem[\protect\citeauthoryear{Jian et~al.}{2023}]{jian2021publishing}
\begin{barticle}
\bauthor{\bsnm{Jian}, \binits{X.}},
\bauthor{\bsnm{Wang}, \binits{Y.}},
\bauthor{\bsnm{Chen}, \binits{L.}}:
\batitle{Publishing graphs under node differential privacy}.
\bjtitle{IEEE Transactions on Knowledge and Data Engineering}
\bvolume{35}(\bissue{4}),
\bfpage{4164}--\blpage{4177}
(\byear{2023})
\end{barticle}
\endbibitem

\bibitem[\protect\citeauthoryear{Barab{\'a}si and Albert}{1999}]{barabasi1999emergence}
\begin{barticle}
\bauthor{\bsnm{Barab{\'a}si}, \binits{A.L.}},
\bauthor{\bsnm{Albert}, \binits{R.}}:
\batitle{Emergence of scaling in random networks}.
\bjtitle{Science}
\bvolume{286}(\bissue{5439}),
\bfpage{509}--\blpage{512}
(\byear{1999})
\end{barticle}
\endbibitem

\bibitem[\protect\citeauthoryear{Watts and Strogatz}{1998}]{watts1998collective}
\begin{barticle}
\bauthor{\bsnm{Watts}, \binits{D.J.}},
\bauthor{\bsnm{Strogatz}, \binits{S.H.}}:
\batitle{Collective dynamics of ‘small-world’ networks}.
\bjtitle{Nature}
\bvolume{393}(\bissue{6684}),
\bfpage{440}--\blpage{442}
(\byear{1998})
\end{barticle}
\endbibitem

\bibitem[\protect\citeauthoryear{Mahadevan et~al.}{2006}]{mahadevan2006systematic}
\begin{barticle}
\bauthor{\bsnm{Mahadevan}, \binits{P.}},
\bauthor{\bsnm{Krioukov}, \binits{D.}},
\bauthor{\bsnm{Fall}, \binits{K.}},
\bauthor{\bsnm{Vahdat}, \binits{A.}}:
\batitle{Systematic topology analysis and generation using degree correlations}.
\bjtitle{ACM SIGCOMM Computer Communication Review}
\bvolume{36}(\bissue{4}),
\bfpage{135}--\blpage{146}
(\byear{2006})
\end{barticle}
\endbibitem

\bibitem[\protect\citeauthoryear{Horawalavithana and Iamnitchi}{2019}]{horawalavithana2019privacy}
\begin{botherref}
\oauthor{\bsnm{Horawalavithana}, \binits{S.}},
\oauthor{\bsnm{Iamnitchi}, \binits{A.}}:
On the privacy of dk-random graphs.
CoRR
\textbf{abs/1907.01695}
(2019)
{\href{https://arxiv.org/abs/1907.01695}{{1907.01695}}}
\end{botherref}
\endbibitem

\bibitem[\protect\citeauthoryear{Lu and Miklau}{2014}]{lu2014exponential}
\begin{bchapter}
\bauthor{\bsnm{Lu}, \binits{W.}},
\bauthor{\bsnm{Miklau}, \binits{G.}}:
\bctitle{Exponential random graph estimation under differential privacy}.
In: \bbtitle{Proceedings of the 20th ACM SIGKDD International Conference on Knowledge Discovery and Data Mining},
pp. \bfpage{921}--\blpage{930}
(\byear{2014})
\end{bchapter}
\endbibitem

\bibitem[\protect\citeauthoryear{Demelius et~al.}{2025}]{demelius2025recent}
\begin{barticle}
\bauthor{\bsnm{Demelius}, \binits{L.}},
\bauthor{\bsnm{Kern}, \binits{R.}},
\bauthor{\bsnm{Tr{\"u}gler}, \binits{A.}}:
\batitle{Recent advances of differential privacy in centralized deep learning: A systematic survey}.
\bjtitle{ACM Computing Surveys}
\bvolume{57}(\bissue{6}),
\bfpage{1}--\blpage{28}
(\byear{2025})
\end{barticle}
\endbibitem

\bibitem[\protect\citeauthoryear{Barabasi and Oltvai}{2004}]{barabasi2004network}
\begin{barticle}
\bauthor{\bsnm{Barabasi}, \binits{A.-L.}},
\bauthor{\bsnm{Oltvai}, \binits{Z.N.}}:
\batitle{Network biology: understanding the cell's functional organization}.
\bjtitle{Nature Reviews Genetics}
\bvolume{5}(\bissue{2}),
\bfpage{101}--\blpage{113}
(\byear{2004})
\end{barticle}
\endbibitem

\bibitem[\protect\citeauthoryear{McKay and Piperno}{2014}]{nauty}
\begin{barticle}
\bauthor{\bsnm{McKay}, \binits{B.D.}},
\bauthor{\bsnm{Piperno}, \binits{A.}}:
\batitle{Practical graph isomorphism, ii}.
\bjtitle{Journal of Symbolic Computation}
\bvolume{60},
\bfpage{94}--\blpage{112}
(\byear{2014})
\end{barticle}
\endbibitem

\bibitem[\protect\citeauthoryear{Traag et~al.}{2019}]{traag2019louvain}
\begin{barticle}
\bauthor{\bsnm{Traag}, \binits{V.A.}},
\bauthor{\bsnm{Waltman}, \binits{L.}},
\bauthor{\bsnm{Van~Eck}, \binits{N.J.}}:
\batitle{From louvain to leiden: guaranteeing well-connected communities}.
\bjtitle{Scientific Reports}
\bvolume{9}(\bissue{1}),
\bfpage{1}--\blpage{12}
(\byear{2019})
\end{barticle}
\endbibitem

\bibitem[\protect\citeauthoryear{Lancichinetti and Fortunato}{2012}]{lancichinetti2012consensus}
\begin{barticle}
\bauthor{\bsnm{Lancichinetti}, \binits{A.}},
\bauthor{\bsnm{Fortunato}, \binits{S.}}:
\batitle{Consensus clustering in complex networks}.
\bjtitle{Scientific Reports}
\bvolume{2}(\bissue{1}),
\bfpage{336}
(\byear{2012})
\end{barticle}
\endbibitem

\bibitem[\protect\citeauthoryear{Brandes}{2001}]{brandes2001faster}
\begin{barticle}
\bauthor{\bsnm{Brandes}, \binits{U.}}:
\batitle{A faster algorithm for betweenness centrality}.
\bjtitle{Journal of mathematical sociology}
\bvolume{25}(\bissue{2}),
\bfpage{163}--\blpage{177}
(\byear{2001})
\end{barticle}
\endbibitem

\bibitem[\protect\citeauthoryear{Zhou and Pei}{2011}]{zhou2011k}
\begin{barticle}
\bauthor{\bsnm{Zhou}, \binits{B.}},
\bauthor{\bsnm{Pei}, \binits{J.}}:
\batitle{The k-anonymity and l-diversity approaches for privacy preservation in social networks against neighborhood attacks}.
\bjtitle{Knowledge and Information Systems}
\bvolume{28}(\bissue{1}),
\bfpage{47}--\blpage{77}
(\byear{2011})
\end{barticle}
\endbibitem

\bibitem[\protect\citeauthoryear{Tripathy and Mitra}{2012}]{tripathy2012algorithm}
\begin{bchapter}
\bauthor{\bsnm{Tripathy}, \binits{B.}},
\bauthor{\bsnm{Mitra}, \binits{A.}}:
\bctitle{An algorithm to achieve k-anonymity and l-diversity anonymisation in social networks}.
In: \bbtitle{2012 Fourth International Conference on Computational Aspects of Social Networks},
pp. \bfpage{126}--\blpage{131}
(\byear{2012})
\end{bchapter}
\endbibitem

\bibitem[\protect\citeauthoryear{Ren et~al.}{2022}]{ren2022personalized}
\begin{botherref}
\oauthor{\bsnm{Ren}, \binits{X.}},
\oauthor{\bsnm{Jiang}, \binits{D.}}, et al.:
A personalized-anonymity model of social network for protecting privacy.
Wireless Communications and Mobile Computing
(2022)
\end{botherref}
\endbibitem

\bibitem[\protect\citeauthoryear{Yuan et~al.}{2010}]{yuan2010personalized}
\begin{bchapter}
\bauthor{\bsnm{Yuan}, \binits{M.}},
\bauthor{\bsnm{Chen}, \binits{L.}},
\bauthor{\bsnm{Yu}, \binits{P.S.}}:
\bctitle{Personalized privacy protection in social networks}.
In: \bbtitle{Proceedings of the VLDB Endowment},
vol. \bseriesno{4},
pp. \bfpage{141}--\blpage{150}
(\byear{2010})
\end{bchapter}
\endbibitem

\bibitem[\protect\citeauthoryear{Hao et~al.}{2014}]{hao2014k}
\begin{barticle}
\bauthor{\bsnm{Hao}, \binits{Y.}},
\bauthor{\bsnm{Cao}, \binits{H.}},
\bauthor{\bsnm{Hu}, \binits{C.}},
\bauthor{\bsnm{Bhattarai}, \binits{K.}},
\bauthor{\bsnm{Misra}, \binits{S.}}:
\batitle{K-anonymity for social networks containing rich structural and textual information}.
\bjtitle{Social Network Analysis and Mining}
\bvolume{4}(\bissue{1}),
\bfpage{223}
(\byear{2014})
\end{barticle}
\endbibitem

\bibitem[\protect\citeauthoryear{Liu et~al.}{2015}]{liu2015k}
\begin{barticle}
\bauthor{\bsnm{Liu}, \binits{C.-G.}},
\bauthor{\bsnm{Liu}, \binits{I.-H.}},
\bauthor{\bsnm{Yao}, \binits{W.-S.}},
\bauthor{\bsnm{Li}, \binits{J.-S.}}:
\batitle{K-anonymity against neighborhood attacks in weighted social networks}.
\bjtitle{Security and Communication Networks}
\bvolume{8}(\bissue{18}),
\bfpage{3864}--\blpage{3882}
(\byear{2015})
\end{barticle}
\endbibitem

\bibitem[\protect\citeauthoryear{Macwan and Patel}{2017}]{macwan2017k}
\begin{barticle}
\bauthor{\bsnm{Macwan}, \binits{K.R.}},
\bauthor{\bsnm{Patel}, \binits{S.J.}}:
\batitle{k-degree anonymity model for social network data publishing}.
\bjtitle{Advances in Electrical \& Computer Engineering}
\bvolume{17}(\bissue{4}),
\bfpage{117}--\blpage{124}
(\byear{2017})
\end{barticle}
\endbibitem

\bibitem[\protect\citeauthoryear{Casas-Roma et~al.}{2013}]{casas2013algorithm}
\begin{bchapter}
\bauthor{\bsnm{Casas-Roma}, \binits{J.}},
\bauthor{\bsnm{Herrera-Joancomart{\'\i}}, \binits{J.}},
\bauthor{\bsnm{Torra}, \binits{V.}}:
\bctitle{An algorithm for k-degree anonymity on large networks}.
In: \bbtitle{Proceedings of the 2013 IEEE/ACM International Conference on Advances in Social Networks Analysis and Mining},
pp. \bfpage{671}--\blpage{675}
(\byear{2013})
\end{bchapter}
\endbibitem

\bibitem[\protect\citeauthoryear{Lu et~al.}{2012}]{lu2012fast}
\begin{bchapter}
\bauthor{\bsnm{Lu}, \binits{X.}},
\bauthor{\bsnm{Song}, \binits{Y.}},
\bauthor{\bsnm{Bressan}, \binits{S.}}:
\bctitle{Fast identity anonymization on graphs}.
In: \bbtitle{Database and Expert Systems Applications},
pp. \bfpage{281}--\blpage{295}
(\byear{2012})
\end{bchapter}
\endbibitem

\bibitem[\protect\citeauthoryear{de~Jong et~al.}{2023}]{dejong2023algorithms}
\begin{botherref}
\oauthor{\bsnm{Jong}, \binits{R.G.}},
\oauthor{\bsnm{Loo}, \binits{M.P.J.}},
\oauthor{\bsnm{Takes}, \binits{F.W.}}:
Algorithms for efficiently computing structural anonymity in complex networks.
ACM Journal of Experimental Algorithmics
\textbf{28}
(2023)
\end{botherref}
\endbibitem

\bibitem[\protect\citeauthoryear{Zhou and Pei}{2008}]{zhou2008preserving}
\begin{bchapter}
\bauthor{\bsnm{Zhou}, \binits{B.}},
\bauthor{\bsnm{Pei}, \binits{J.}}:
\bctitle{Preserving privacy in social networks against neighborhood attacks}.
In: \bbtitle{Proceedings of the 24th IEEE International Conference on Data Engineering},
pp. \bfpage{506}--\blpage{515}
(\byear{2008})
\end{bchapter}
\endbibitem

\bibitem[\protect\citeauthoryear{Alavi et~al.}{2019}]{alavi2019attacker}
\begin{bchapter}
\bauthor{\bsnm{Alavi}, \binits{A.}},
\bauthor{\bsnm{Gupta}, \binits{R.}},
\bauthor{\bsnm{Qian}, \binits{Z.}}:
\bctitle{When the attacker knows a lot: The gaga graph anonymizer}.
In: \bbtitle{Proceedings of the 21st Springer International Conference on Information Security},
pp. \bfpage{211}--\blpage{230}
(\byear{2019})
\end{bchapter}
\endbibitem

\bibitem[\protect\citeauthoryear{Hay et~al.}{2007}]{hay2007anonymizing}
\begin{botherref}
\oauthor{\bsnm{Hay}, \binits{M.}},
\oauthor{\bsnm{Miklau}, \binits{G.}},
\oauthor{\bsnm{Jensen}, \binits{D.}},
\oauthor{\bsnm{Weis}, \binits{P.}},
\oauthor{\bsnm{Srivastava}, \binits{S.}}:
Anonymizing social networks.
Computer Science Department Faculty Publication Series,
180
(2007)
\end{botherref}
\endbibitem

\bibitem[\protect\citeauthoryear{Wang et~al.}{2013}]{wang2013outsourcing}
\begin{bchapter}
\bauthor{\bsnm{Wang}, \binits{G.}},
\bauthor{\bsnm{Liu}, \binits{Q.}},
\bauthor{\bsnm{Li}, \binits{F.}},
\bauthor{\bsnm{Yang}, \binits{S.}},
\bauthor{\bsnm{Wu}, \binits{J.}}:
\bctitle{Outsourcing privacy-preserving social networks to a cloud}.
In: \bbtitle{2013 Proceedings IEEE INFOCOM},
pp. \bfpage{2886}--\blpage{2894}
(\byear{2013})
\end{bchapter}
\endbibitem

\bibitem[\protect\citeauthoryear{Hao et~al.}{2024}]{hao2024mlda}
\begin{barticle}
\bauthor{\bsnm{Hao}, \binits{Y.}},
\bauthor{\bsnm{Li}, \binits{L.}},
\bauthor{\bsnm{Chang}, \binits{L.}},
\bauthor{\bsnm{Gu}, \binits{T.}}:
\batitle{Mlda: a multi-level k-degree anonymity scheme on directed social network graphs}.
\bjtitle{Frontiers of Computer Science}
\bvolume{18}(\bissue{2}),
\bfpage{182814}
(\byear{2024})
\end{barticle}
\endbibitem

\bibitem[\protect\citeauthoryear{Rajabzadeh et~al.}{2020}]{rajabzadeh2020graph}
\begin{barticle}
\bauthor{\bsnm{Rajabzadeh}, \binits{S.}},
\bauthor{\bsnm{Shahsafi}, \binits{P.}},
\bauthor{\bsnm{Khoramnejadi}, \binits{M.}}:
\batitle{A graph modification approach for k-anonymity in social networks using the genetic algorithm}.
\bjtitle{Social Network Analysis and Mining}
\bvolume{10},
\bfpage{1}--\blpage{17}
(\byear{2020})
\end{barticle}
\endbibitem

\bibitem[\protect\citeauthoryear{Zhang et~al.}{2019}]{zhang2019large}
\begin{barticle}
\bauthor{\bsnm{Zhang}, \binits{X.}},
\bauthor{\bsnm{Liu}, \binits{J.}},
\bauthor{\bsnm{Li}, \binits{J.}},
\bauthor{\bsnm{Liu}, \binits{L.}}:
\batitle{Large-scale dynamic social network directed graph k-in\&out-degree anonymity algorithm for protecting community structure}.
\bjtitle{IEEE Access}
\bvolume{7},
\bfpage{108371}--\blpage{108383}
(\byear{2019})
\end{barticle}
\endbibitem

\bibitem[\protect\citeauthoryear{Mohapatra and Patra}{2019}]{mohapatra2019graph}
\begin{bchapter}
\bauthor{\bsnm{Mohapatra}, \binits{D.}},
\bauthor{\bsnm{Patra}, \binits{M.R.}}:
\bctitle{Graph anonymization using hierarchical clustering}.
In: \bbtitle{Computational Intelligence in Data Mining},
pp. \bfpage{145}--\blpage{154}
(\byear{2019})
\end{bchapter}
\endbibitem

\bibitem[\protect\citeauthoryear{Gao et~al.}{2018}]{gao2018resisting}
\begin{barticle}
\bauthor{\bsnm{Gao}, \binits{J.}},
\bauthor{\bsnm{Ping}, \binits{Q.}},
\bauthor{\bsnm{Wang}, \binits{J.}}:
\batitle{Resisting re-identification mining on social graph data}.
\bjtitle{World Wide Web}
\bvolume{21},
\bfpage{1759}--\blpage{1771}
(\byear{2018})
\end{barticle}
\endbibitem

\bibitem[\protect\citeauthoryear{Tai et~al.}{2011}]{tai2011privacy}
\begin{bchapter}
\bauthor{\bsnm{Tai}, \binits{C.-H.}},
\bauthor{\bsnm{Yu}, \binits{P.S.}},
\bauthor{\bsnm{Yang}, \binits{D.-N.}},
\bauthor{\bsnm{Chen}, \binits{M.-S.}}:
\bctitle{Privacy-preserving social network publication against friendship attacks}.
In: \bbtitle{Proceedings of the 17th ACM SIGKDD International Conference on Knowledge Discovery and Data Mining},
pp. \bfpage{1262}--\blpage{1270}
(\byear{2011})
\end{bchapter}
\endbibitem

\bibitem[\protect\citeauthoryear{Wu et~al.}{2010}]{wu2010k}
\begin{bchapter}
\bauthor{\bsnm{Wu}, \binits{W.}},
\bauthor{\bsnm{Xiao}, \binits{Y.}},
\bauthor{\bsnm{Wang}, \binits{W.}},
\bauthor{\bsnm{He}, \binits{Z.}},
\bauthor{\bsnm{Wang}, \binits{Z.}}:
\bctitle{k-symmetry model for identity anonymization in social networks}.
In: \bbtitle{Proceedings of the 13th International Conference on Extending Database Technology},
pp. \bfpage{111}--\blpage{122}
(\byear{2010})
\end{bchapter}
\endbibitem

\bibitem[\protect\citeauthoryear{Cheng et~al.}{2010}]{cheng2010k}
\begin{bchapter}
\bauthor{\bsnm{Cheng}, \binits{J.}},
\bauthor{\bsnm{Fu}, \binits{A.W.-c.}},
\bauthor{\bsnm{Liu}, \binits{J.}}:
\bctitle{K-isomorphism: privacy preserving network publication against structural attacks}.
In: \bbtitle{Proceedings of the 2010 ACM SIGMOD International Conference on Management of Data},
pp. \bfpage{459}--\blpage{470}
(\byear{2010})
\end{bchapter}
\endbibitem

\bibitem[\protect\citeauthoryear{Narayanan and Shmatikov}{2009}]{narayanan2009anonymizing}
\begin{bchapter}
\bauthor{\bsnm{Narayanan}, \binits{A.}},
\bauthor{\bsnm{Shmatikov}, \binits{V.}}:
\bctitle{De-anonymizing social networks}.
In: \bbtitle{Proceedings of the 30th IEEE Symposium on Security and Privacy},
pp. \bfpage{173}--\blpage{187}
(\byear{2009})
\end{bchapter}
\endbibitem

\bibitem[\protect\citeauthoryear{Fu et~al.}{2015}]{fu2015effective}
\begin{botherref}
\oauthor{\bsnm{Fu}, \binits{H.}},
\oauthor{\bsnm{Zhang}, \binits{A.}},
\oauthor{\bsnm{Xie}, \binits{X.}}:
Effective social graph deanonymization based on graph structure and descriptive information.
ACM Transactions on Intelligent Systems and Technology
\textbf{6}(4)
(2015)
\end{botherref}
\endbibitem

\bibitem[\protect\citeauthoryear{Mohapatra and Patra}{2017}]{mohapatra2017level}
\begin{barticle}
\bauthor{\bsnm{Mohapatra}, \binits{D.}},
\bauthor{\bsnm{Patra}, \binits{M.R.}}:
\batitle{A level-cut heuristic-based clustering approach for social graph anonymization}.
\bjtitle{Social Network Analysis and Mining}
\bvolume{7},
\bfpage{1}--\blpage{13}
(\byear{2017})
\end{barticle}
\endbibitem

\bibitem[\protect\citeauthoryear{Kunegis}{2013}]{kunegis2013konect}
\begin{bchapter}
\bauthor{\bsnm{Kunegis}, \binits{J.}}:
\bctitle{Konect: the {K}oblenz network collection}.
In: \bbtitle{Proceedings of the 22nd International Conference on World Wide Web},
pp. \bfpage{1343}--\blpage{1350}
(\byear{2013})
\end{bchapter}
\endbibitem

\bibitem[\protect\citeauthoryear{Sociopatterns}{2021}]{sociopatterns}
\begin{botherref}
\oauthor{\bsnm{Sociopatterns}}:
Sociopatterns: Datasets
(2021).
\url{http://www.sociopatterns.org/datasets/}
\end{botherref}
\endbibitem

\bibitem[\protect\citeauthoryear{Sapiezynski et~al.}{2019}]{sapiezynski2019copenhagen}
\begin{botherref}
\oauthor{\bsnm{Sapiezynski}, \binits{P.}},
\oauthor{\bsnm{Stopczynski}, \binits{A.}},
\oauthor{\bsnm{Lassen}, \binits{D.D.}},
\oauthor{\bsnm{Jørgensen}, \binits{S.L.}}:
The Copenhagen Networks Study interaction data. Figshare.
\url{https://doi.org/10.6084/m9.figshare.7267433.v1} (last accessed May 2022)
(2019)
\end{botherref}
\endbibitem

\bibitem[\protect\citeauthoryear{Rossi and Ahmed}{2015}]{networksrepository}
\begin{bchapter}
\bauthor{\bsnm{Rossi}, \binits{R.A.}},
\bauthor{\bsnm{Ahmed}, \binits{N.K.}}:
\bctitle{The network data repository with interactive graph analytics and visualization}.
In: \bbtitle{Proceedings of the 29th AAAI Conference on Artificial Intelligence},
pp. \bfpage{4292}--\blpage{4293}
(\byear{2015})
\end{bchapter}
\endbibitem

\bibitem[\protect\citeauthoryear{Leskovec and Krevl}{2014}]{snapnets}
\begin{botherref}
\oauthor{\bsnm{Leskovec}, \binits{J.}},
\oauthor{\bsnm{Krevl}, \binits{A.}}:
SNAP Datasets: Stanford Large Network Dataset Collection.
\url{http://snap.stanford.edu/data} (last accessed May 2022)
(2014)
\end{botherref}
\endbibitem

\bibitem[\protect\citeauthoryear{Zitnik et~al.}{2018}]{biosnapnets}
\begin{botherref}
\oauthor{\bsnm{Zitnik}, \binits{M.}},
\oauthor{\bsnm{Sosi\v{c}}, \binits{R.}},
\oauthor{\bsnm{Maheshwari}, \binits{S.}},
\oauthor{\bsnm{Leskovec}, \binits{J.}}:
BioSNAP Datasets: Stanford. Biomedical Network Dataset Collection.
\url{http://snap.stanford.edu/biodata} (last accessed May 2023)
(2018)
\end{botherref}
\endbibitem

\bibitem[\protect\citeauthoryear{Fire}{2020}]{dataforgoodlab}
\begin{botherref}
\oauthor{\bsnm{Fire}, \binits{M.}}:
Data 4 good lab.
\url{https://data4goodlab.github.io/MichaelFire/\#section3} (last accessed May 2023)
(2020)
\end{botherref}
\endbibitem

\bibitem[\protect\citeauthoryear{Csardi and Nepusz}{2006}]{igraph}
\begin{botherref}
\oauthor{\bsnm{Csardi}, \binits{G.}},
\oauthor{\bsnm{Nepusz}, \binits{T.}}:
The igraph software package for complex network research.
InterJournal Complex Systems,
1695
(2006)
\end{botherref}
\endbibitem

\end{thebibliography}

\newpage
\begin{appendices}

\section{Strictness}\label{app:order}
This appendix accompanies Section~\ref{sec:theory}. Theorem~\ref{thm:anon} states a property of strictness: if a measure is more strict, anonymization through perturbation for this measure will require at least as many operations as for a less (or equally strict) measure.
Theorem~\ref{thm:meas_order} and~\ref{thm:meashay} are used to make the ordering in Figure~\ref{fig:ordering}.

For Theorem~\ref{thm:anon} we use the notion of graph alterations which is defined as operations that modify the graph, such as edge deletion and addition, and node deletion or addition.

\begin{theorem}\label{thm:anon}
    Given a graph $G=(V, E)$ and measures $M_1 \geq M_2$. To make $G$ anonymous for $M_1$, at least as many alterations of the graph are required as for $M_2$.
\end{theorem}

\begin{proof}
	From the notion of strictness (see Definition~\ref{def:strictness} in Section~\ref{sec:theory}) it follows that if $G$ is anonymous for $M_1$, then $G$ is anonymous for $M_2$.
    This implies that if $G$ can be anonymized for $M_1$ with $X$ alterations, then it can also be anonymized for $M_2$ using the same $X$ alterations.
	  Hence, to anonymize a graph for $M_1$, one needs at least as many edge deletions as for $M_2$.
\end{proof}

\begin{theorem}\label{thm:meas_order}
    For $d \geq 1$, the following holds: 
    \begin{itemize}
        \item \textsc{degree} $\leq$ \textsc{count}$(d)$ $\leq$ \textsc{degdist}$(d)$ $\leq$ \textsc{$d$-$k$-anonymity} $\leq$ \textsc{hybrid}$(d)$
        \item \textsc{vrq}$(d)$ $\leq$ \textsc{hybrid}$(d)$
    \end{itemize}
    
    For $d \geq 2$: 
    \begin{itemize}
        \item \textsc{vrq}$(d-1)$ $\leq$ \textsc{$d$-$k$-anonymity}.
        \item \textsc{vrq}$(d-1)$ $\leq$ \textsc{vrq}$(d)$.
    \end{itemize}
\end{theorem} 

\begin{proof}
    Given a graph $G = (V, E)$ and nodes $v, w$:
    \begin{itemize}
        \item \textsc{degree} $\leq$ \textsc{count}. Since the degree of the considered node equals the number of nodes in the $1$-neighborhood minus one, the degree can be derived from \textsc{count} at $d=1$.
        As a result, two nodes equivalent according to \textsc{count} are equivalent for \textsc{degree}, ensuring the requirement: 
        $v \cong_{\textsc{count}} w \rightarrow v \cong_{\textsc{degree}} w$.
        \item  \textsc{count} $\leq$ \textsc{degdist}. Both the number of nodes and edges in the $d$-neighborhood can be derived from \textsc{degdist}.
        The number of nodes equals the number of degree values accounted for, the number of edges equals the sum of the degrees divided by two.
        Hence, 
        $v \cong_{\textsc{degdist}} w \rightarrow v \cong_{\textsc{count}} w$.
        \item \textsc{degdist} $\leq$ \textsc{$d$-$k$-anonymity}. When two nodes are equivalent for \textsc{$d$-$k$-anonymity}, the $d$-neighborhoods need to be isomorphic.
        This implies that each node in $V_{N_d}(v)$ should be mapped onto a node in $V_{N_d}(w)$ such that the isomorphism property is satisfied.
        As a result each node should be mapped onto a node with the same degree, otherwise the isomorphism can not be valid. 
        Hence we can conclude
        $v \cong_{\textsc{$d$-$k$-anonymity}} w \rightarrow v \cong_{\textsc{degdist}} w$.
        \item \textsc{vrq} $\leq$ \textsc{hybrid} and \textsc{$d$-$k$-anonymity} $\leq$ \textsc{hybrid}. This holds because \textsc{hybrid} captures both \textsc{vrq} and \textsc{$d$-$k$-anonymity}.
        \item \textsc{vrq}(d-1) $\leq$ \textsc{$d$-$k$-anonymity}, with $d \geq 2$.
        As the $d$-neighborhood of a node contains all nodes at distance $d-1$ and all edges attached to these nodes, the degree distributions of these nodes must be equal in order for the $d$-neighborhoods to be isomorphic.
        Hence $v \cong_{\textsc{$d$-$k$-anonymity}} w \rightarrow v \cong_{\textsc{vrq}(d-1)} w$.

        \item \textsc{vrq}$(d-1)$ $\leq$  \textsc{vrq}$(d)$ with $d \geq 2$ holds because the measure is computed recursively.
    \end{itemize}
\end{proof}

It is not possible to include \textsc{vrq} elsewhere in the ordering of the neighborhood based measures.
As illustrated in Table~\ref{tab:meas_overview}, the reach of \textsc{vrq} is slightly further than the $d$-neighborhood, as it accounts for the existence of connections to nodes at distance $d+1$.
At the same time, it is less complete than other measures such as \textsc{count}.
It is, however, more strict than \textsc{degree}, as the number of degrees counted by \textsc{vrq} at distance 1 equals the degree of the node.
In Theorem~\ref{thm:meashay} we prove that we can not place this measure in the order of the neighborhood based measures when using the same value for $d$.

\begin{theorem}\label{thm:meashay}
     \textsc{count}$(d)$ $\nleq$ \textsc{vrq}$(d)$ ,  \textsc{vrq}$(d)$ $\nleq$ \textsc{count}$(d)$  and   \textsc{vrq}$(d)$ $\nleq$ \textsc{count}$(d+1)$.
\end{theorem}

    \begin{figure}[h]
        \centering
        \includegraphics[width=0.9\textwidth]{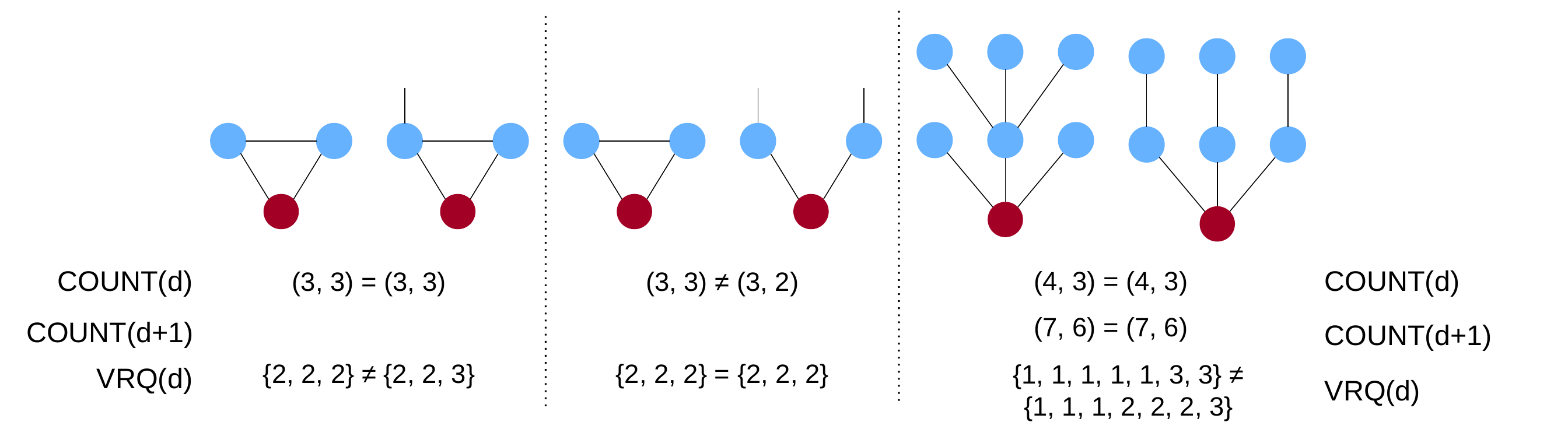}
        \caption{Counterexamples accompanying the proof for Theorem~\ref{thm:meashay}. This figure shows three examples: 1) an example where \textsc{vrq} is able to distinguish between two nodes and \textsc{count} is not (left), 2) one where \textsc{count} is able to distinguish between two nodes but \textsc{vrq} is not (middle) and  3) an example where \textsc{vrq} is able to distinguish between two nodes but \textsc{count($d$+1)} is not (right).
        }
        \label{fig:proof_meashay}
    \end{figure}

\begin{proof}\label{proof:meashay}
    We prove this theorem by use of the counterexamples in Figure~\ref{fig:proof_meashay}. 
    The figure shows one case for which \textsc{count} is able to distinguish between two nodes and \textsc{vrq} is not, an example showing the converse where \textsc{vrq} is able to distinguish between the nodes and \textsc{count} is not and an example where \textsc{vrq} is able to distinguish between two nodes and \textsc{count}$(d)$ and $d+1$ and is not.
    Hence, we can not conclude that one measure is more strict than the other.
\end{proof}

\section{Empirical differences between measures}\label{app:diff}
This appendix accompanies Section~\ref{sec:emp_comparing}, where we study the uniqueness obtained by various measures for $k$-anonymity.
Figure~\ref{fig:diff} denotes the difference in uniqueness obtained from various measures for $d=1$ and $d=2$. 
The biggest difference for all networks and distances is between \textsc{count} and \textsc{degree}, and after that \textsc{vrq} and \textsc{$d$-$k$-anonymity}.
For the other measures, the difference is often very small.

    \begin{figure}[ht]
        \centering
        \includegraphics[width=\textwidth]{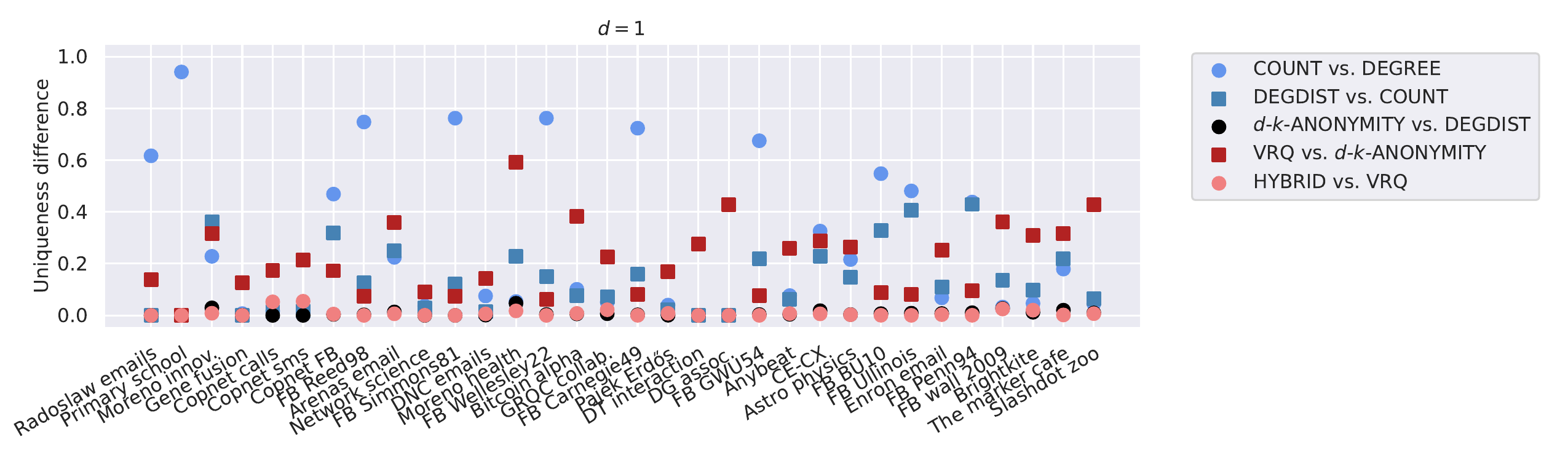}
        \includegraphics[width=\textwidth]{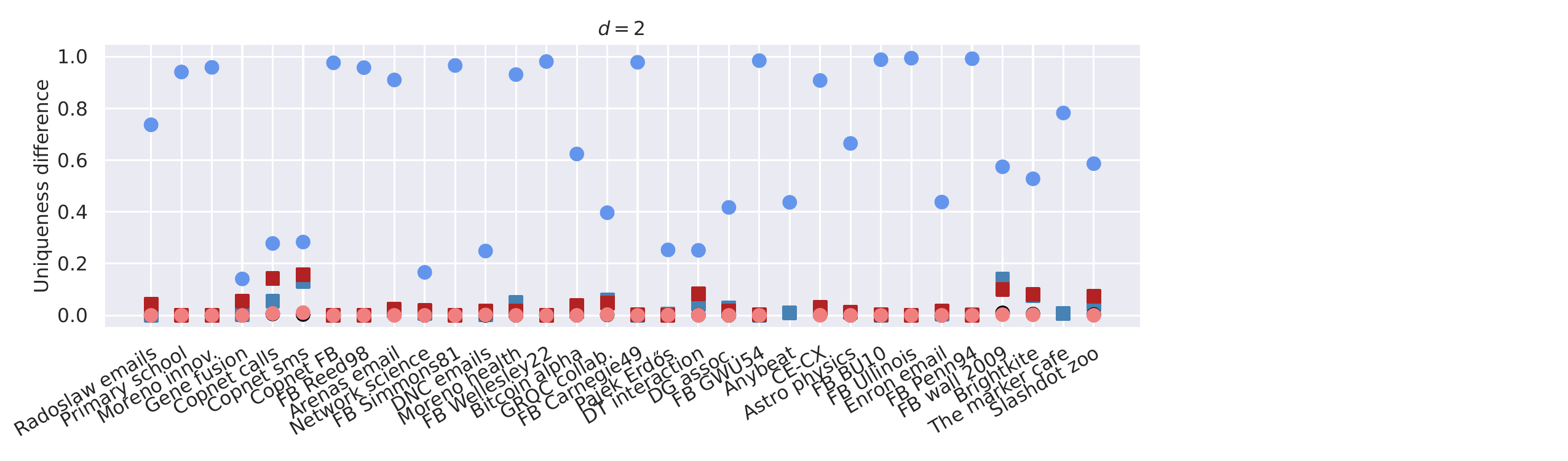}
        \caption{Difference in fraction of unique nodes (vertical axis) on different datasets (horizontal axis) using different measures for $k$-anonymity: \textsc{count} vs. \textsc{degree} (\textcolor{cornflowerblue}{blue circle}),  \textsc{degdist} vs. \textsc{count} (\textcolor{steelblue}{dark blue square}), \textsc{$d$-$k$-anonymity} vs. \textsc{degdist} (black circle), \textsc{vrq} (\textcolor{firebrick}{red square}) and \textsc{hybrid} vs. \textsc{vrq} (\textcolor{lightcoral}{pink circle}).
        }
        \label{fig:diff}
    \end{figure}

\newpage

\section{Empirical differences and network properties}\label{app:diffexplain}
This appendix accompanies Section~\ref{sec:emp_comparing} in which we study the uniqueness obtained by various measures for $k$-anonymity.
To understand which network properties cause the difference found between measures, we plot the differences reported in Figure~\ref{fig:diff} against various network properties.
The resulting can be found in Figure~\ref{fig:corr_diff} for the combinations deemed significant, as indicated in Table~\ref{tab:corr_diff}.

    \begin{figure}[ht]
           \centering
           \includegraphics[width=\textwidth]{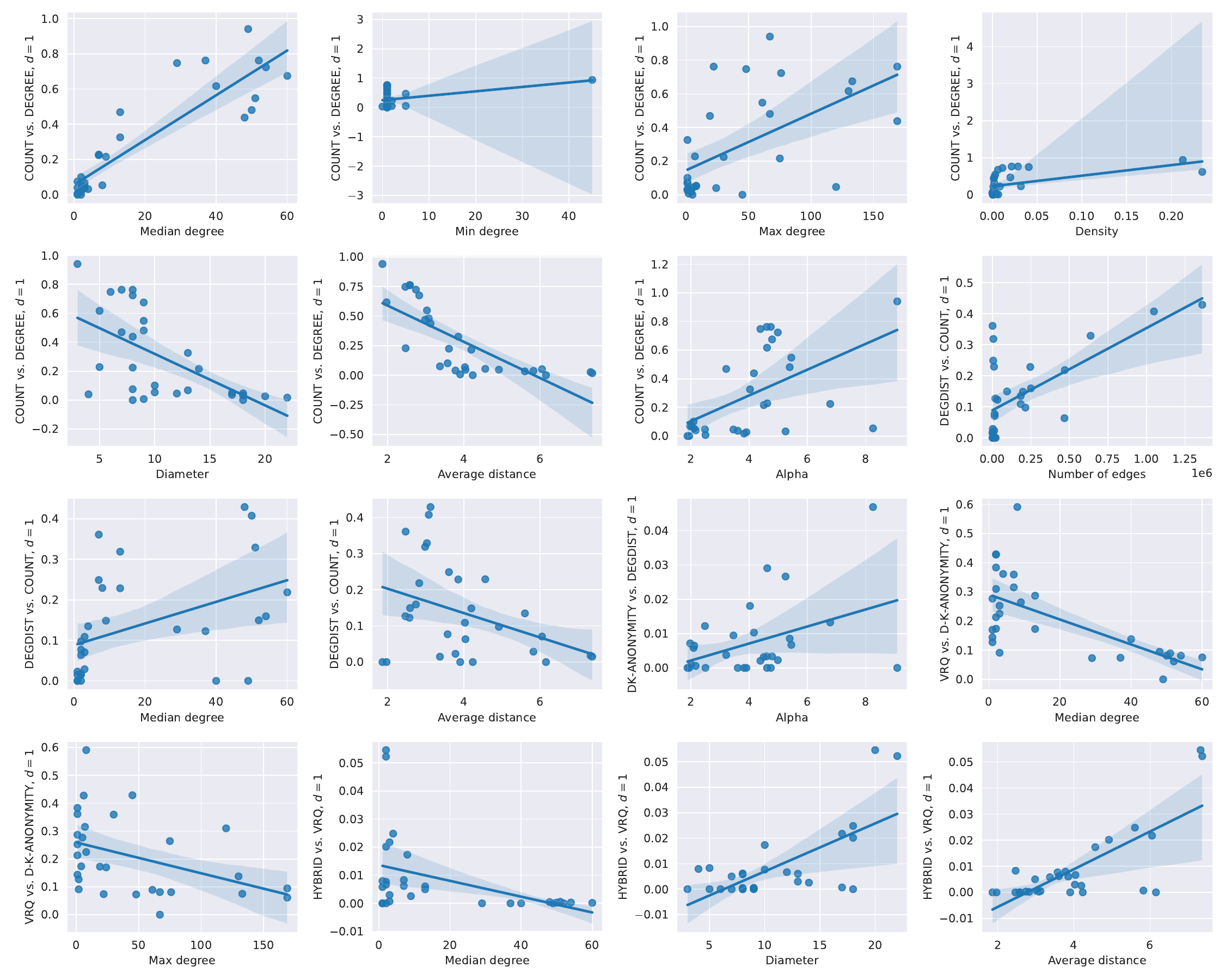}
           \caption{Uniqueness differences and network properties. Each subfigure shows a specific combination of network property (horizontal axis) and difference between measures (vertical axis). Included combinations are significantly correlated as indicated in Table~\ref{tab:corr_diff}.}
           \label{fig:corr_diff}
    \end{figure}

\newpage
\section{Anonymity-cascade}\label{app:cascade}
Figure~\ref{fig:rescascade_final} presents results accompanying Section~\ref{sec:casc}, showing the uniqueness for anonymity-cascade final.

    \begin{figure}[h]
        \centering
        \includegraphics[height=0.85\textheight]{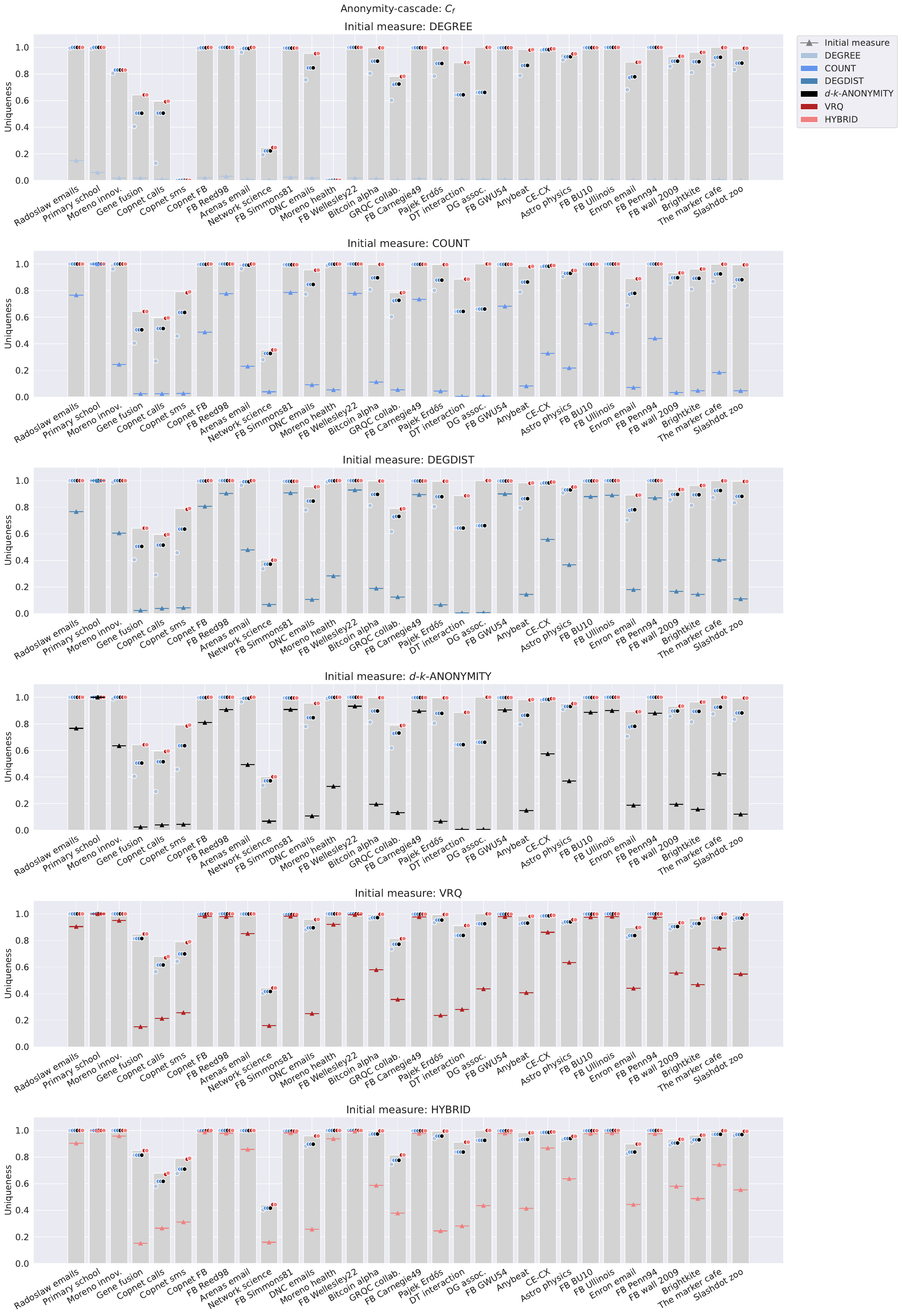}
        \caption{Uniqueness using anonymity-cascade final. Each figure corresponds to a different initial measure (triangle and line) and is combined with all of the cascade measures (dots).
        The grey bar indicates the highest obtained uniqueness for the network given the initial measure. Measures used are: {\textsc{degree}} (\textcolor{lightsteelblue}{lightblue}), \textsc{count} (\textcolor{cornflowerblue}{blue}),  \textsc{degdist} (\textcolor{steelblue}{dark blue}), \textsc{$d$-$k$-anonymity} (black), \textsc{vrq} (\textcolor{firebrick}{red}) and \textsc{hybrid} (\textcolor{lightcoral}{pink}).
        }
        \label{fig:rescascade_final}
    \end{figure}

\newpage
\clearpage
\section{Measures and utility}\label{app:utility}
This appendix accompanies Section~\ref{sub:utility} of the main paper containing the Pareto fronts found for 14 networks.

\begin{figure}[htbp]
\begin{center}       
    {\includegraphics[width=\textwidth]{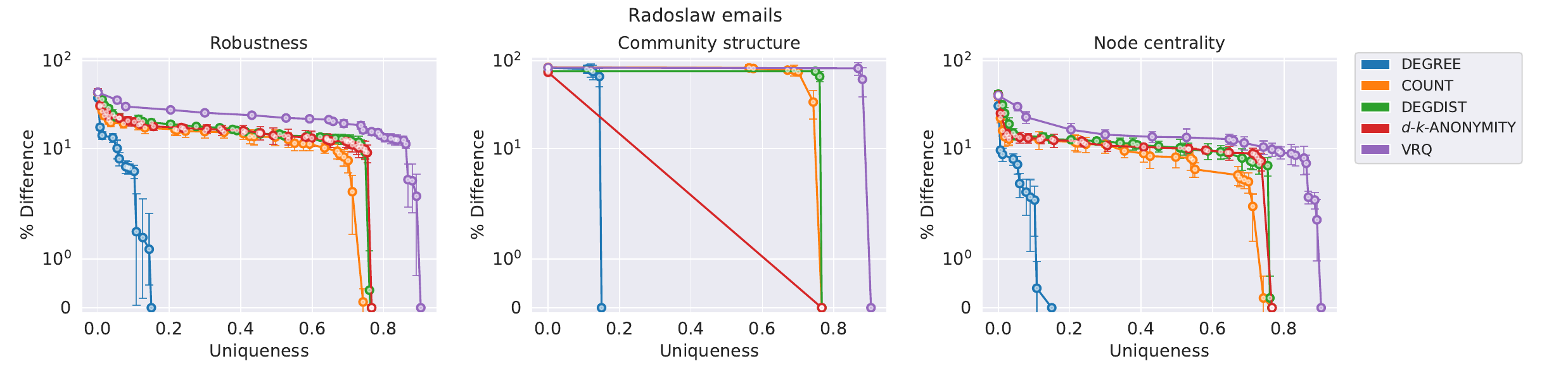}}
    {\includegraphics[width=\textwidth]{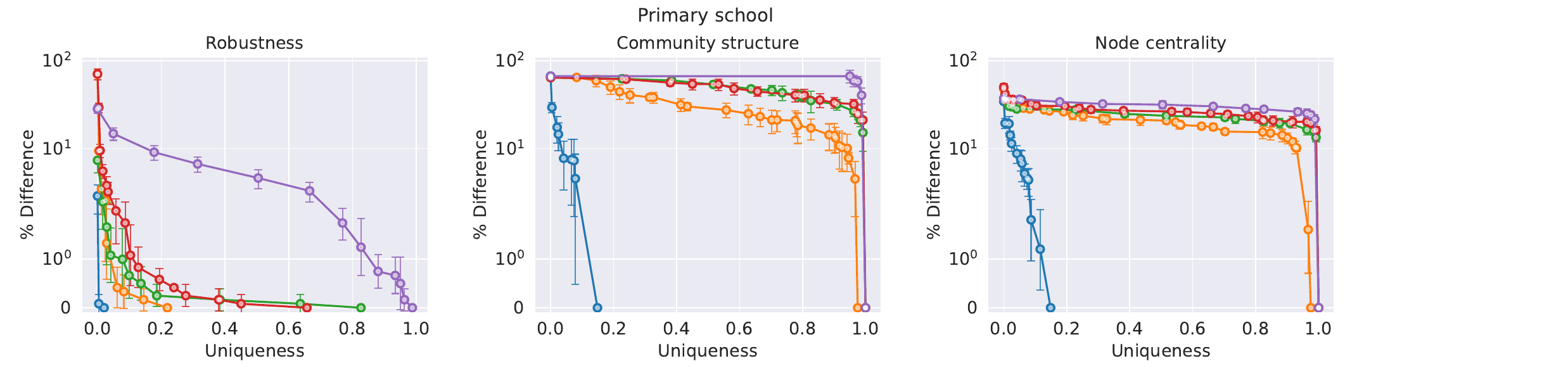}}
    {\includegraphics[width=\textwidth]{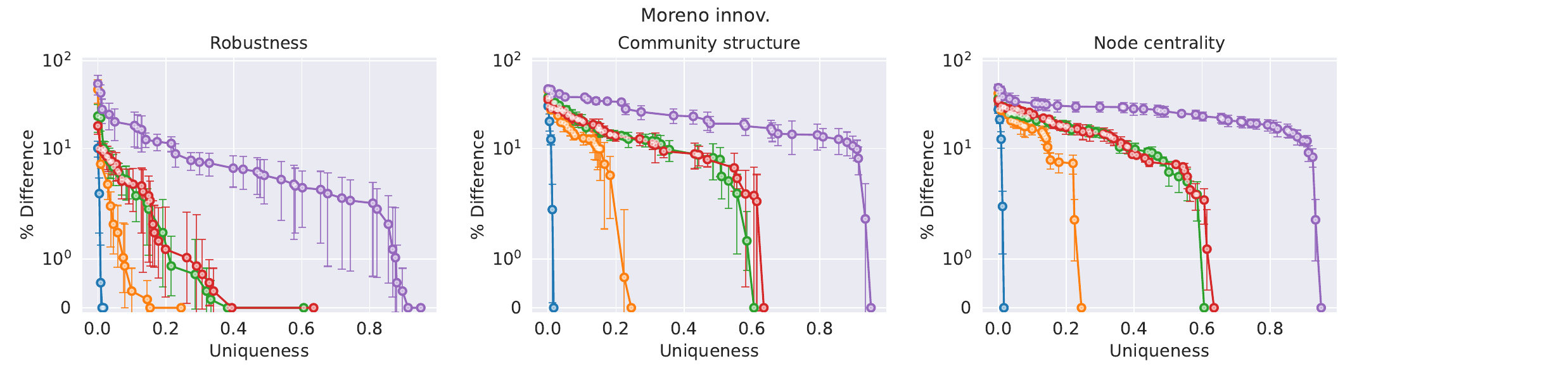}}
    {\includegraphics[width=\textwidth]{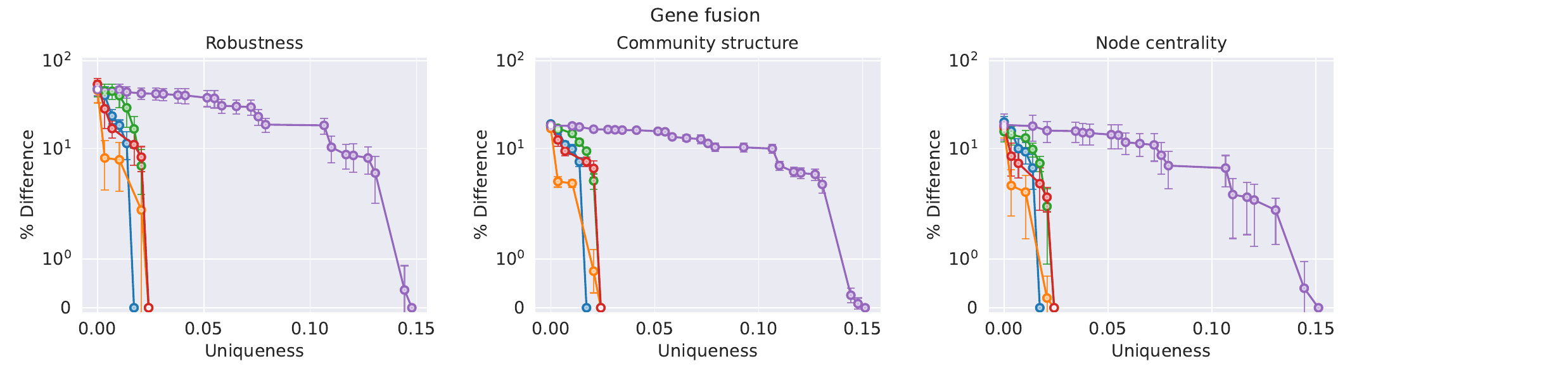}}
\end{center}
\end{figure}
\begin{figure}[htbp]
\begin{center}  
    {\includegraphics[width=\textwidth]{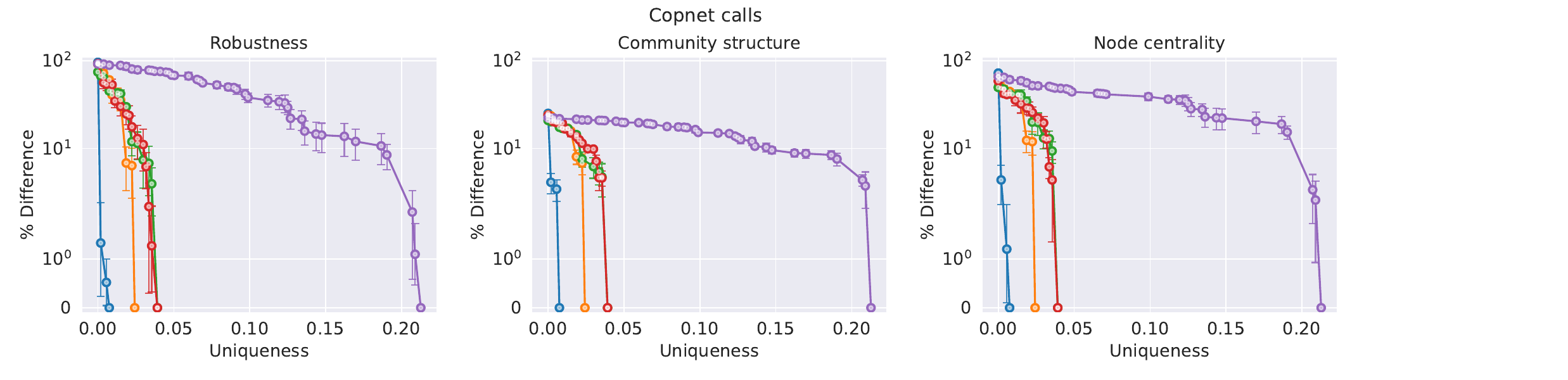}}
    {\includegraphics[width=\textwidth]{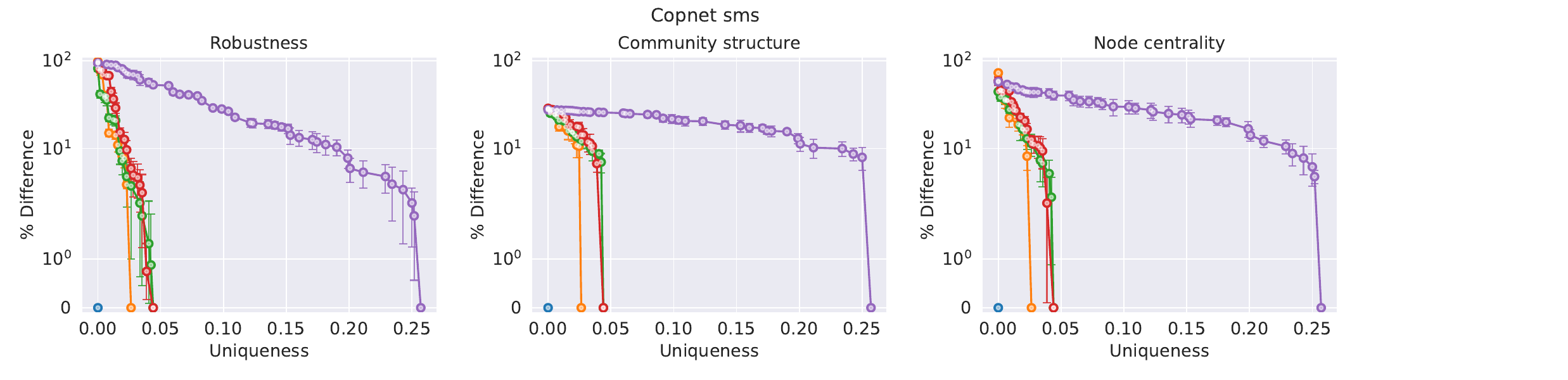}}
    {\includegraphics[width=\textwidth]{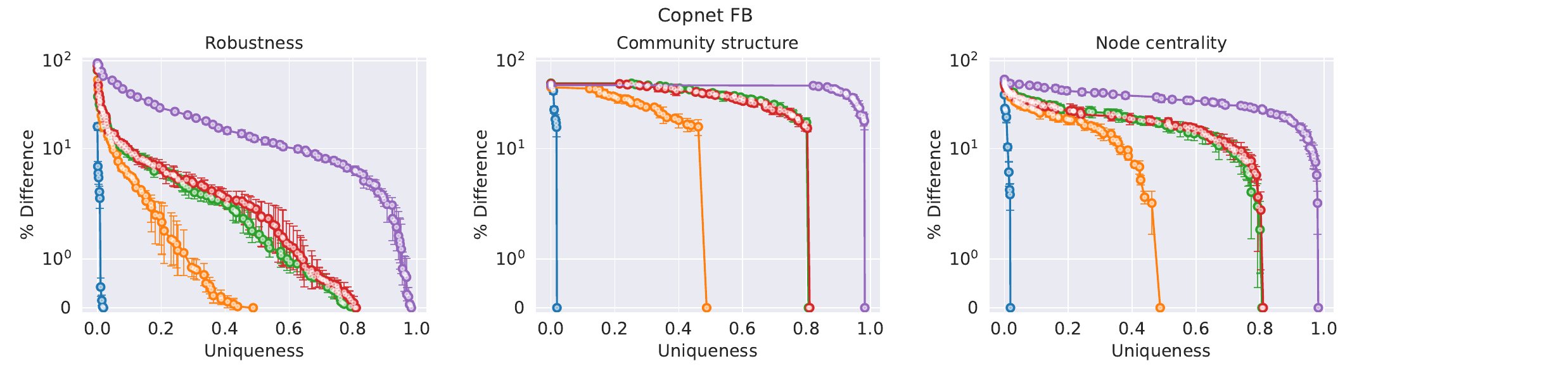}}
    {\includegraphics[width=\textwidth]{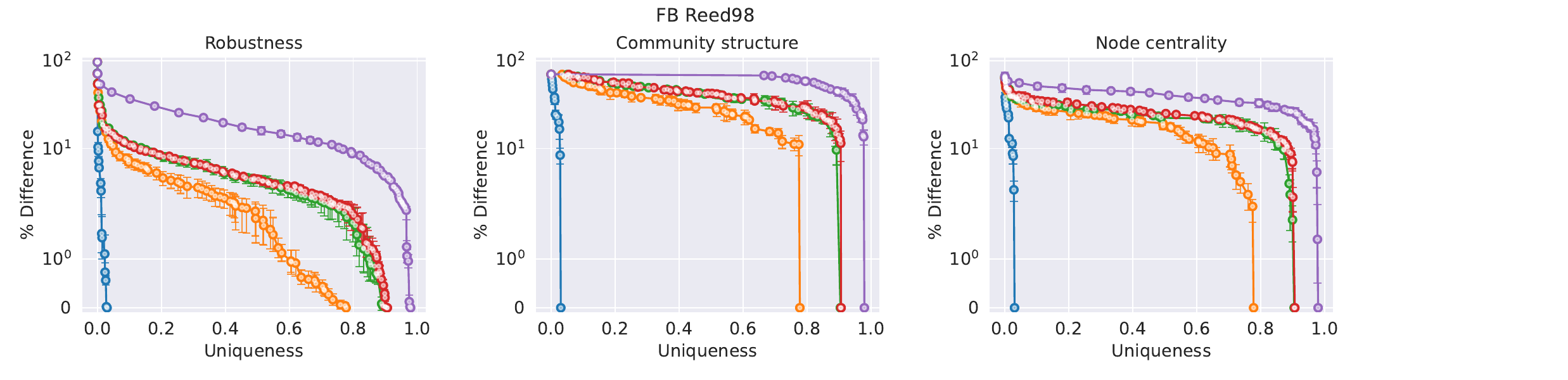}}
    {\includegraphics[width=\textwidth]{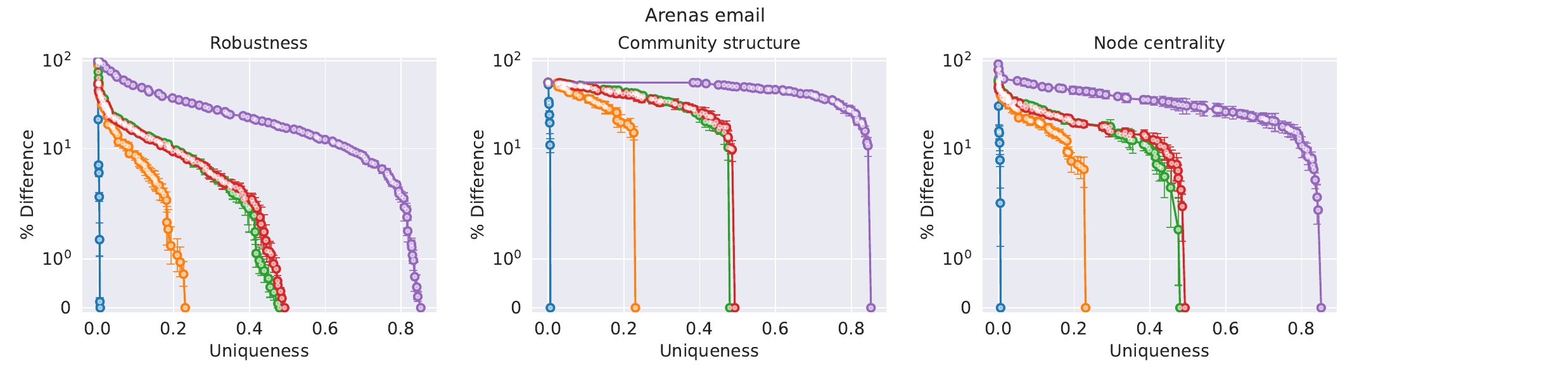}}
    {\includegraphics[width=\textwidth]{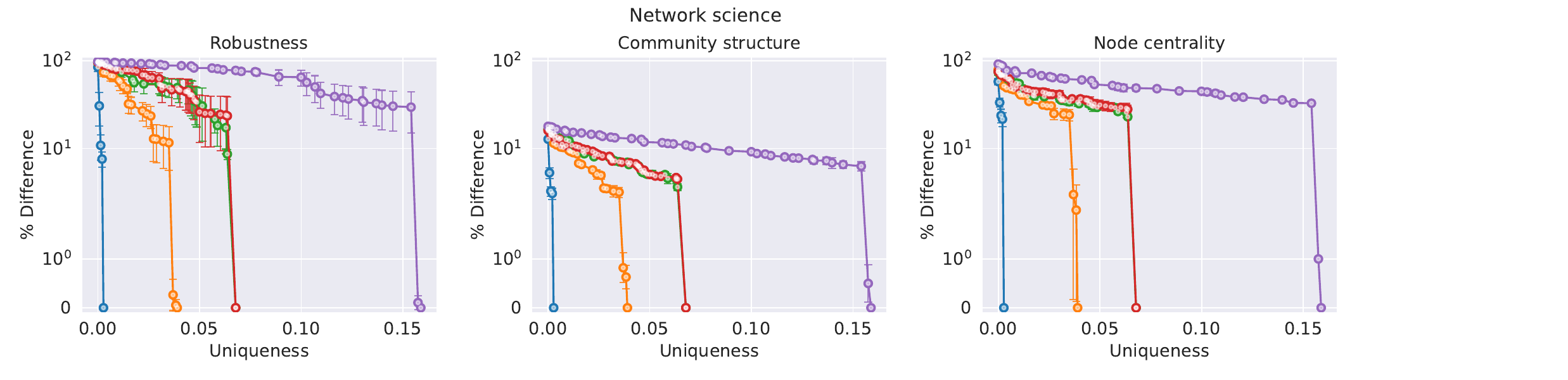}}
\end{center}
\end{figure}
\begin{figure}[htbp]
\begin{center} 
    {\includegraphics[width=\textwidth]{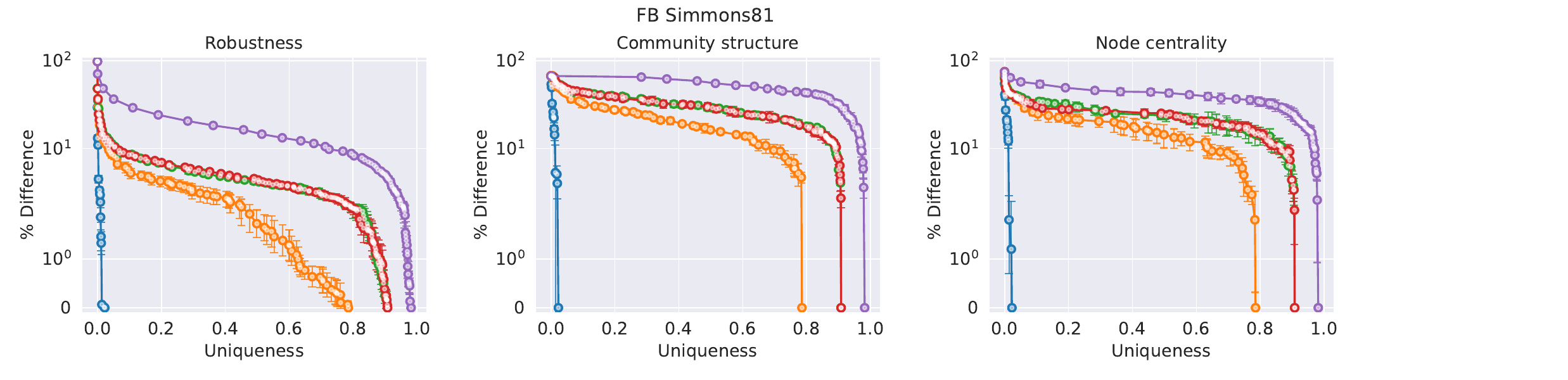}}
    {\includegraphics[width=\textwidth]{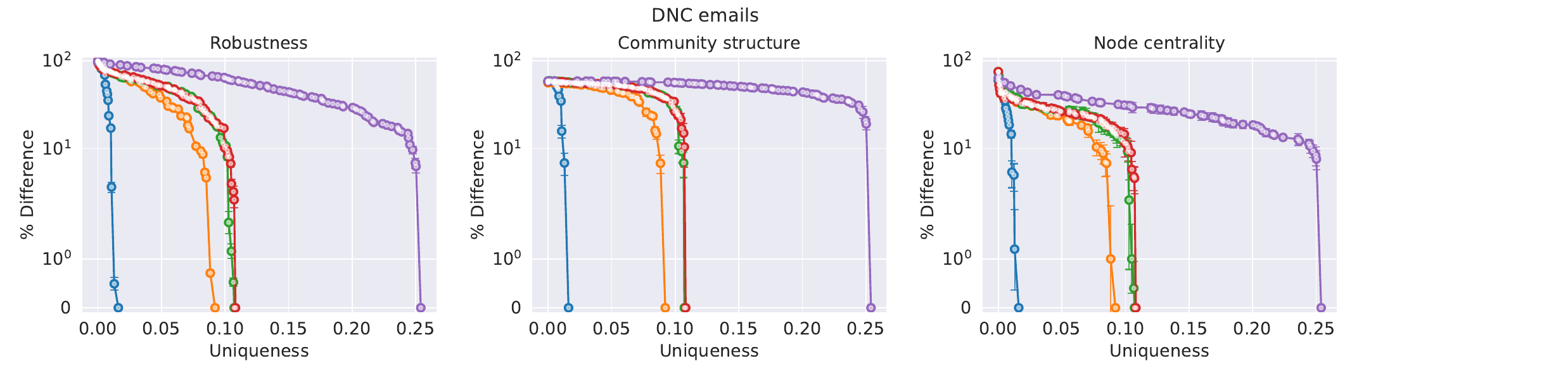}}
    {\includegraphics[width=\textwidth]{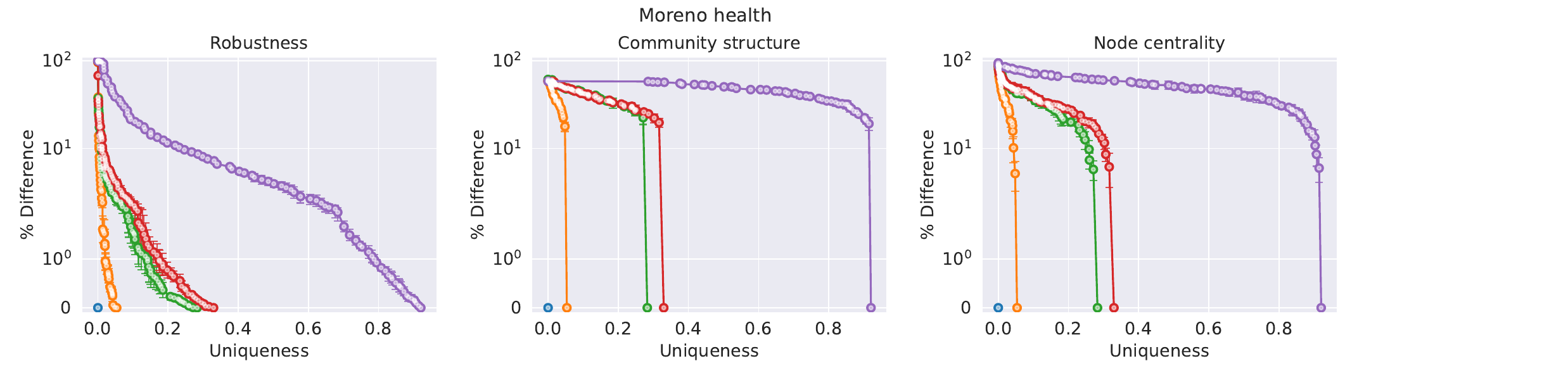}}
    {\includegraphics[width=\textwidth]{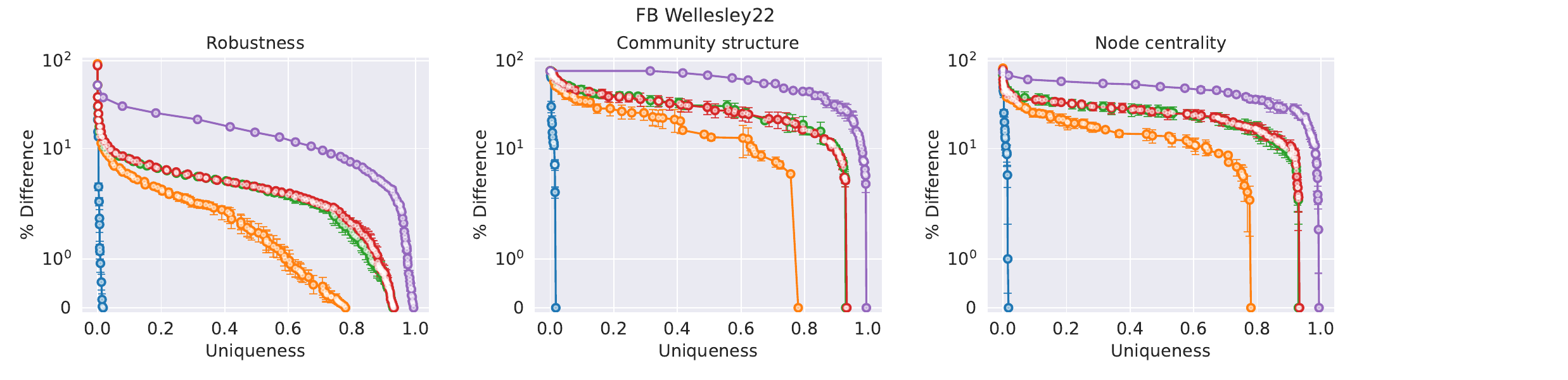}}
    \caption{Pareto optimal solutions found in terms of uniqueness (horizontal axis) and performance on a downstream task (vertical axis, one per column) for 14 networks (rows). Each dot represents a solution found when anonymizing using edge sampling for one of the five anonymity measures (colors). Error bars indicate results $\pm$ one standard deviation. }
	\label{fig:app:util}
\end{center}
\end{figure}

\end{appendices}
\end{document}